\documentclass[11pt]{article}
\usepackage[utf8]{inputenc}
\usepackage{style}
\usepackage{mdframed}
\usepackage{mdwlist}
\usepackage{nicefrac}
\usepackage[parfill]{parskip}
\usepackage{setspace}
\usepackage[colorlinks]{hyperref}
\usepackage[ruled,linesnumbered]{algorithm2e}
\usepackage{amsmath}
\usepackage{amssymb}
\usepackage{complexity}
\newcommand{\ug}{\textsc{UniqueGames}}

\newcommand{\Ex}{\E}
\newcommand{\cP}{{\cal P}}
\newcommand{\tright}{\blacktriangleright}
\allowdisplaybreaks
\newcommand{\sym}{{\rm sym}}
\newcommand{\defeq}{\overset{\rm def}{=}}
\usepackage{microtype}
\newcommand{\vphi}{\varphi}
\newtheorem{problem}[theorem]{Problem}
\usepackage{thm-restate}

\usepackage{boxedminipage}

\newcommand{\paren}[1]{\left(#1 \right )}

\newcommand{\Abs}[1]{\left\lvert#1\right\rvert}

\newcommand{\norm}[1]{\left\lVert#1\right\rVert}

\let\e\varepsilon

\newcommand{\smallsetexpansion}{{\sc Small-set edge-expansion}}
\newcommand{\smallsetvertexexpansion}{{\sc Small-set vertex-expansion}}
\newcommand{\stronguniquegames}{{\sc Strong Unique Games}}
\newcommand{\uniquegames}{{\sc Unique Games}}
\newcommand{\oddcycletransversal}{{\sc Odd cycle transversal}}
\newcommand{\hypersse}{{\sc Hypergraph-SSE}}

\newcommand{\hug}{$d$-{\sc ary Unique Games}}
\newcommand{\sbug}{{\sc Strong Bipartite UG}}
\newcommand{\strug}{{\sc Strong Unique Games}}
\renewcommand{\ug}{{\sc Unique Games}}
\newcommand{\cH}{\mathcal{H}}

\newcommand{\bigO}{\mathcal{O}}
\newcommand{\bigo}[1]{\bigO\left(#1\right)}
\newcommand{\tbigO}{\tilde{\mathcal{O}}}
\newcommand{\tbigo}[1]{\tbigO\left(#1\right)}

\newcommand{\sval}{{\cV}}
\newcommand{\Prob}[1]{\Pr\left[#1\right]}
\newcommand{\cE}{{\mathcal{E}}}
\newcommand{\sdp}{{\sf SDP}}
\newcommand{\wh}[1]{\widehat{#1}}

\title{Approximation Algorithms and Hardness for \\ Strong Unique Games}
\author{Suprovat Ghoshal\\Indian Institute of Science\\Bangalore, India \\suprovat@iisc.ac.in \and
	Anand Louis
	\\Indian Institute of Science\\Bangalore, India \\anandl@iisc.ac.in }

\date{}
\begin{document}

\begin{titlepage}
\maketitle

\begin{abstract}
The \uniquegames~problem is a central problem in algorithms and complexity theory.
Given an instance of \uniquegames, 
the \stronguniquegames~problem asks to find the largest subset of vertices,
such that the \uniquegames~instance induced on them is completely satisfiable.
While the  \uniquegames~problem has been studied extensively, to the best of our knowledge, there hasn't been much work studying the approximability of the \stronguniquegames~problem. 
Given an instance with label set size $k$ where a set of $1 - \e$ fraction of the vertices
induce an instance that is completely satisfiable, our first algorithm produces
a set of $1 - \tbigo{k^2} \e \sqrt{\log n}$ fraction of the vertices such that the 
\uniquegames~induced on them is completely satisfiable.
In the same setting, our second algorithm produces a set of $1 - \tbigo{k^2} \sqrt{\e \log d}$
(here $d$ is the largest vertex degree of the graph) fraction of the vertices 
such that the \uniquegames~induced on them is completely satisfiable.
The technical core of our results is a new connection between \stronguniquegames~and
{\em small-set vertex-expansion} in graphs.
Complementing this, assuming the Unique Games Conjecture, we prove that there exists an absolute constant $C$
such that it is NP-hard to compute a set of size larger than $1 - C \sqrt{\e \log k \log d}$ 
such that all the constraints induced on this set are satisfied. 

For the \uniquegames~problem, given an instance that has as assignment satisfying $1 - \epsilon$
fraction of the constraints, there is a polynomial time algorithm 
[Charikar, Makarychev, Makarychev - STOC 2006] that computes an assignment satisfying 
$1 - \bigo{\sqrt{\epsilon \log k}}$ fraction of the constraints; [Khot et al. - FOCS 2004] 
prove a matching (up to constant factors) Unique Games hardness. 
Therefore, our hardness results suggest that the \stronguniquegames~problem might be harder
to approximate than the \uniquegames~problem.

Given an undirected graph $G(V,E)$ the \oddcycletransversal~problem,
asks to delete the least fraction of vertices to make the induced graph on the 
remaining vertices bipartite. 
As a corollary to our main algorithmic results, we obtain an algorithm that 
outputs a set $S$ such the graph induced on $V \setminus S$ is bipartite, and 
$\Abs{S}/n \leq \bigo{\sqrt{\epsilon \log d}}$
(here $d$ is the largest vertex degree and $\epsilon$ is the optimal fraction of
vertices that need to be deleted).
Assuming the Unique Games conjecture, we prove a matching (up to constant factors) hardness.

\end{abstract}

\end{titlepage}

\tableofcontents

\newpage

\section{Introduction}

The \uniquegames~problem \cite{Khot02a} is a central problem in computational complexity theory.
Formally, the problem is the following.
\begin{problem}[\uniquegames]				
Given an instance of \uniquegames~$\cG(V_{\cG},E_{\cG},[k],\{\pi_e\}_{e \in E_{\cG}})$,
consisting of a graph $(V_{\cG},E_{\cG})$, alphabet $[k]$, and bijections $\pi_{(u,v)}: [k] \to [k]$ for each
$(u,v) \in E_{\cG}$, the goal is to compute an assignment $\sigma: V_{\cG} \to [k]$ which maximizes the 
number of satisfied constraints $\Abs{\set{\set{u,v} \in E_{\cG}: \pi_{(u,v)}(\sigma(u)) = \sigma(v)}}$.
The value of this game, denoted by $\val(\cG)$, is the fraction of constraints satisfied by this optimal assignment.
\end{problem}
We say that an instance is {\em satisfiable} if its value is $1$. 
A natural algorithmic question regarding the \uniquegames~problem is that given an
instance that is ``almost satisfiable'', can one compute an assignment that ``almost satisfies'' 
the instance? 
There are multiple ways to formalize the notion of an instance being ``almost satisfiable''. One
way to quantify this is using the value of the instance: the larger the value,
the closer it is to being satisfiable. 
Khot's {\em Unique Games Conjecture} \cite{Khot02a}, conjectures that even if an instance 
has value close to $1$, it is NP-hard to find an assignment that satisfies even a few constraints.
More formally, the conjecture says the following.
\begin{conjecture}[Unique Games Conjecture~\cite{Khot02a}]				
\label{conj:ug}
	There exists $\epsilon_0 \in (0,1)$ such that for every choice of $0 < \epsilon_c,\epsilon_s \leq \epsilon_0$,
there exists a $k \in \N$ such that 
	the following holds. Given a \ug~instance $\cG(V,E,[k],\{\pi_{e}\}_{e \in E})$, it is \NP-Hard to distinguish between the following cases:
	\begin{itemize}
		\item [YES:] $\val{\cG} \geq 1 - \epsilon_c$.
		\item [NO:]  $\val(\cG) \leq \epsilon_s$.
	\end{itemize}
	Furthermore, the underlying constraint graph of $\cG$ is regular.
\end{conjecture}

This conjecture, a central topic of study in the last two decades, 
implies optimal hardness of approximation for many problems \cite{KR08,KKMO07,Rag08} etc.

Another way to quantify the notion of a \uniquegames~instance being close to satisfiable 
is by the relative size of the largest set of vertices that induce an instance that is 
fully satisfiable. The \stronguniquegames~problem studies this notion and is formally 
defined as follows. 
\begin{problem}[\stronguniquegames]		\label{prop:str-ug}
Given an instance of \stronguniquegames~instance $\cG(V_{\cG},E_{\cG},[k],\{\pi_e\}_{e \in E_{\cG}})$,
consisting of a graph $(V_{\cG},E_{\cG})$, alphabet $[k]$, and bijections $\pi_e: [k] \to [k]$ for each
$e \in E_{\cG}$, the goal is to compute the largest cardinality set $S \subset V_{\cG}$ such that the
value of the \uniquegames~instance induced on $S$ is $1$. 
The value of this game, denoted by $\sval(\cG)$, is defined as $\Abs{S}/\Abs{V_{\cG}}$ where $S$ is an optimal set.
\end{problem}

While the \uniquegames~problem is an {\em edge deletion} problem 
(delete the smallest number of edges to obtain an instance that is fully satisfiable),
the \stronguniquegames~problem is a {\em vertex deletion} problem
(delete the smallest number of vertices to obtain an instance that is fully satisfiable).
This allows it to express well studied combinatorial optimization problems such as the 
\oddcycletransversal~as special cases.

The \stronguniquegames~problem has been studied in several different contexts.  
Khot and Regev \cite{KR08} gave a reduction from \uniquegames~to \stronguniquegames, and used it
to prove optimal Unique Games hardness for the minimum vertex cover problem. More recently, Bhangale and Khot
\cite{BK19} showed that the recent sequence of works on the $2$-to-$2$ games~\cite{DBLP:conf/stoc/DinurKKMS18,DBLP:conf/focs/KhotMS18} implies that a variant of the \stronguniquegames~problem is NP-Hard 
in certain parameter regimes (see Section \ref{sec:related}),
and used that to obtain improved inapproximability results for problems such as {\sc Max Acyclic Subgraph}, {\sc Max Independent Set}, etc.

While there has been a lot of work on developing approximation algorithms for the \uniquegames~problem
(see Section \ref{sec:related} for a brief survey), to the best of our knowledge, there hasn't 
been much work on approximation algorithms for the \stronguniquegames~problem.
Our main contributions are approximation algorithms and Unique Games hardness results 
for the \stronguniquegames~problem.

\subsection{Our results}
Our main algorithmic results are the following.
\begin{theorem}				
\label{thm:partial-ug-1}
There exists a polynomial time randomized algorithm which takes as input an instance 
of \stronguniquegames~$\cG(V_{\cG},E_{\cG},[k],\{\pi_e\}_{e \in E})$ with $\sval(\cG) = 1 -\e$,
and outputs a set $S' \subseteq V_{\cG} $ such that value of the \uniquegames~instance induced 
in $S'$ is $1$ and\footnote{$\tbigo{k}$ denotes $\bigo{k\ {\sf polylog}(k)}$} 
 $\Abs{S'}/n \geq \paren{1 - \tbigo{k^2} \e \sqrt{\log n}}$. 
\end{theorem}

\begin{theorem}				
\label{thm:partial-ug-2}
There exists a polynomial time randomized algorithm which takes as input an instance 
of \stronguniquegames~$\cG(V_{\cG},E_{cG},[k],\{\pi_e\}_{e \in E})$ with $\sval(\cG) = 1 -\e$ and the
graph $(V_{\cG},E_{\cG})$ having maximum vertex degree $d$,
and outputs a set $S' \subseteq V_{\cG} $ such that value of the \uniquegames~instance induced 
in $S'$ is $1$ and $\Abs{S'}/n \geq \paren{1 - \tbigo{k^2} \sqrt{\e \log d}}$. 
\end{theorem}

We also give an algorithm that gives slightly improved guarantees (compared to Theorem \ref{thm:partial-ug-1}), for which the running time has polynomial dependence on $1/\epsilon$, as stated in the following theorem.

\begin{theorem}				
	\label{thm:partial-ug-3}
	There exists a  randomized algorithm which takes as input an instance 
	of \stronguniquegames~$\cG(V_{\cG},E_{cG},[k],\{\pi_e\}_{e \in E})$ with $\sval(\cG) = 1 -\e$ and the
	graph $(V_{\cG},E_{\cG})$ having maximum vertex degree $d$,
	and outputs a set $S' \subseteq V_{\cG} $ such that value of the \uniquegames~instance induced 
	in $S'$ is $1$ and $\Abs{S'}/n \geq \paren{1 - \tbigo{k^2} \e\sqrt{ \log d \e n}}$. The algorithm runs in time ${\rm poly}(n,k,1/\epsilon)$.
\end{theorem}
 
Raghavendra and Steurer \cite{RS10} (see also \cite{RST12}) showed a close connection between the 
\uniquegames~problem and small-set {\em edge}-expansion in graphs. 
The technical core of our results is a new connection between \stronguniquegames~and
small-set {\em vertex}-expansion (see Section \ref{sec:prelim} for the definitions of
these expansion quantities).

Complementing these algorithmic results, we proving the following hardness for \stronguniquegames.
\begin{restatable}{rethm}{hardness}		
	\label{thm:strong-ug-informal}
	Fix $\epsilon \in (0,1)$ and let  $k \geq k(\epsilon)$ and $d \geq \left[C \epsilon^{-2} \log k, d(\epsilon)\right]$. 
	Assuming the Unique Games Conjecture (Conjecture \ref{conj:ug}), the following holds. 
	Given a \stronguniquegames~instance $\cG(V_{\cG},E_{\cG},[k],\{\pi_e\}_{e \in E_{\cG}})$
	such that the underlying graph $(V_{\cG},E_{\cG})$ has maximum degree $d$,
	it is \NP-Hard to distinguish between the following cases. 
	\begin{itemize}
		\item {\bf YES}: $\sval(\cG) \geq 1 - \epsilon$.
		\item {\bf NO}: $\sval(\cG) < \min\left(1 - C'\sqrt{\epsilon \log d \log k},0.01\right)$.
	\end{itemize}
	where $C, C' > 0$ are absolute constants and $k(\epsilon),d(\epsilon)$ are constants depending only on $\epsilon$.
\end{restatable}

Analogous to this, 
Khot et al. \cite{KKMO07} showed that assuming the Unique Games conjecture, 
it is NP-hard to distinguish between the following cases of an instance $\cG$ of 
\uniquegames~having alphabet size $k$ (i) $\val(\cG) \geq 1 - \epsilon$, and, (ii)
$\val(\cG) \leq 1 - \bigo{\sqrt{\epsilon \log k}}$.
Charikar et al. \cite{CMM06a} gave a polynomial time algorithm that takes as input 
an instance $\cG$ of \uniquegames~having alphabet size $k$ and $\val(\cG) \geq 1 - \epsilon$,
and outputs an assignment satisfying $1 - \bigo{\sqrt{\epsilon \log k}}$ fraction of the 
constraints.
Therefore, the Unique Games hardness of the \uniquegames~problem had been
settled (up to constant factors). 
To the best of our knowledge, our work (Theorem \ref{thm:strong-ug-informal}) is the first to suggest that the 
\stronguniquegames~problem might be strictly harder to approximate than the \uniquegames~problem.

\subsubsection*{Odd cycle transversal}
Given an undirected graph $G=(V,E)$ the \oddcycletransversal~problem,
an NP-hard problem, 
asks to delete the least fraction of vertices to make the induced graph on the 
remaining vertices bipartite. Since this problem is a special case of the 
\stronguniquegames~problem, we get the following result.
\begin{corollary}[Corollary to Theorem \ref{thm:partial-ug-2}]
\label{cor:oct}
There exists a polynomial time randomized algorithm which takes as input 
a $G= (V,E)$ having maximum vertex degree $d$, and outputs a set $S$ such
the graph induced on $V \setminus S$ is bipartite, and 
	$\Abs{S}/n \leq \bigo{\sqrt{\opt \log d}}$.
Here $\opt$ denotes the fraction of vertices in an optimal set of vertices for the
\oddcycletransversal~problem on $G$.
\end{corollary}
Agarwal et al. \cite{ACMM05} gave an $\bigo{\sqrt{\log n}}$ approximation algorithm 
for the \oddcycletransversal~problem.
The approximation guarantee in Corollary \ref{cor:oct} is better than this 
when $\opt$ is ``large'' and $d$ is ``small''. 
Furthermore, the $\bigo{\sqrt{\log n}}$-approximation guarantee is only meaningful when $\opt = \bigo{1/\sqrt{\log n}}$,
whereas Corollary \ref{cor:oct} gives a non-trivial guarantee when $\opt = \bigo{1/\log d}$,
the latter range can be much larger than the former range for ``small'' values of $d$.

Complementing Corollary \ref{cor:oct}, we give matching Unique Games hardness (up to constant factors).
\begin{theorem}
\label{thm:oct-hardness}
Let $\epsilon \in (0,1)$ and $d \geq C \epsilon^{-2}$ be fixed constants. 
Assuming the Unique Games Conjecture (Conjecture \ref{conj:ug}), the following holds. 
Given a graph $G = (V,E)$ having maximum degree $d$, it is \NP-Hard to distinguish between
the following cases. 
	\begin{itemize}
		\item {\bf YES}: $\opt \leq \epsilon$.
		\item {\bf NO}: $\opt \geq C'\sqrt{\epsilon \log d}$.
	\end{itemize}
	where $\opt$ denotes the optimal fraction of vertices for the \oddcycletransversal~problem 
	and $C, C' > 0$ are absolute constants. 

\end{theorem}
We note that Theorem \ref{thm:oct-hardness} doesn't follow as a direct corollary to
Theorem \ref{thm:strong-ug-informal} and requires some additional ideas.

The edge-deletion version of the \oddcycletransversal~problem asks to delete
the smallest fraction of edges to make the remaining graph bipartite. 
The seminar work of Goemans and Williamson \cite{GW94} gave an algorithm that 
deletes $\bigo{\sqrt{\opt}}$ fraction of edges, where $\opt$ denotes the optimal 
fraction of edges to be deleted. 
Khot et al. \cite{KKMO07} gave a matching (up to constant factors) Unique Games 
hardness for this problem. 
Theorem \ref{thm:oct-hardness} suggests that the \oddcycletransversal~problem might
be harder to approximate than its edge-deletion counterpart.

\subsection{Related Work}
\label{sec:related}
\paragraph{Unique Games.}
For a \uniquegames~instance of value $1 - \e$, Charikar et al. \cite{CMM06a} gave a
polynomial time algorithm that outputs an assignment satisfying $1 - \bigo{\sqrt{\e \log k}}$
fraction of the constraints. Chlamtac et al. \cite{CMM06} gave a polynomial time algorithm
that outputs an assignment satisfying $1 - \bigo{\e \sqrt{\log n \log k}}$ fraction of the 
constraints.

\paragraph{Small-set expansion.}
Raghavendra and Steurer \cite{RS10} (see also \cite{RST12}) showed a close connection between the 
\uniquegames~problem and \smallsetexpansion~in graphs. 
\cite{RST10a} and \cite{BFKM11} gave bi-criteria approximation algorithms for small set edge-expansion in graphs. 
\cite{lm16} gave approximation algorithms for small-set vertex-expansion in graphs and 
small-set expansion in hypergraphs; the core of their results was the construction of 
hypergraph orthogonal separators which we describe in detail in Theorem \ref{thm:hos-1} and Theorem \ref{thm:hos-2}. Feige et al.~\cite{FLH08} study the related problem of finding the minimum weight vertex separators in graphs for which they give a $\sqrt{\log {\rm opt}}$-approximation.

\paragraph{Hardness of approximation.}
Khot and Regev \cite{KR08} gave a reduction from \uniquegames~to \stronguniquegames,
and proved that Conjecture \ref{conj:ug} implies an analogous conjecture for the 
\stronguniquegames~problem; 
this also proved a constant factor inapproximability of \stronguniquegames~assuming the Unique Games Conjecture. 
We note that their reduction was not a factor preserving reduction.
Bhangale and Khot \cite{BK19} proved that for every $\epsilon > 0$, there exists a $k$ such
that the following holds: given a \uniquegames~instance of $(V,E,[k],\{\pi_e\}_{e \in E})$ 
where the underlying graph $(V,E)$ is a bipartite graph with $A,B \subset V$ as the two color
classes and is regular on the $A$ side, it is NP-hard to distinguish between the following two cases
(i) there exists $A' \subset A$ of size at least $(1/2 - \epsilon)|A|$ and an assignment
that satisfies all the edges incident on $A'$, and 
(ii) every assignment satisfies at most $\epsilon$ fraction of edges. 

There has been a extensive work on proving hardness of approximation results assuming
the Unique Games Conjecture, we describe some of the most relevant results here. 
Khot et al. \cite{KKMO07} showed that the assuming the Unique Games conjecture, 
it is NP-hard to distinguish between the following cases of an instance $\cG$ of 
\uniquegames~having alphabet size $k$ (i) $\val(\cG) \geq 1 - \epsilon$, and,  (ii)
$\val(\cG) \leq 1 - \bigo{\sqrt{\epsilon \log k}}$.

\cite{LRV13} proved that for the problem of computing the vertex expansion of graphs, 
it is SSE-hard (see \cite{RS10} for definition of the small-set expansion hypothesis)
to obtain an approximation guarantee better than $\bigo{\sqrt{OPT \log d}}$ on graphs of 
maximum degree $d$; their work also gave a matching (up to constant factors) approximation algorithm.

\paragraph{Partial Problems.}
"Strong" or "partial" versions of some other problems have been studied in the literature.
A graph $G(V,E)$ is said to be $\alpha$-partially $k$-colorable if there exists a set 
$S \subset V$ such that $\Abs{S} \geq \alpha \Abs{V}$ and the graph induced on $S$
is $k$-colorable. 
For $\alpha$-partially $3$-colorable graphs, approximation algorithms have been studied
in the context of pseudorandom instances \cite{kumar2018finding}, and worst case
and semi-random instances \cite{GLS19}. 
Agarwal et al. \cite{ACMM05} gave a $\bigo{\sqrt{\log n}}$ approximation algorithm for the 
\oddcycletransversal~problem.
This problem is fixed parameter tractable when parameterized by the number of bad vertices
\cite{RSV04,NRRS12}. Khot and Bansal~\cite{BK09} showed that assuming 
the Unique Games conjecture,
\oddcycletransversal~is NP-Hard to approximate to any constant factor.

\section{Overview}

\subsection{Approximation Algorithms for \stronguniquegames}
The first observation is that unlike the traditional \uniquegames~problem, \stronguniquegames~cannot be formulated as a Max-CSP, since the global constraint that the induced game on the set output by the algorithm has to be completely satisfiable. Consequently, off the shelf algorithms (and in particular, Raghavendra's optimal SDP+Rounding framework~\cite{Rag08,RS09}) known for solving such CSPs do not apply here. While \uniquegames~formulated as a Max-CSP does not shed much light in this setting, it's connection with graph partitioning problems turns out to be the main insight here. In particular, it shares tight connections with the \smallsetexpansion~problem; we will describe it in some detail here since this motivates the key ideas behind our approach. We begin by recalling the definition of the label extended graph.
\begin{definition}[Label extended graph]				
\label{defn:lab-gr}
Given a \uniquegames~instance $\mathcal{G}(V_{\cG},E_\cG,[k],\{\pi_{e,v}\}_{e \in E})$, 
the corresponding label extended graph $G= (V,E)$ is constructed as follows. The vertex set 
is given by $V \defeq \cup_{v \in V}\cC_v$ where $\cC_v \defeq \{v\} \times [k]$ is the 
cloud of vertices associated with $v \in V_\cG$. Furthermore, for every edge constraint $(u,v) \in E$, we add an edge between vertices $(u,i)$ and $(v,j)$ if $\pi_{(u,v)}(i) = j$. 
\end{definition}
Given a labeling $\sigma:V_\cG \to [k]$, one can naturally encode it as the following subset $S_\sigma := \set{ (a,\sigma(a))): a \in V_{\cG}}$. Furthermore, the goodness of a labeling for $\cG$ translates nicely into structural conditions for the corresponding subsets in the following way. Given a graph $G' = (V',E')$, we define the {\em edge expansion} of a set $S \subset V$ as 
\[ \phi^{E'}_{G'}(S) \defeq \frac{|\partial^{E'}_{G'}(S)|}{|S||V' \setminus S|}, \] 
where $\partial^{E'}_{G'}(S)$ is the set of edges in $E'$ with exactly one vertex in $S$. The following connection between satisfiability of assignments and expansion of sets in the label extended graph is well known.
\begin{proposition}[\cite{RS10}]			
\label{prop-sse}
	An assignment $\sigma:V_{\cG} \to [k]$ satisfies at least $(1 - \e)$-fraction of edges if and only if $\phi^E_G(S_\sigma) \leq \e$.
\end{proposition}
In other words, good labelings for $\cG$ translate to non-expanding sets in the label extended graph $G$. This connection has been instrumental in the development in connecting \uniquegames~to \smallsetexpansion, which has eventually lead to the best known integrality gaps \cite{KV05,RS09a} and hardness of approximation results for several graph partitioning problems \cite{RST12,LRV13}. It is therefore natural to ask if analogous connections can be made for the \stronguniquegames~problem. The key to this lies in first understanding how the above connections fails in this context. Informally, given a labeling $\sigma:V_\cG \to [k]$, the edge boundary of the set $\partial^E_{\cG}(S_\sigma)$ exactly counts the number of edge constraints violated by $\sigma$, which does not say how many vertices the violated edges are incident on. However, we can choose to count the vertices on which the violated edges are incident on. In the language of the expansions, this translates to looking at the {\em vertex expansion} of the set which is defined as follows. For any subset $S \subset V$, we define its vertex expansion to be 
\[ \phi^V_G(S)\defeq \frac{|\partial^V_G(S)|}{|S|}, \]
where $\partial^V(S)$ is the set of vertices in $V \setminus S$ whose neighborhood intersects with $S$. As in Proposition \ref{prop-sse}, it turns out that a quantitatively similar statement connecting \stronguniquegames~and the vertex expansion of the label extended graph can also be shown, as stated by the following proposition.
\begin{proposition}			\label{prop:ssve-ug}
Let $\cG$ be a \stronguniquegames~instance, and let $G = (V,E)$ be its label extended graph. Then the following conditions hold:
\begin{itemize}
	\item If $\sval(G) \geq 1 - \e$, then there exists a subset $S \subset V$ of size 
		at least $(1 - \epsilon)|V_{\cG}|$ such that $\Abs{\partial^V_G(S)} \leq \epsilon n k$.
	\item If there exists a non-repeating subset $S \subset V$ of size at least $(1 - \delta)|V_{\cG}|$ 
		such that $\Abs{\partial^V_G(S)} \leq \epsilon n$, then $\sval(G) \geq 1 - \e - \delta$.
\end{itemize}
Here a subset $S$ is non-repeating if $|S \cap \cC_v| \leq 1$ for every vertex $v \in V_\cG$.
\end{proposition}
Equipped with the above intuition, one can rephrase the \stronguniquegames~problem as follows: {\em Can we efficiently find a set $S \subset V$ of size $n(1-\epsilon)$ (where $n = |V_\cG|$) which corresponds to a valid labeling (i.e., is non-repeating) and has small vertex expansion?} If one chooses to ignore the requirement of the set $S$ encoding a valid partial assignment, then this is precisely the \smallsetvertexexpansion~problem. As one would expect, the exact version of this an $\NP$-Hard problem, but we know of efficient approximation algorithms for this \cite{lm16}. In particular, it admits an approximation preserving reduction to the Hypergraph Small Set Expansion (HSSE in short) problem \cite{lm16}: given a hypergraph $H = (V',E')$ and a parameter $\delta \in (0,1/2]$, compute a set $S \subset V'$ having size at most $\delta |V'|$ which has the least value of $|\partial^V_{H'}(S)|/|S|$ ($\partial^V_{H'}(S)$
is the set of edges in $E'$ which have at least one vertex from $S$ and at least one vertex from $V \setminus S$). 

\paragraph{Approximation Algorithm for HSSE.}
In order to motivate our final algorithm a short description of Louis and Makarychev's approximation algorithm for HSSE \cite{lm16} is in order. 
For a hypergraph $H' = (V',E')$, an SDP relaxation for $\delta$-\hypersse~is the following: 
\begin{eqnarray*}
	\text{minimize}  &  \sum_{e \in E'} \max_{u,v \in e}\| x_u - x_v\|^2 & \\
	\text{subject\space to}& \sum_{v \in V'}\langle x_u, x_v \rangle  \leq  \delta|V'|\|x_u\|^2 & \forall u \in V' \\
						   & \sum_{v \in V'}\|x_v\|^2  = 1 \\
	& \|x_u - x_v\|^2 +\|x_u - x_w\|^2 \geq \|x_w - x_v\|^2 & \forall u,v,w \in V' \\
	& 0 \le \langle x_u , x_v \rangle \leq \|x_u\|^2 &\forall u,v \in  V' 
\end{eqnarray*}
In the above SDP, for every vertex $v \in V'$, the variable $x_v$ is meant to indicate whether the vertex $v$ should be included in the small-set $S$. The objective is exactly the vectorized form of the hypergraph cut function. The first constraint is used to control the size of the set, and the second constraint is to just to normalize the denominator of the objective. The $\ell^2_2$-triangle inequalities would be used to finding a embedding from which a set with small vertex boundary can be efficiently rounded. 
We describe the rounding algorithm \cite{lm16} in Section \ref{sec:hos}.

However this does not directly give us a reduction from \stronguniquegames~to \hypersse~because any correct algorithm has to ensure that the subset of $G$ output by the reduction has to also encode a valid labeling: 
there are no constraints in the SDP that rule out the choice of sets which intersect a cloud corresponding to a vertex more than once. 
Moreover, the above approximation algorithm for $\delta$-\smallsetvertexexpansion~only guarantees that the size of the set returned is $\Theta(\delta n)$. In particular, even if the algorithm returned a subset $S$ which encodes a valid partial labeling, the set itself could be of size, say $n/100$, which would leave $0.99$ fraction of the vertices unused, which is wasteful. 
We note that these issues are not specific to \stronguniquegames~and would have arisen
in the works on designing approximation algorithms for \uniquegames~as well\footnote{While
the approximation algorithms for \uniquegames~\cite{CMM06a,CMM06} were presented in 2006,
the algorithms for \smallsetexpansion~\cite{RST10a,BFKM11} came a few years later.}.
We use the ideas from approximation algorithms for \uniquegames~due to \cite{CMM06a,CMM06} to overcome these issues. 
We add orthogonality constraints among vertices belonging to the same cloud. In the SDP, for any vertex $v \in V_\cG$ and a pair of distinct labels $i,j \in [k]$, we add the constraint $\langle x_{(a,i)},x_{(a,j)} \rangle = 0$ 
(in \cite{lm16}'s reduction from vertex expansion to hypergraph expansion, 
$(a,i),(a,j)$ are also vertices of $H'$). 
Given a set of vectors, 
Chlamtac et al. \cite{CMM06} (see also \cite{CMM06a}) gave an algorithm to sample from a distribution
over subsets of these vectors which has the property that the probability of two orthogonal vectors being sampled
is ``small''; they called this distribution the {\em orthogonal separators}.
They showed how to iteratively use the algorithm for sampling from this distribution to
obtain a labeling for the \uniquegames~instance. 
Let $S^{(t)}$ denote the vertices of the \uniquegames~instance being labeled in iteration $t$. 
While one can argue about the probability of an edge inside $S^{(t)}$ being satisfied 
using the properties of the orthogonal separators, 
there might not be any easy way to argue about the edges/constraints between $S^{(t)}$ and $S^{(t')}$
for $t \neq t'$.
Therefore, as a conservative assumption, Chlamtac et al. \cite{CMM06} count all the edges leaving 
$S^{(t)}$ as unsatisfied edges and bound the number of such edges. 

At a high level, our approach can be viewed as using \cite{lm16}'s {\em hypergraph orthogonal separator}
(see Section \ref{sec:hos}) along with Chlamtac et al.'s approach \cite{CMM06}.
We iteratively use the hypergraph orthogonal separator to obtain a labeling of the vertices
in the instance. 
However, as before, for $t \neq t'$, it could be the case that $\sval(\cG[S^{(t)}]) = 1$ and 
$\sval(\cG[S^{(t')}]) = 1$, but taken together 
they might induce violated constraints in $\cG[S^{(t)} \cup S^{(t')}]$.
Analogous to the case of \uniquegames, at any iteration $t$, we delete the internal boundary of $S^{(t)}$,
\[
\partial_{\rm int}(S^{(t)}) \defeq \left\{ u \in S^{(t)} : \exists v \notin S^{(t)} \mbox{ s.t. }  v \in N_\cG(u) \right\}  
\]
i.e, any vertex which has a outgoing edge (and therefore can potentially be violated in the future iterations) is deleted. 
We can bound the number of such vertices deleted using the properties of the hypergraph orthogonal separator. 

While our analysis is similar in spirit to the analysis of Chlamtac et al. \cite{CMM06} for \uniquegames, 
there are subtle differences.

\subsection{Hardness of approximating \stronguniquegames}

The reduction for the above theorem goes through in steps with several intermediate problems which we describe below. 

\paragraph{\ug~to \hug .} The first step of the reduction, which is also the key step of the reduction, is to establish an arity dependent hardness for a hypergraph variant of \ug, which we call \hug.  Here, the constraints here are hyperedges. Specifically, for any hyperedge $e = (v_1,v_2,\ldots,v_d)$, we have $d$-bijection constraints $\{\pi_{e,v_i}\}_{i \in [d]}$, and we say that an assignment $(\sigma_1,\sigma_2,\ldots,\sigma_d) \in [k]^d$ satisfies the hyperedge $e$ if and only if $\pi_{e,v_i}(\sigma_i) = \pi_{e,v_j}(\sigma_j)$ for every $i,j \in [d]$. In this first step, we reduce an instance $\cG$ of $(\epsilon,1-\epsilon)$-\ug~ to an instance $\cH$ of $\left(1 - \epsilon, 1 - \bigo{\sqrt{\epsilon \log d \log k}}\right)$-\hug~(see Definition \ref{def:abhug} for formal definitions of these problems). We shall discuss this step in more details in Section \ref{sec:overviewstrug}.

\paragraph{\hug~to \sbug.} In the second step, we switch to a bipartite setting by reducing \hug~instances $\cH$ to a \strug~instance $G_{SB}$. Here, given a bipartite constraint graph $G$ (where the constraints are bijections), the objective is to find a labeling which maximizes the number of left vertices for which all incident constraints are satisfied. Furthermore, we shall also require that the left and right degrees to be bounded by ${\rm poly}(d)$, so that the final \strug~instance also has its degree bounded by ${\rm poly}(d)$. As a first step, given a \hug~instance, we can construct the bipartite graph $G_{SB}$ where the \hug~constraints and the \hug~vertices are the set of left ($V_L$) and right vertices ($V_R$) respectively. Furthermore, for every hypergraph constraint supported on hyperedge $ e = (v_1,v_2,\ldots,v_d)$ with projection functions $\pi_{e,v_i}$, we add the constraints $\pi_{e,v_i}$ between left vertex $u_e$ (identified with hyperedge $e$) and right vertex $v_i$, for every $i \in [d]$. Now it is easy to verify that given any labeling $\sigma:V \to [k]$, it can be extended to a labeling $\sigma':V_L \cup V_R \to [k]$ such that there is a one-to-one correspondence between the number of hyperedge constraints in $\cH$ satisfied by $\sigma$ and the number of left vertices for which all incident constraints are satisfied by $\sigma'$. In particular, $\sigma$ satisfies at least $(1 - \epsilon)$ fraction of hyperedge constraints in $\cH$ if and only if at least $(1 - \epsilon)$-fraction of left vertices in $G_{SB}$ have all the constraints incident on them satisfied by $\sigma'$.

Furthermore, since every hyperedge has arity $d$, the graph is $d$-left regular. However, the right degrees are not necessarily bounded by $d$. This is fixed by sub-sampling using the following two step process: (i) We construct $G'$ on vertex set $(V_L',V_R)$ as follows: for each right vertex $v \in V_R$, we sample $\ell$ left-neighbors (where $\ell = \ell(d)$ is a function of $d$) and  (ii) We construct $G''$ on vertex set $(V_L'',V_R')$ from $G$ by removing all large degree right vertices and all their left neighbors. Using standard tail bounds we can show that in (i) the completeness and soundness are approximately preserved and in (ii) only a small fraction of right vertices are deleted, which implies that the completeness and soundness parameters of the sub-sampled instance are again approximately preserved. Similar ideas based on random sampling have been used for degree reduction in the context of other problems as well such as vertex expansion \cite{LRV13,AKS09}, etc.

\paragraph{\sbug~to \strug .} The last step of the reduction transforms \sbug~to \strug~instances again preserving the completeness and soundness parameters exactly. This step is the same as the reduction in \cite{KR08}. However, we need to show that if we start with an instance of maximum degree $d$, then the reduction produces an instance of maximum degree ${\sf poly}(d)$; we verify that the reduction in \cite{KR08} does indeed satisfy this.
Their reduction is the following. Given an \sbug~instances $G_{SB}$ (with left and right degrees bounded by $d$ and $\bigo{d \ell}$ respectively), \cite{KR08} construct the graph $\cG$ on the vertex $V_L''$ as follows. Whenever a pair of left vertices $u,u'$ are incident on a common right vertex $v$, the composed constraint $\pi_{u \to u'} = \pi_{v,u'}\circ \pi^{-1}_{v,u}$ is added. Again, it is straightforward to verify that the reduction satisfies the following property. A labeling $\sigma:V_L'' \cup V_R \to [k]$ satisfies all the constraints incident on a set $S \subset V_L''$ if and only if its restriction $\sigma|_{V_L''}$ satisfies all constraints induced on $\cG[S]$. Furthermore, the neighbors of a vertex $u \in V_L''$ in $\cG$ are exactly the $2$-step neighbors of $u$ in $G''$. Therefore, the max degree of $\cG$ is bounded by ${\rm poly}(dl)$.

\subsection{Hardness of \hug} 
\label{sec:overviewstrug}

As mentioned earlier, the first step of the reduction is also the key technical step in proof of Theorem \ref{thm:strong-ug-informal}, which differs significantly from known hardness results for \hug. As is standard, the reduction from \ug~ to \hug~ is via a {\em dictatorship testing gadget}. Our reduction can be thought of as an extension of the {\em Long code} based dictatorship tests (for e.g., see \cite{KKMO07}, etc.) to the hypergraph setting using a new family of Markov chains (Section \ref{sec:gadget}) in which non-dictatorial cuts behave similar to the (analytic) small set vertex expansion of Gaussian graphs (Theorem \ref{thm:bound}). It is important to note that the choice of the Markov operator is crucial for the reduction. In particular, as we will describe below, a straightforward generalization of long code tests to the setting of hypergraph constraints will give us weaker hardness, since the noise hypercube (which is the Markov operator for standard long code based tests) does not have the desired vertex expansion properties. 

A first attempt at generalizing the long code test to the setting of \hug~is in Figure \ref{fig:pcptest-0-intro}. \\
\begin{figure}[h!]
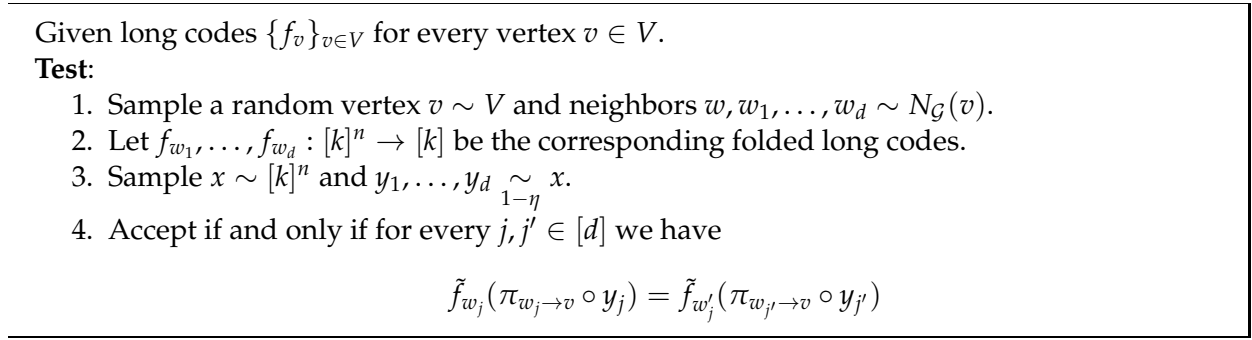

	\begin{mdframed}
		Given long codes $\{f_v\}_{v \in V}$ for every vertex $v \in V$. \\
		{\bf Test}: 
		\begin{enumerate}
			\item Sample a random vertex $v \sim V$ and neighbors $w,w_1,\ldots,w_d \sim N_{\cG}(v)$.  
			\item Let ${f}_{w_1},\ldots,{f}_{w_d}:[k]^n \to [k]$ be the corresponding folded long codes.
			\item Sample $x \sim [k]^n$ and $y_1,\ldots,y_d \underset{1 - \eta}{\sim} x$.
			\item Accept if and only if for every $j,j' \in [d]$ we have 
			\[
			\tilde{f}_{w_j}(\pi_{w_j \to v} \circ {y}_j) = \tilde{f}_{w_j'}(\pi_{w_{j'} \to v} \circ {y}_{j'}) 
			\]
		\end{enumerate}
	\end{mdframed}
	\caption{$d$-ary Long Code Test}
	\label{fig:pcptest-0-intro}
\end{figure}

In the above test, $\pi \circ x$ denotes the permuted string $(x_{\pi(1)},x_{\pi(2)},\ldots,x_{\pi(n)})$. The completeness and soundness of this test follows from known noise stability properties of the noisy $k$-ary hypercube \cite{KKMO07,MOO10}. 
For proving completeness, let $\sigma:V \to [n]$ be the labeling which satisfies at least $(1 - \epsilon)$-fraction of the constraints. Then, for every vertex $v \in V$, we have $f_v = \chi_{\sigma(v)}$ (i.e. the $\sigma(v)^{th}$ dictator function). Then, using the stability of the dictator functions it is easy to see that the above test accepts with probability at least $1 - \epsilon - \eta d$. On the other hand, analogous to \cite{KKMO07}, the soundness depends on the $d$-ary expansion of small sets in the hypercube. This can be seen as follows. For simplicity, suppose we assume that all the constraints in the hyperedges are identity constraints. Moreover, as is standard, we shall instead be interpreting the $k$-ary long codes as functions to the the $k-1$-dimensional simplex i.e, $f_v : [k]^n \to \Delta_k$. Furthermore, we will write the vector valued functions as $f_v = (f^1_v,\ldots,f^k_v)$. Using the above interpretation, the acceptance probability of the test can be expressed as 
\[
\Pr\left[\mbox{Test Accepts}\right] = \Ex_{v \sim V}\Ex_{w_1,\ldots,w_d \sim N_{\cG}(v)} \Ex_{x,\left(y_j\right)^d_{j = 1}}\left[\sum_{i \in [k]} \prod_{j \in [d]}f^v_j(y_j)\right].
\] 
Let $g_v \defeq \Ex_{w \sim v} f_v$ be the function averaged over the neighborhood. Using a sequence of now standard steps, we can express each of the inner expectation terms as
\begin{equation}				\label{eq:isop-bound}
\Ex_{x}\left[\sum_{i \in [k]} \left(\Gamma_{1 - \eta} g^i_v(x)\right)^d  \right] 
\end{equation}
where $\Gamma_{1 - \eta}$ is the noise operator on the $k$-ary hypercube. Now, for a fixed vertex, suppose we assume that the $\Gamma_{1 - \eta} g_v^i$s do not have influential coordinates (otherwise we can decode a labeling for $v$), we can use the {\em Invariance Principle} \cite{MOO10}s to pass on to the Gaussian space and lower bound the above expression by 
\[
\Pr_{g \sim N(0,1)^m}\Pr_{g_1,\ldots,g_d \underset{1 - \eta}{\sim} g} \left[\forall j \in [d]: G(g_j) = 1 \right].
\]
where $G$ multilinear polynomial representation of $g$ in terms of the Fourier basis of the $k$-ary noisy hypercube. Therefore the soundness analysis reduces to the following small set isoperimetry question over the Gaussian space: {\em What is the probability that all $d$ correlated copies of a Gaussian random variable lands inside a set $S$ of volume $1/k$?}. This can be answered by the general form of {\em Borell's isoperimetric inequality}~\cite{Bor85} which says that the above expression is maximized for halfspaces. Combining this with quantitative bounds given by the Gaussian isoperimetric inequality, we can lower bound the above expression by $1 - O(\sqrt{\eta \log d \log k})$. 

However, as can be seen, the above only gives us completeness $(1 - 2\eta d)$ vs. $1 - O(\sqrt{\eta \log d \log k})$-hardness i.e., the completeness and soundness guarantees cannot be combined together to get the desired inapproximability. The central issue here is in the choice of the underlying gadget i.e., {\em the noisy hypercube does not have small volume dictator cuts with small vertex expansion}. Instead, to get obtain the desired inapproximability, we need the underlying gadget to satisfy the following properties.
\begin{itemize}
	\item[1.] {\bf Completeness}: There are $k$-dictator cuts with vertex expansion at most $\epsilon$.
	\item[2.] {\bf Soundness}: The value of Eq. \ref{eq:isop-bound} for non dictatorial cuts of volume $1/k$ is at most $\frac1k - O(\frac1k\sqrt{\epsilon \log d \log k})$.
\end{itemize}
Towards this, we propose a new Markov chain based gadget that builds on the gadget used in the hardness reduction for optimal inapproximability (based on the small-set expansion hypothesis \cite{RS10}) of Vertex Expansion in graphs \cite{LRV13}. Informally, our gadget consists of vertex set $V = \uplus_{i \in [k]}V_i$ where each $V_i$  consists of two vertices $s_i,t_i$, where conditioned $a \in V_i$ we have $a = s_i$ w.p. $\approx 1 - \epsilon$. Furthermore, (a) the $s_i$ vertices are only connected to the $t_i$ vertices, (b) and the induced graph on the $t_i$ vertices form a expander. Here (a) ensures that each $V_i$ has small internal boundary, and (b) ensures that the spectral gap of the Markov chain is at least $\Omega(\epsilon)$. Therefore, using (a) and (b), we are above to obtain the desired completeness and soundness guarantees (see Claim \ref{cl:vert-exp}, Lemma \ref{lem:invar}). 

In summary, for our reduction, we define our long codes to be on the Markov chain based gadget; combining these codes with $d$-ary long code test gives us our reduction. We crucially use the spectral properties of the Markov Chain to establish the completeness and soundness directions. A crucial ingredient in the soundness analysis is the {\em Exchangeable Gaussians Theorem}, which is a generalization of Borell's isoperimetric inequality to non-spherically correlated Gaussians introduced by Isakkson and Mossel \cite{IM12}.  

\begin{remark}
	We point out that the result of Raghavendra~\cite{Rag08} implies that there exists a canonical SDP for \hug~whose integrality gap matches the optimal Unique Games based inapproximability for \hug. We leave bounding the integraility gap of the SDP corresponding to \hug~ as an open question.
\end{remark}

\paragraph{ Comparison with \cite{LRV13}.} Given that our results are motivated by the connection of \stronguniquegames~to small-set vertex expansion (Proposition \ref{prop-sse}), it is not surprising that components of our reduction are inspired by previous work such as \cite{LRV13} which studied SSE-hardness of approximating vertex expansion in graphs. Our choice of the Markov operator can be thought of as a generalization of the Markov operator of \cite{LRV13} to the small volume setting, along with some added bells and whistles such as positive semidefiniteness. The overall road map of the complete reduction also shares certain standard similar steps such as uniformization and sparsification, which are standard in reductions which show degree dependent hardness (e.g.,~\cite{AKS09}).

However,  several new technical ingredients are needed to obtain our inapproximability results for \stronguniquegames. Our construction of the Markov operator has the crucial property that it is an edge-expander (i.e. we obtain a lower bound on its {\em spectral gap}) in which small sets have small vertex expansion. This enables us to obtain a best-of-both-worlds guarantee (Claim \ref{cl:M-prop}). Consequently the isoperimetric inequalities required to obtain our soundness parameters are also derived differently (Theorem \ref{thm:bound}). Due to the different choice of the outer verifier i.e.,~\ug, the analysis of the underlying dictatorship tests (Theorem \ref{thm:sound}) rely on existing techniques \cite{KKMO07} which are more suited towards decoding labelings. Our re-interpretation of \hug~ as \ug~with hyperedge constraints as vertices (Theorem \ref{thm:bipartite-ug}), although elementary, is crucial in the chain of reductions. The subsequent steps seek to preserve distinct structural properties that arise out of the asymmetry in the roles of left and right vertices in \sbug, and hence require different combinatorial constructions and analysis (Theorem \ref{thm:deg-red}, Theorem \ref{thm:oct-red}). The technical differences in these steps are subtle.

\part{Approximation Algorithms}

\section{Preliminaries}
\label{sec:prelim}

We introduce some notation that will be used frequently in the rest of the paper. We use $G = (V,E)$ to denote a graph $G$ with vertex set $V$ and edge set $E$. For any subset $S \subset V$, we use $G[S]$ to denote the subgraph of $G$ induced by the vertex set $S$ Throughout the paper, we shall use $\cG = (V_\cG,E_\cG,[k],\{\pi_e\}_{e\in E_\cG}\})$ to denote a \stronguniquegames~instance as defined in Problem \ref{prop:str-ug}. We will use $n = |V_\cG|$ to denote the number of vertices in the \stronguniquegames instance, and use $V_\cG$ and $[n]$ interchangeably to denote the vertex set.

We shall always use $G = (V,E)$ to denote its label extended graph (as defined in Definition \ref{defn:lab-gr}), and $H = (V',E')$ to be the label extended hypergraph defined later in Definition \ref{defn:lab-hyper}. For a graph $G = (V,E)$ and a set subset $S \subseteq V$, we shall use $G[S]$ to denote the graph induced by $S$ in $G$, and extend the notation analogously to hypergraphs as well. For the label extended graph $G = (V,E)$ and the label extended hypergraph $H = (V',E')$, for any subsets $T \subset V$ and $T' \subset V'$, we define the ${\rm Vert}(T)$ and ${\rm Vert}(T')$ to be the projection of the subsets into the Unique Games vertex set. Formally, we shall define
\[
{\rm Vert}(T) = \left\{ a \in V_\cG| \exists i \in [k] \mbox{ s.t. } (a,i) \in T \right\}
\]
and 
\[
{\rm Vert}(T') = \left\{ a \in V_\cG| \exists i \in [k] \mbox{ s.t. } (a,i) \in T' \right\}
\]
Finally, for a subset $T \subseteq V$ we define $\partial^V_G(T) = \{v \notin T| N_G(v) \cap T \neq \emptyset \}$ to be the vertex boundary of the graph. Also closely related is the internal boundary of the set $T$ denoted by  $\partial_{\rm int}(T)$ which can be formally defined as
\[
\partial_{\rm int}(T) = \{ v \in T | N_G(v) \cap T^c \neq \emptyset\}
\] 
Similarly for a subset $T' \subseteq V'$ we use $\partial^{E'}_H(T') = \{e \in E' | e \cap T \neq \emptyset \mbox{ and } e \cap T^c \neq \emptyset \}$ to denote the hyperedge boundary of the set. Whenever the context is clear, we will drop the indexing by $G$ and $H$ in the notation for vertex boundary and hyperedge boundary respectively.

\subsection{Hypergraph Orthogonal Separators.}
\label{sec:hos}
A key tool used in the rounding algorithms are hypergraph orthogonal separators, which were first introduced in the context of hypergraphs in \cite{lm16}.

\begin{definition}[Hypergraph Orthogonal Separator]				
\label{def:hos-1}
Let $X = \set{\bar{u}: u \in V}$ be a set of vectors in the unit ball that satisfy
$\ell_2^2$ triangle inequalities. We say that a random set $S \subset V$ is a 
{\em hypergraph $m$-orthogonal separator} with distortion $D \geq 1$, probability scale
$\alpha > 0$, and separation threshold $\beta \in (0,1)$, if it satisfies the following properties.
\begin{enumerate}
	\item For every $u \in V$, $\Prob{u \in S } = \alpha \norm{\bar{u}}^2$.
	\item For every $u,v \in V$ such that $\norm{\bar{u} - \bar{v}}^2 \geq 
	\beta \min \set{\norm{\bar{u}}^2, \norm{\bar{v}}^2}$,
	\[ \Prob{u \in S \textrm{ and } v \in S} \leq \alpha \frac{\min \set{\norm{\bar{u}}^2, \norm{\bar{v}}^2}}{m}  \]
	\item For every $e \subset V$, $\Prob{e \textrm { is cut by } S} \leq \alpha D \max_{u,v \in e} \norm{\bar{u} - \bar{v}}^2$. 
\end{enumerate}
\end{definition}
 
For proving Theorem \ref{thm:partial-ug-1}, we shall use the construction guaranteed by the following theorem from \cite{lm16},

 which is a slight refinement of $\ell^2_2$-separators defined in 
 \begin{theorem}  
 	\label{thm:hos-1}
 	Let $d \geq 2$ be an integer and let $H = (V,E)$ be a hypergraph. Then there is a polynomial-time randomized algorithm that given a set
 	of vectors $\set{\bar{u}: u \in V}$ satisfying $\ell_2^2$ triangle inequalities ,
 	parameters $m \geq 2$ and $\beta \in (0,1)$, generates a hypergraph $m$-orthogonal separator with 
 	probability scale $\alpha \geq 1/n$ and distortion $D = \bigo{\beta^{-1} m \log m \log \log m \sqrt{\log n (1/\beta)}}$.
 \end{theorem}
 
Additionally, we shall also be using the notion of $\ell_2$-$\ell_2^2$ hypergraph orthogonal separators due to \cite{lm16} in proving Theorem \ref{thm:partial-ug-2}.

\begin{definition}[$\ell_2$-$\ell_2^2$ Hypergraph Orthogonal Separator]
\label{def:hos-2}
Let $X = \set{\bar{u}: u \in V}$ be a set of vectors in the unit ball that satisfy
$\ell_2^2$ triangle inequalities. We say that a random set $S \subset V$ is a 
{\em $\ell_2-\ell_2^2$ hypergraph $m$-orthogonal separator} with $\ell_2$-distortion 
$D_{\ell_2}: \N \to \R$, $\ell_2^2$-distortion $D_{\ell_2^2}$, probability scale
$\alpha > 0$, and separation threshold $\beta \in (0,1)$, if it satisfies the following properties.
\begin{enumerate}
	\item For every $u \in V$, $\Prob{u \in S } = \alpha \norm{\bar{u}}^2$.
	\item For every $u,v \in V$ such that $\norm{\bar{u} - \bar{v}}^2 \geq 
	\beta \min \set{\norm{\bar{u}}^2, \norm{\bar{v}}^2}$,
	\[ \Prob{u \in S \textrm{ and } v \in S} \leq \alpha \frac{\min \set{\norm{\bar{u}}^2, \norm{\bar{v}}^2}}{m}  \]
	\item For every $e \subset V$, 
	\[ \Prob{e \textrm { is cut by } S} \leq \alpha D_{\ell_2^2} \max_{u,v \in e} \norm{\bar{u} - \bar{v}}^2
	+ \alpha D_{\ell_2}(\Abs{e}) \min_{w \in e} \norm{\bar{w}} \cdot \max_{u,v \in e} \norm{\bar{u} - \bar{v}} . \] 
\end{enumerate}
\end{definition}

Definition \ref{def:hos-2} differs from Definition \ref{def:hos-1} only in item 3.

\begin{theorem}[\cite{lm16}]
\label{thm:hos-2}
There is a polynomial-time randomized algorithm that given a set of vertices $V$, a set
of vectors $\set{\bar{u}: u \in V}$ satisfying $\ell_2^2$ triangle inequalities ,
parameters $m \geq 2$ and $\beta \in (0,1)$, generates a $\ell_2-\ell_2^2$ hypergraph $m$-orthogonal separator with 
probability scale $\alpha \geq 1/n$ and distortions 
\[D_{\ell_2}(r) = \bigo{\beta^{-1/2} m \log m \log \log m \sqrt{\log r}} \qquad \textrm{and} \qquad
	D_{\ell_2^2} = \bigo{m} .\]
where $r$ is the largest arity of any hyperedge in the hypergraph.
\end{theorem}

\subsection{Notions of Expansions}

In this work, we frequently use the following well studied notions of small set expansions.
Given a graph $G = (V,E)$, the vertex expansion of a set $S$ is defined as $\phi^V_{G}(S) := \frac{|\partial^V_{G}(S)|}{|S|}$,
where $\partial^V(S)$ is the set of vertices in $V \setminus S$ whose neighborhood intersects with $S$.
While the vertex expansion only accounts for the outer boundary of a set, another related notion i.e., that of symmetric vertex expansion takes into account the inner boundary as well. Formally, again for a graph $G = (V,E)$, we define the symmetric expansion of a set $S$ as $\Phi^V_{G}(S) := \frac{|\partial_{\rm int}(S)| + |\partial^V_{G}(S)|} {|S|}$,
where $\partial_{\rm int}(S) := \partial^V_{G}(V \setminus S)$. 
Finally, given a hypergraph $H = (V,E)$, we define the edges expansion of a set as
$\phi^E_H(S)  = \frac{|\partial^E_H(S)|}{|S|}$,
where $\partial^E_H(S)$ is the set of hyperedges with at least one vertex in $S$ and at least one vertex in $V \setminus S$.

In particular, we will be interested in the small set profile for the above notions of expansions. The $\delta$-\smallsetvertexexpansion~of a graph is defined as 
\[
\phi^V_{\delta}(G) = \min_{S \subset V : |S| {= }\delta n} \phi^V_{G}(S)
\]
Analogously, we can also define the small set symmetric vertex expansion of a graph  and the small set edge expansion of a hypergraph $H = (V,E)$ as 
\[
\Phi^V_{\delta}(G) \overset{\rm def}{=} \min_{S \subset V : |S| {= }\delta n} \Phi^V_{G}(S) \qquad\qquad{\rm and}\qquad\qquad \phi^E_{\delta}(H) \overset{\rm def}{=} \min_{S \subset V : |S| {= }\delta n} \phi^E_{H}(S) 
\]

\subsection{The Label Extended Hypergraph}

We recall the standard reduction from \smallsetvertexexpansion~to Small Set Hypergraph Edge Expansion \cite{LRV13,lm16}. It is important to note that we do not use the reduction in a black box way; the analysis of the algorithm will explicitly use details of the constructions used in the reduction. To begin with, the following lemma states the properties of the reduction.

\begin{lemma}			\label{lem:transform}
	Given a graph $G = (V,E)$, one can construct an intermediate weighted graph $G_\sym(V_\sym,E_\sym,w_\sym)$, and a final weighted hypergraph $H(V',E',w)$ such that (i) $\Phi^{V_\sym}_\delta(G_\sym)  \leq \phi^V_{2\delta}(G)$ and $\phi^{E'}_{2\delta}(H) \leq \Phi^{V_\sym}_{2\delta}(G_{\rm sym})$
\end{lemma}
\begin{proof}
	The construction is as follows. 
	
	{\bf Step (i)} $G \mapsto G_{\rm sym}$. Given $G = (V,E)$ we construct a weighted bipartite graph $G_{\sym} = (V_\sym,E_\sym,w_\sym)$ where $V_\sym = V \cup E$ and we add and edge between $v \in V$ and $e \in E$ if edge $e$ is incident to vertex $v$. For every $v \in V$, we assign weight $w_{\sym}(v) = 1$, and for every edge $(u,v) \in E$ we assign the corresponding vertex weight $w_\sym((u,v)) = \min\left(\frac{1}{\rm deg(u)},\frac{1}{\rm deg(v)}\right)$. 
	
	{\bf Step (ii)} $G_{\rm sym} \mapsto H$. Given $G_\sym = (V_\sym,E_\sym)$ we construct the weighted hypergraph $H = (V',E')$ with $V' = V_{\rm sym}$. Furthermore, we include the following hyperedegs: for every vertex $v \in V'$, we introduce the hyperedge $\{v\} \cup N_{G_\sym}(v)$ with weight $w\left(\{v\} \cup N_{G_\sym}(v)\right) = w_\sym(v)$.
	
	It can be verified that the above construction satisfies the expansion preserving properties, for the complete proofs see \cite{LRV13}. 
	
\end{proof}

Now we define the Label Extended Hypergraph as follows

\begin{definition}[Label Extended Hypergraph]			\label{defn:lab-hyper}
	The label extended hypergraph is the weighted hypergraph $H = (V_H,E_H,w)$  obtained by applying Lemma \ref{lem:transform} on the label extended graph $G$. In particular, the vertex set $V_H = V \cup E$ consists of graph vertices $(u,i) \in V$ and edge vertices $\{(u,i),(v,j)\} \in E$. Furthermore, every hyperedge $e \in E_H$ can be identified with a vertex $v \in V_H$ such that $e = \{v\} \cup N_{G_\sym}(v)$. We denote by $\cE(v)$ the hyperedge corresponding to vertex $v$. 
\end{definition}

\section{Approximation Algorithm for \stronguniquegames}		
\label{sec:algo}

We use Algorithm \ref{alg:main} for Theorem \ref{thm:partial-ug-1}.

\begin{algorithm}
	\SetAlgoLined
	\KwIn{A \stronguniquegames~Instance $\mathcal{G}(V_\cG,E_\cG,[k],\{\pi_{e,v}\}_{e \in E,v \in e})$ such that $\sval(\cG) \geq 1 - \epsilon$.}
	Construct the label extended graph $G = (V,E)$ as in Definition \ref{defn:lab-gr}\;
	Construct the label extended hypergraph $H = (V',E',w)$ from $G$ as in Definition \ref{defn:lab-hyper} using $G_\sym$ as the intermediate graph as in Lemma \ref{lem:transform}\;
	Solve the following StrongUG-SDP:
	\[	\text{minimize} \qquad  \sum_{e \in E'} w(e)\max_{u,v \in e}\| x_u - x_v\|^2  \]
		\text{subject\space to}
	\begin{align}
		&\sum_{a \in [n]}\sum_{i \in [k]}\|x_{(a,i)}\|^2 = n(1-\epsilon) \label{eq:c1} \\
		&\sum_{i \in [k]} \|x_{(a,i)}\|^2 \leq 1& \forall a \in [n] \label{eq:am1}\\
		&\|x_u - x_v\|^2 +\|x_u - x_w\|^2 \geq \|x_w - x_v\|^2 & \forall u,v,w \in V' \label{eq:l22} \\
		&\langle x_u, x_v \rangle = 0 &\forall u = (a,i),v = (a,j), i ,j \in [k], i\neq j, a \in V \label{eq:ort} \\
		&0 \leq \langle x_u, x_v \rangle \leq \|x_u\|^2 &\forall u,v \in V' \label{eq:cs}
	\end{align}\\
	Use the Algorithm \ref{alg:ver} to round the SDP solution, and output the solution obtained.
	\caption{Partial Unique Games}
	\label{alg:main}
\end{algorithm}

Given a \stronguniquegames~instance $\cG(V_\cG,E_\cG,[k],\{\pi_e\}_{e \in E})$ such that $\sval(\cG) \geq 1- \epsilon$, the algorithm first constructs the label extended graph $G = (V,E)$. Using Proposition \ref{prop:ssve-ug}, we know that it contains a non-repeating set $S \subset V$ such that $\phi^V_G(S) \leq \epsilon k$ such that $|S| \geq n(1- \epsilon)$. Then using the sequence of transformations described in Lemma \ref{lem:transform}, the algorithm constructs the label extended hypergraph $H = (V',E')$ (as in Definition \ref{defn:lab-hyper}). Again using the properties of the reduction in Lemma \ref{lem:transform}, we know that there exists a non-hypergraph expanding set $S' \subseteq V'$ of weight $2\delta |V'|$ such that $\phi^{E'}_{H}(S') \leq 2 \epsilon n k$. We then solve the Strong-UG SDP on $H$. The objective of the SDP is $\sum_{e \in E'} \max_{u,v \in e}\| x_u - x_v\|^2$ which is a relaxation of the cut function of the hypergraph.
We now give a brief description of each of the constraints of the above SDP.
\begin{itemize}
	\item[$\tright$] $\sum_{a \in [n]}\sum_{i \in [k]}\|x_{(a,i)}\|^2 = n(1-\epsilon)$ is used to normalize the denominator in the expression of hyeprgraph expansion.
	\item[$\tright$] $\sum_{i \in [k]}\|x_{(a,i)}\|^2 \leq 1$ and $\langle x_{(a,i)},x_{(a,j)} \rangle = 0$ are intended to ensure that for every $a \in [n]$ at most one label is chosen per vertex.
	\item[$\tright$] $0 \leq \langle x_u,x_v \rangle \leq \|x_u\|^2$ and the $\ell^2_2$-triangle inequalities ensure that the vector SDP solution is $\ell^2_2$-embeddable, and therefore can be rounded orthogonal separators.
\end{itemize} 

In particular, using $S'$ as witness, we can construct a feasible $\{0,1\}$-solution for which the value of the objective is at most $2\epsilon n k$, which implies that the optimal vector solution obtained by solving the SDP also has objective at most $2\epsilon n k$ (Lemma \ref{obs:spd-opt}). The algorithm proceeds use this vector solution to construct a solution to the \stronguniquegames~instance using the VertexExpansionRound algorithm (Algorithm \ref{alg:ver}), which we describe below.

\subsection{Rounding the SDP solution}

The algorithm for rounding of the vector solution is the following.

\begin{algorithm}
	\SetAlgoLined
	\KwIn{Vector solution $\{x_v\}_{v \in V'}$ for StrongUG-SDP}
	
	{\it Thresholding}: \label{step:thr}
	Construct set $V' \gets \left\{a \in [n] | \sum_{i \in [k]} \|x_{(a,i)}\|^2  \geq \frac12\right\}$\;
	Initialize hypergraph orthogonal separator (Theorem \ref{thm:hos-1}/Theorem \ref{thm:hos-2})
	with parameters $m = 10k, \beta = 1/2$.
	Initialize $t \gets 1$,$H_1 \gets H[V']$ and $C^{(1)} \gets \emptyset$\;
	\While{$H_t  \neq \emptyset$ and  $t \leq \frac{10}{\alpha}\log n $}{ \label{step:while}
		Sample hypergraph orthogonal separator $\Gamma^{(t)}$ for hypergraph $H[V_t]$ as in
		(Theorem \ref{thm:hos-1}/Theorem \ref{thm:hos-2})\;			
		{\it Deleting the boundary}: \label{step:boundary}
		Delete the set of vertices $\left\{v \in V_\cG \Big| \exists  i \in [k] : \cE(v,i) \in \partial_{H_t}(\Gamma^{(t)}) \right\}$ from $V_\cG$ and update the hypergraph $H_t$ and the set $\Gamma^{(t)}$ accordingly\;
		{\it Identifying uniquely labeled vertices}:
		Let $S^{(t)} \gets \left\{(v,i) \in \Gamma^{(t)} : \not{\exists} j \in [k]\setminus\{i\} \mbox{ s.t. } (v,j) \in \Gamma^{(t)} \right\}$\; 
		{\it Update the labeled set and hypergraph for next iteration}:
		Update the following	
		\[
		C^{(t+1)} \gets C^{(t)} \cup S^{(t)} \qquad\qquad\qquad V^{(t+1)}  \gets V^{(t)} \setminus {S}^{(t)} \qquad\qquad\qquad t \gets t + 1.
		\]
	}
	\If{$H_t \neq \emptyset$}{
		Delete left over vertices from $H_t$\;
	}
	Return the subset $V_{\rm fin} = \{v \in V| \exists i \in [k] : (v,i) \in C^{(t)}\}$ and the corresponding labeling $\sigma_{\rm fin} :V \mapsto [k]$;
	
	\caption{VertexExpansionRound}
	\label{alg:ver}
\end{algorithm}

The VertexExpansionRound algorithm takes as input the optimal solution $\{x_{v}\}_{v \in V'}$ to the SDP. It then proceeds to decide which vertices are to be deleted and retained $\cG$ in two phases. We describe them in sequence.

\paragraph{The thresholding step.} In the first step, the algorithm deletes all vertices $a \in [n]$ such that $\sum_{i \in [k]}\|x_{a,i}\|^2  \leq \frac12$ i.e, vertices which are not expected have large contribution to the objective. By averaging, it follows that $V'$ retains at least $(1 - 2\epsilon)$-fraction of vertices (Lemma \ref{lem:thr}).

\paragraph{The ``while'' loop.} This is the core component of the rounding algorithm. At each iteration, the algorithm proceeds to find a subset of the label extended graph on the surviving vertices with small vertex expansion, then deletes the vertices in the internal boundary of the set, thereby making sure that (i) the subset of vertex-label pairs chosen at iteration $t$ satisfy all the constraints induced amongst themselves, and (ii) they will not induce violated edges with new vertices picked in subsequent iterations. 

This idea is implemented in the following way. At iteration $t$, the algorithm samples a hypergraph orthogonal seperator $\Gamma^{(t)}$, and deletes all the vertices $v \in V'$ (and even the corresponding \stronguniquegames~vertices in $\cG$) for which the corresponding hyperedge $\cE(v)$ crosses the cut i.e., $\cE(v) \in \partial^{E'}_{H_t}(\Gamma^{(t)})$. Then it includes all the remaining vertices which appear with unique labels (i.e., $\{v \in V_\cG : |\cC_v \cap \Gamma^{(t)}| = 1\}$) to the current solution $S^{(t)}$ and puts back all the remaining vertices into the hypergraph for consideration in future iterations. A careful charging argument ensures that every vertex deleted from the internal boundary of the set can be charged to its hyperedge getting cut (Lemma \ref{lem:Vf-sat}). Furthermore, the number of hyperedges getting cut can be bounded as a function of the SDP objective using the guarantees of the orthogonal separators (Lemma \ref{lem:V-bound}).


\subsection{Proof of Theorem \ref{thm:partial-ug-1}}
\label{sec:main-proof}	

The proof of Theorem \ref{thm:partial-ug-1} will require us to show that the solution $(V,\sigma)$ returned by the algorithm satisfies the following conditions.

\begin{itemize}
	\item[(i)] The labeling $\sigma : V \mapsto [k]$ satisfies all the constraints in the induced game $\cG[V]$
	\item[(ii)] The set $V$ is not too small.
\end{itemize}

The proof of the first point relies on the following iterative invariant maintained by the algorithm. At any iteration $t$, recall that $V_t$ is the set of vertices added to the set $S$. By deleting all the hyperedges that cross the hypergraph orthogonal separator in iteration $t$, the algorithm ensures that the internal boundary of the set $V_t \cup S^{(t)}$ in the label extended graph $G$, and consequently all the constraints induced in the set $S^{(t)}$ is always satisfiable. 

For point (ii), observe that the VertexExpansionRound algorithm discards vertices in two phases (a) the thresholding step  and (b) the iterations of the while loop. For the thresholding step, we observe that any feasible SDP solution must satisfy the constraint $\sum_{a \in [n]} \sum_{i \in [k]} \|x_{(a,i)}\|^2 \geq n(1-\epsilon)$ and $\sum_{a \in [n]}\|x_{(a,i)}\|^2 \leq 1$ for every vertex $a \in [n]$, and therefore by an averaging argument we can show that the thresholding step retains most of the vertices. Secondly, the analysis of the while loop involves the observation than during any of its iterations, the cost of any vertex getting deleted can be charged to the corresponding hyperedge getting cut by the hypergraph orthogonal separator. Therefore, using the third property of the orthogonal separators, we can bound the expected number of vertices discarded in each iteration as a function of the SDP optimal. Finally, using the first and second properties of the hypergraph orthogonal separator, we show that after $T$ iterations (for an appropriate choices of $T$) all vertices must have been removed from the hypergraph (by either adding them to the set $S$, or by discarding them from consideration). In the following subsections, we establish the points discussed here formally.

\subsubsection{$\mathcal{G}[V_{\rm fin}]$ is fully satisfiable}			
\label{sec:soln-sat}

Here we show that the Unique Games instance induced on the set $V_{\rm fin}$ is fully satisfiable using the labeling $\sigma_{\rm fin}$. In fact, we shall prove a stronger statement which claims that for any iteration $t$, the Unique Game induced on the set of vertices collected till that iteration is completely satisfiable.
\begin{lemma}			\label{lem:Vf-sat}
For any iteration $t$, let $V^{(t)}_\cG := {\rm Vert}\left(C^{(t)}\right)$ denote 
the set of \stronguniquegames~vertices collected till the end of iteration $t$ 
and let $\sigma_{C^{(t)}}:V^{(t)}_\cG  \mapsto [k]$ denote the corresponding labeling 
determined by the set $C^{(t)}$. Then $\sigma_{C^{(t)}}$ satisfies all the constraints in the 
\uniquegames~instance induced on $\cG\left[V^{(t)}_\cG\right]$.
\end{lemma}
Towards proving this lemma, we first make a couple of easily verifiable observations that will be useful in proving the main lemma of this sections.

\begin{observation}		\label{obs:ob2}
	For any iteration $t$, let $\hat{V}^{(t)} = \left\{v \in V_{\cG}| \exists ! i \in [k] \mbox{ s.t. } (v,i) \in \Gamma^{(t)}  \right\}$ be the set of vertices that are sampled with unique labels. Then $\hat{V}^{(t)} \cap V^{(t')}_H = \emptyset$ for every iteration $t' \geq t+1$.
\end{observation}
\begin{proof}
	Since $V^{(t')}_H \subseteq V^{(t)}_H$ for every iteration $t' \geq t+1$, it suffices to show the statement for $t' = t+1$. Observe that for any vertex $v \in \hat{V}^{(t)}$, at least one of the events holds for iteration $t$.
	\begin{itemize}
		\item $\cE((v,i)) \in \partial_{H^{(t)}}\left(\Gamma^{(t)}\right)$ for some $i \in [k]$. In this case, the vertex $v$ gets deleted from $H^{(t)}$.
		\item $\cE((v,i)) \notin \partial_{H^{(t)}} \left(\Gamma^{(t)}\right)$ for any $i \in [k]$. In this case, the label of the vertex gets decided in iteration $t$, and it is removed from $H^{(t)}$.
	\end{itemize} 
	So in both cases above, the cloud of vertices corresponding to vertex $v$ gets removed from $H^{(t+1)}$.
\end{proof}
	
\begin{observation} 		
\label{obs:ob3}
	For any iteration $t$, let $A \subset S^{(t)}$ be such that for all $u \in A, i \in [k]$, $\cE((u,i)) \notin \partial_{H^{(t)}}(\Gamma^{(t)})$. Furthermore, let $\sigma_A : A \mapsto [k]$ be the corresponding labeling. Then all constraints induced in $\cG[{\rm Vert}(A)]$ are completely satisfied by $\sigma_A$.
\end{observation}
\begin{proof}
	We prove this by contradiction. Suppose there exists a edge $(u,v)$ in the induced \uniquegames~instance $\cG[{\rm Vert}(A)]$ i.e., $\sigma_A(v) \neq \pi_{uv}(\sigma_A(u)) = l$ (say). Then by definition of $\sigma_A$, we must have $(v,l) \notin A$. But $(v,l) \in N_{G_{\rm sym}}(u,i)$ which implies that the edge $e = \{(u,i),(v,l)\} \in \cE((v,l))$ i.e, $\cE((v,l)) \in \partial_{H^{(t)}}(\Gamma^{(t)})$, which gives us the contradiction.  
\end{proof}

Our proof of Lemma \ref{lem:Vf-sat} uses the following iterative invariant maintained by the algorithm: for any iteration $t$, the set of vertices collected in $C^{(t+1)}$ has an empty internal boundary in the hypergraph $H^{(t+1)}$. 
This directly implies that any labeling assigned to vertices added in subsequent iterations is not going to induce violated edges incident on the $V^{(t)}_\cG$. 
	
\begin{lemma}[Iterative Invariant]
\label{lem:iter}
For any iteration $t$, 
	$N_{\cG}\left({\rm Vert}(C^{(t)}) \right) \cap {\rm Vert}(H^{(t)}) = \emptyset$.
\end{lemma}
	
\begin{proof}
	The proof is by induction on the number of iterations. We assume that the claim holds up to some iteration $t$ i.e., 
$N_{\cG}\left({\rm Vert}(C^{(t)}) \right) \cap {\rm Vert}(H^{(t)}) = \emptyset$. 
Recall that $S^{(t)}$ is the set of ``uniquely picked'' vertices in iteration $t$ 
and that ${\rm Vert}(C^{(t+1)}) = {\rm Vert}(C^{(t)}) \cup {\rm Vert}(S^{(t)}) $.
		
Therefore, to prove the inductive claim it suffices to show that 
(i) $N_{\cG}\left({\rm Vert}(C^{(t)}) \right) \cap {\rm Vert}(H^{(t+1)}) = \emptyset$ and 
(ii) $ N_{\cG}\left({\rm Vert}(S^{(t)}) \right) \cap {\rm Vert}(H^{(t+1)}) = \emptyset$. 
The first point follows from the induction hypothesis, since 
$N_{\cG}\left({\rm Vert}(C^{(t)}) \right) \cap {\rm Vert}(H^{(t+1)})  
\subseteq 
N_{\cG}\left({\rm Vert}(C^{(t)}) \right) \cap {\rm Vert}(H^{(t)}) = \emptyset$.
For the second point, fix any vertex $w \in {\rm Vert}(H^{(t+1)})$. There are two cases:
\begin{enumerate}
\item
{\bf $\cC_w \cap \Gamma^{(t)} = \emptyset$.} If $w \notin N_\cG({\rm Vert}(S^{(t)}))$, then we are done. Otherwise, suppose there exists $(u,i) \in S^{(t)}$ such that $w \in N_{\cG}(u)$.
Since vertex $u$ does not get deleted in Step \ref{step:boundary} of Algorithm \ref{alg:ver}, 
we have $\cE((u,i)) \notin \partial_{H^{(t)}}(\Gamma^{(t)})$, which implies that the vertex $\{(u,i),(w,\pi_{uw}(i))\} \in \Gamma^{(t)}$.
Since $\cC_w \cap \Gamma^{(t)} = \emptyset$, this implies that 
$\cE((w,\pi_{uw}(i))) \in \partial_{H^{(t)}}(\Gamma^{(t)})$, and therefore vertex $w$ 
gets deleted in Step \ref{step:boundary} of Algorithm \ref{alg:ver}, in iteration $t$. 
Therefore, any vertex $w$ such that $\cC_w \cap \Gamma^{(t)} = \emptyset$ 
belongs to ${\rm Vert}(H^{(t+1)})$ only if $w \notin N_\cG({\rm Vert}(S^{(t)}))$.		

\item
{\bf $\cC_w \cap \Gamma^{(t)} \neq \emptyset$.} 
Using Observation \ref{obs:ob2}, since $w \in {\rm Vert}(H^{(t+1)})$, we must have 
$|\cC_w \cap \Gamma^{(t)}| \geq 2$.
Suppose $(w,j),(w,j') \in \cC_w \cap \Gamma^{(t)}$ for a pair of distinct labels $j,j' \in [k]$.
Furthermore, since vertex $w$ does not get deleted in Step \ref{step:boundary} of Algorithm \ref{alg:ver} in iteration $t$, 
we must have $\cE((w,j)),\cE((w,j')) \notin \partial_{H^{(t)}}(\Gamma^{(t)})$.
If $N_{\cG}(w) \cap  {\rm Vert}(S^{(t))} = \emptyset$, then we are done. 
If not, fix a vertex $(u,i)\in S^{(t)}$ such that $u \in N_{\cG}(w)$. 
Since the vertex $(u,i)$ is uniquely picked from $\cC_u$, $(u,\pi_{wu}(j)) \notin \Gamma^{(t)}$ or $(u,\pi_{wu}(j')) \notin \Gamma^{(t)}$. 
We claim that $\cE((u,\pi_{wu}(j))) \in \partial_{H^{(t)}}(\Gamma^{(t)})$ or $\cE((u,\pi_{wu}(j'))) \in \partial_{H^{(t)}}(\Gamma^{(t)})$
To see this, without loss of generality, suppose $(u,\pi_{wu}(j)) \notin \Gamma^{(t)}$. Note that by construction $N_{G_\sym}((w,j)) = \left\{ \{(w,j)(v,\pi_{wv}(j))\} : v \in N_{\cG}(w) \right\}$ and therefore 
\[ \cE((w,j)) = \left\{(w,j)\right\} \cup \left\{\{(w,j),(v,\pi_{wv}(j))\} : v \in N_\cG(w) \right\}.\]
Let $l = \pi_{wu}(j)$. Then 
$\cE((u,l))= \left\{(u,l)\right\} \cup \left\{\{(u,l),(v',\pi_{uv'}(l))\} : v' \in N_\cG(u) \right\}.$
Now consider the edge $e = \{(w,j),(u,l)\}$ in the label extended graph. 
By construction, $e$ is a vertex in the hypergraph $H^{(t)}$, and  $e \in \cE((w,j))$ 
and $e \in \cE((u,l))$. But $\cE((w,j)) \notin \partial_{H_T}(\Gamma^{(t)})$ and $(w,j) \in \Gamma^{(t)}$ 
which implies that $e \in \Gamma^{(t)}$. This together with the fact that $(u,l) \notin \Gamma^{(t)}$ implies that $\cE((u,l)) \in \partial_{H^{(t)}}(\Gamma^{(t)})$.
		
From the above observation it follows that vertex $u$ must get deleted. Since the 
above arguments apply to any choice of $u \in N_\cG(w)$ we have that all the vertices in 
$N_{\cG}(w) \cap {\rm Vert}(S^{(t)})$ get deleted in Step \ref{step:boundary} of 
Algorithm \ref{alg:ver} in  iteration $t$.
\end{enumerate}
Therefore, the iterative invariant holds.		
\end{proof}

\begin{proof}[Proof of Lemma \ref{lem:Vf-sat}]	
Using Lemma \ref{lem:iter}, we now complete the proof of the Lemma \ref{lem:Vf-sat}. 
The proof is again by induction on the number of iterations. Suppose 
$\sval(\cG[V^{(t)}_\cG])  = 1$ with labeling $\sigma_{C^{(t)}}$, for some iteration $t$. 
We shall then use this to show that $\sval(\cG[V^{(t+1)}_\cG]) = 1$ with 
labeling $\sigma_{C^{(t+1)}}$. Note that $C^{(t+1)} \gets C^{(t)} \cup S^{(t)}$
and by induction hypothesis we know $\sval(\cG[V^{(t)}_\cG]) = 1$ with labeling $\sigma_{S^{(t)}}$.
Furthermore from Lemma \ref{lem:iter}, we know vertices in $V^{(t)}_\cG$ do not 
share constraints with ${\rm Vert}(S^{(t)})$. Therefore, it suffices to show that 
$\sval(\cG[S^{(t)}]) = 1$ with labeling $\sigma_{S^{(t)}}$. However, we know that 
for every $(u,i) \in S^{(t)} \times [k]$, we must have 
	$\cE((u,i)) \notin \partial_{H^{(t)}}(\Gamma^{(t)})$ (otherwise vertex $u$ gets deleted
in Step \ref{step:boundary} of Algorithm \ref{alg:ver}). Therefore, using 
Observation \ref{obs:ob3} we get that $\sigma_{S^{(t)}} : S^{(t)} \mapsto [k]$ 
satisfies all the constraints in \uniquegames~instance induced on 
${\rm Vert}(S^{(t)})$. Therefore, we have that $\sval(\cG[V^{(t+1)}_\cG]) = 1$ with 
labeling $\sigma_{C^{(t+1)}}$, which proves the induction step.
\end{proof}

\subsubsection{Bounding number of discarded vertices}			
\label{sec:vert-bound}
The number of vertices discarded in step \ref{step:thr} of Algorithm \ref{alg:ver} can be bounded
by an application of Markov's inequality.
\begin{lemma}			
\label{lem:thr}
Let $V'$ be the thresholded set of vertices constructed in Step \ref{step:thr} of Algorithm \ref{alg:ver}. 
	Then $|V'| \geq (1-2\epsilon)n$.
\end{lemma}
\begin{proof}
	From the SDP constraints, we have that 
	\[
	\E_{a \sim [n]} \left[\sum_{i \in [k]} \|x_{(a,i)}\|^2\right]  = \frac1n \sum_{a \in [n]}\sum_{i \in [k]} \|x_{a,i}\|^2 \geq (1-\epsilon).
	\]
	On the other hand, for every $a \in [n]$, we also have that $\sum_{i \in [k]} \|x_{(a,i)}\|^2 \leq 1$. Therefore, using Markov's inequality, it follows that for at least $1 - 2 \epsilon$ fraction of the vertices we have $\sum_{i \in [k]}\|x_{(a,i)}\|^2 \geq 1/2$.  
\end{proof}

Bounding the number of vertices deleted in Step \ref{step:while} (the ``while loop'') of Algorithm \ref{alg:ver}
requires more work.
As a first step, we show that the optimal value of the StrongUG-SDP is small as stated below.
\begin{restatable}{lem}{SdpOpt}  \label{obs:spd-opt}
	The SDP optimal is at most $2\epsilon n k$.
\end{restatable}

\begin{proof}
	It suffices to show a feasible solution $(\hat{x}_v)_{v \in V_H}$ under which the objective is at most $2\epsilon n k$.
	Let $S \subseteq [n]$ be the set of size at least $(1-\epsilon)n$ such that $\val(\mathcal{G}[S]) = 1$, and let $\sigma_S : S \mapsto [k]$ be the corresponding labeling which fully satisfies $\mathcal{G}[S]$. Let $S' = \{(a,\sigma_S(a)) | a \in S\}$ denote the corresponding non-vertex expanding set in the label extended graph $G$. We set the $\hat{x}_v$ variables as follows:
	\[
	\widehat{x}_{v}= 
	\begin{cases}
	\mathbbm{1}\left((a,i) \in S' \right)  & \text{if } v = (a,i) \mbox{ for some } a \in [n],\ i \in [k]\\
	\mathbbm{1}\left( (a,i), (b,j) \in S' \right) 
			& \text{if }  v = ((a,i),(b,j)) \mbox{ for some } a,b \in [n],\ i,j \in [k]
	\end{cases}
	\]
	We argue that $(\hat{x}_v)_{v \in V_H}$ forms a feasible solution to the StrongUG-SDP. 
	To begin with, for every $\ug$-vertex $a \in [n]$, we have 
	\[
	\sum_{i \in [k]} \|x_{a,i}\|^2 
	= \sum_{i \in [k]} \mathbbm{1}\left(\left\{\sigma_S(a) = i\right\} \wedge \left\{a \in S\right\}\right) 
	= \mathbbm{1}\left(\left\{a \in S\right\}\right) \leq 1.
	\] 
	This satisfies Constraint \ref{eq:am1}. Constraint \ref{eq:c1} also holds since
	\[
	\sum_{a \in [n]}\sum_{i \in [k]} \|x_{a,i}\|^2 
	= \sum_{a \in [n]}\mathbbm{1}\left(\left\{a \in S\right\}\right) 
	= |S| \geq n(1-\epsilon).
	\]
	Furthermore, since for every vertex $v \in V$, $\hat{x}_v \in \{0,1\}$, the solution trivially satisfies the $\ell^2_2$-triangle inequalities (Constraint \ref{eq:l22}) and Cauchy-Schwarz constraints (Constraints \ref{eq:cs}). Finally, for every $a \in [n]$ at most one of the $\hat{x}_{a,i}$ variables are set to $1$, the orthogonality constraints (Constraint \ref{eq:ort}) are also satisfied. 
Therefore, $(\hat{x}_v)_{v \in V_H}$ is a feasible SDP solution.

	Now we claim that for $(\hat{x}_v)_{v \in V_H}$, the objective value is at most $2\epsilon n k$. This is an immediate consequence of the fact that the vector $(\hat{x}_v)_{v \in V}$ is precisely the indicator of the non-hyperedge expanding small set that is given by the sequence of reductions from Lemma \ref{lem:transform}.
To begin with, we claim the following identities are satisfied by the $(\hat{x}_v)_{v \in V}$ vectors.
\begin{claim}
\label{cl:x-identity}
\begin{enumerate}
	\item For any vertex $(a,i) \in V$ where $a \in [n],\ i \in [k]$, we have 
\[
	\max_{u,v \in \cE((a,i))} \|\hat{x}_u - \hat{x}_v\|^2 
		= \mathbbm{1}\left(\{(a,i) \in \partial^{V}_G(S')\}\right).
	\]
\item For any edge vertex $v' = \{(a,i),(b,j)\}$ we have 
	\[
	\max_{u,v \in e\{(a,i),(b,j)\}} \|\hat{x}_u - \hat{x}_v\|^2 
		= \mathbbm{1}\left(\{(a,i) \in \partial^{E}_G(S')\}\right).
	\]
\end{enumerate}
\end{claim}
\begin{proof}
For the first point observe that for any $(a,i) \in V$, by construction we have   
\begin{align*}
	\max_{u,v \in \cE((a,i))} \|\hat{x}_u - \hat{x}_v\|^2 
& = \mathbbm{1}\left(\left\{(a,i) \notin S'\right\} \wedge \left\{ N_{G}((a,i)) \cap S' \neq \emptyset\right\} \right) \\ 
& \qquad +  \mathbbm{1}\left(\left\{(a,i) \in S'\right\} \wedge \left\{ N_{G}((a,i)) \cap (V \setminus S') \neq \emptyset\right\} \right) \\
	& = \mathbbm{1}\left((a,i) \in \partial_{V}(S')\right).
\end{align*}
	Similarly, for any edge vertex $((a,i),(b,j))$ we have
\begin{align*}
\max_{u,v \in \cE(((a,i),(b,j)))} \|\hat{x}_u - \hat{x}_v \|^2 
& = \mathbbm{1}\left(\{(a,i) \in S',(b,j)\notin S'\} \vee \{(a,i) \in S',(b,j)\notin S'\}\right) \\
& = \mathbbm{1}\left(\{(a,i),(b,j)\} \in \partial_{E}(S')\right).
\end{align*}
	\end{proof}

	We can bound the objective evaluated with the $\hat{x}_v$ vectors as follows.
	\begin{align*}
	&\sum_{e \in E'} w(e) \max_{g,h \in e} \|\hat{x}_g - \hat{x}_h \|^2  \\
	&= \sum_{v \in V} w({\cE(v')}) \max_{g,h \in \cE(v)} \|\hat{x}_g - \hat{x}_h\|^2  +  \sum_{e \in E} w(e) \max_{g,h \in \cE(e)} \|\hat{x}_g - \hat{x}_h\|^2 \\
		&= \sum_{v \in V} w(v) \mathbbm{1}\left(\left\{v \in \partial^{V}(S') \right\}\right)  +  \sum_{e\in E} w({\cE(e)}) \mathbbm{1}\left(\left\{e \in \partial^{E}(S')\right\}\right) & \textrm{(Using Claim \ref{cl:x-identity})}\\
	&\leq |\partial^{V}(S')| +   \sum_{u \in V} \frac{1}{{\rm deg}_{G}(u)} \sum_{v \in N_{G}(u)} \mathbbm{1}\left(\left\{(u,v) \in \partial^{E}(S')\right\}\right) \\
	&\leq 2|\partial^{V}(S')| \leq 2 \epsilon n k .
	\end{align*}
	This gives us the required bound on the SDP objective.
	\end{proof}
	
We observe that in any iteration $t$ of the while loop, every vertex deleted can be charged to hyperedge crossing 
the cut induced $\Gamma^{(t)}$. This along with the properties of the hypergraph orthogonal separators will give us 
the bound on the expected number of vertices deleted. This is formally stated in the following lemma.
\begin{restatable}{lem}{Vbound}  \label{lem:V-bound}
	The expected number of vertices discarded during all the iterations of the while loop is at most
	\[
	\rho(\epsilon,n,k) = O\left(\epsilon n k^2\log k \log\log k \sqrt{\log n }\right).
	\]
\end{restatable}
We defer the proof of the above lemma to Section \ref{sec:Vbound}. 

\subsubsection*{Putting things together}

\begin{proof}[Proof of  Theorem \ref{thm:partial-ug-1}]
Lemma \ref{lem:Vf-sat} establishes that the \uniquegames~instance induced $\cG[V_{\rm fin}]$ is completely
satisfiable using the labeling $\sigma_{\rm fin}: V_{\rm fin} \mapsto [k]$ returned by the algorithm. 
Lemma \ref{lem:V-bound} establishes that the expected number of vertices discarded by Algorithm \ref{alg:ver}
is at most ${O}\left(\epsilon n k^2 \log k\log\log k\sqrt{\log n}\right)$. 
Finally, Lemma \ref{lem:orth-sep-prop} will show that the ``while loop'' terminates in 
$O\left(\frac1\alpha\log n\right)$ iterations w.h.p. This finishes the proof of Theorem \ref{thm:partial-ug-1}.
\end{proof}

\subsection{Analysis of orthogonal separator}		
\label{sec:aux}
In this subsection, we present our calculations related to the use of the hypergraph orthogonal separator.
		
\subsubsection{Termination}
The following lemma which lower bounds the probability of a vertex being picked uniquely in any iteration;
it is identical to the corresponding lemma in \cite{CMM06}, we reproduce \cite{CMM06}'s proof here for completeness.
\begin{lemma}[\cite{CMM06}]			
\label{lem:unique}
	Let $V^{(t)}_H := \left\{v \in V_{\cG} | \exists i \in [k] \mbox{ s.t. } (v,i) \in H^{(t)} \right\}$ denote the set of Unique Game vertices in the hypergraph $H^{(t)}$ at the beginning of iteration $t$. Then for any $u \in V^{(t)}$ we have
	\[
	\Pr_{\Gamma^{(t)}} \left[|\cC_u \cap \Gamma^{(t)}| = 1\right] \geq \frac{\alpha}{4}.
	\]
\end{lemma}
\begin{proof}
	By definition, we can write
	\begin{eqnarray*}
		\Pr_{\Gamma^{(t)}} \left[|\cC_u \cap \Gamma^{(t)}| = 1\right]  &=& \sum_{i \in [k]}\Pr\Big[(u,i) \in \Gamma^{(t)} \wedge \not\exists j \neq i \mbox{ s.t. } (u,j) \in \Gamma^{(t)}\Big] \\
		&\geq& \sum_{i \in [k]}\Pr\Big[(u,i) \in \Gamma^{(t)}\Big] - \sum_{j \neq i}\Pr\Big[(u,i),(u,j) \in \Gamma^{(t)}\Big] \\
		&\overset{1}{=}& \alpha \sum_{i \in [k]} \|x_{(u,i)}\|^2 - \sum_{i \in [k]}\sum_{j \neq i}\Pr\Big[(u,i),(u,j) \in \Gamma^{(t)}\Big] \\
		&\overset{2}{\geq}& \frac{\alpha}{2} \sum_{i \in [k]} \|x_{(u,i)}\|^2 - \alpha\sum_{i \in [k]}\sum_{j \neq i} \min\Big(\|x_{(u,i)}\|^2,\|x_{(u,j)}\|^2\Big) \\
		&{\geq}& \frac{\alpha}{2} - \alpha\sum_{i \in [k]}\sum_{j \neq i} \frac{\|x_{(u,i)}\|^2 + \|x_{(u,j)}\|^2}{2} \\
		&\overset{3}{\geq}& \frac{\alpha}{2} - \frac{\alpha k}{2m} \geq \frac{\alpha}{4}		\label{eqn:uniq2}
	\end{eqnarray*}
	Here step $1$ uses the first property of the orthogonal separator, step $2$ uses the fact that since vertex $u \in V^{(t)}_H$, we must have $\sum_{i \in [k]}\|x_{u,i}\|^2  \geq \frac12$, and the inequality in step $3$ follows from our choice of parameter $m$.
\end{proof}

The following lemma gives a high probability upper bound for the number of iterations performed by the ``while loop'';
it is identical to the corresponding lemma in \cite{CMM06}, we reproduce \cite{CMM06}'s proof here for completeness.
\begin{lemma}[\cite{CMM06}]
\label{lem:orth-sep-prop}
	With probability at least $1 - n^{-O(1)}$, the while loop terminates in $T = (1/\alpha)\log(10n)$ iterations.
\end{lemma}
\begin{proof}
	We begin by observing that whenever a vertex $v$ is uniquely picked with some label $i$, it is always removed from the hypergraph $H$ (by either adding it to the set $S$ or deleting it from the hypergraph $H$). Let $\Theta_{v,t}$ denote the event that the vertex $v$ is not uniquely picked in iteration $t$. From lemma \ref{lem:unique} we know that $\Pr_{\Gamma^{(t)}}\Big[\Theta_{v,t}\Big] \leq 1 - \frac{\alpha}{4}$. We say that a vertex $v$ is undecided at iteration $T$ if $v \in {\rm Vert}(H^{(t)})$ i.e., it has not been deleted of included in the set of candidate vertices at least once. Then by definition we have 
	\begin{eqnarray*}
		\Pr_{\left(\Gamma^{(t)}\right)^T_{t = 1}} \left[\exists v \in [n] : \{v~\mbox{ is undecided} \}\right]
		&\leq& \sum_{a \in [n]} \Pr_{\left(\Gamma^{(t)}\right)^T_{t = 1}} \left[ \{v \mbox{ is undecided} \}\right] \\
		&\leq& \sum_{a \in [n]} \Pr_{\left(\Gamma^{(t)}\right)^T_{t = 1}} \left[ \forall t \in [T], \Theta_{v,t}\right] \\
		&\leq& \sum_{a \in [n]} \prod_{t \in [T]}\left(1 - \frac{\alpha}{4}\right) \\
		&\leq& n \left(1 - \frac{\alpha}{4}\right)^T \\
		&\leq& n^{-O(1)} 
	\end{eqnarray*} 
	where the last inequality follows from our choice of $T$. 
\end{proof}

\subsubsection{Analysis of the ``while'' loop}		
\label{sec:Vbound}

We now prove Lemma \ref{lem:V-bound}. 
We remark that a similar argument was used in \cite{CMM06} to bound the number of edges violated. 
However, since we have to ensure that the candidate set of vertices induce fully satisfiable \uniquegames~instances, 
the success and failure events tracked by us are quite different from theirs.

\Vbound*
\begin{proof}
	Fix an iteration $t$. Then for every vertex $u \in V_\cG$, we define a success event ${\rm Succ}(u)$ and a failure event ${\rm Fail}(u)$ as follows. 
	\[
	{\rm Succ}(u) \defeq \left\{ |\cC_u \cap \Gamma^{(t)}| = 1 \right\} \bigwedge \left\{ \wedge_{i \in [k]}\cE(u,i) \notin \partial_{H^{(t)}}(\Gamma^{(t)}) \right\}
	\]
	and a failure event 
	\[
	{\rm Fail}(u) \defeq \left\{ \vee_{i \in [k]} \cE(u,i) \in \partial_{H^{(t)}}(\Gamma^{(t)}) \right\}.
	\]
	Furthermore, we define the quantity $\Delta(u) \defeq D_{\ell^2_2}\ \sum_{i \in [k]}\max_{a,b \in \cE(u,i)} \|x_{a} - x_{b}\|^2$. We assume without loss of generality that $\Delta(u) \leq 1/8$ for every $u \in V_\cG$ (for the vertices with $\Delta(u) \geq 1/8$, we can delete them as is, and it would have been ``paid for'' by the objective function by a losing a factor of $8$). Then we can upper bound the probability of failure as 
	\begin{eqnarray*}
		\Pr_{\{\Gamma^{(t)}\}} \left[\left\{ \vee_{i \in [k]}\cE(u,i) \notin \partial_{H^{(t)}}(\Gamma^{(t)}) \right\}\right] 
	&\overset{1}{\leq}& \sum_{i \in [k]} \Pr_{\Gamma^{(t)}}\left[\cE(u,i) \in \partial_{H^{(t)}}(\Gamma^{(t)})\right] \\
	&\leq& \alpha \sum_{i \in [k]} D(\ell^2_2) \max_{a,b \in \cE(u,i)}| \|x_a - x_b \|^2\\
	&=& \alpha \Delta(u)
	\end{eqnarray*}
	where inequality $1$ uses the first property of the orthogonal separator, and the last equality is by definition. On the other hand, we can lower bound the success probability by 
	\begin{align*}
	\Pr\left[{\rm Succ}(u)\right] 
	&\geq \Pr_{\Gamma^{(t)}} \left[\left\{ |\cC_u \cap \Gamma^{(t)}| = 1 \right\}\right] 
	- \Pr_{\Gamma^{(t)}}\left[\left\{ \vee_{i \in [k]} \cE(u,i) \notin \partial_{H^{(t)}}(\Gamma^{(t)}) \right\}\right] \\
	& \geq \Pr_{\Gamma^{(t)}} \left[|\cC_u \cap \Gamma^{(t)}| = 1\right] - \alpha \Delta(u) \\
		& \geq \frac{\alpha}{4}- \alpha \Delta(u) \qquad \qquad \textrm{(Using Lemma \ref{lem:unique}).}
	\end{align*}
	Therefore, in any iteration, we have
	\[
	\frac{\Pr\left[{\rm Fail}(u)\right]}{\Pr\left[{\rm Succ}(u)\right]} \leq \frac{\Delta(u)}{\frac14 - \Delta(u)} \leq 8\Delta(u).
	\]
	Since the above is true for any iteration $t$, then the probability of a vertex $u$ getting deleted over all the iterations is at most $8\Delta(u)$. Therefore, the expected fraction of vertices deleted during all the iterations of the while loop is at most
	\[
		\sum_{u \in V_\cG} 8 \Delta(u) 
		\leq 8 D_{\ell^2_2}\ \sum_{e \in E''} \max_{a,b \in e} \|x_a - x_b\|^2 
		\leq 8D_{\ell^2_2}\ \epsilon nk. 
	\]
	
	Plugging in the value of $D_{\ell^2_2}$ from Theorem \ref{thm:hos-1} gives us the desired bound.
	
\end{proof}

\section{Proof of Theorem \ref{thm:partial-ug-2}} \label{sec:thm2}
	The algorithm for Theorem \ref{thm:partial-ug-2} is almost identical to that of Theorem \ref{thm:partial-ug-1}. The only difference is that the VertexExpansionRound algorithm (Algorithm \ref{alg:ver}) uses the $\ell^2_2$-$\ell_2$ hypergraph orthogonal separator (defined in Definition \ref{def:hos-2}) instead of the $\ell^2_2$ hypergraph orthogonal separator. As before, the arguments from Section \ref{sec:soln-sat} apply as is, and show that the set $V$ returned by the algorithm satisfies $\sval(\cG[V]) = 1$. Again as before, the thresholding step will remove at most $2\epsilon n $ vertices. 
	
	The only difference is in the analysis of the while loop where for every vertex $u$, we define $\Delta(u)$ to be 
	\[
	\Delta(u) \defeq D_{\ell^2_2} \sum_{i \in [k]} \max_{a,b \in \cE(u,i)} \|x_a - x_b\|^2 + \sum_{i \in [k]} D_{\ell_2}(|\cE(u,i)|) \min_{w \in \cE(u,i)}\|x_w\| \min_{a,b \in \cE(u,i) }\|x_a - x_b\|
	\]
	Again by an identical sequence of arguments, we will be bound the expected fraction of vertices deleted by $\sum_{u \in V_{\cG}} \Delta(u)$, which in turn can be bounded by Lemma \ref{lem:cut-bound} as 
	\[
		\sum_{u \in V_{\cG}} \Delta(u) \leq O(\epsilon n k^2\log k\log\log k) + O\left(\sqrt{\eta^H_{\rm max}\epsilon n k^2 \log k \log\log k} \right) . \]
	Furthermore, again from Lemma \ref{lem:cut-bound} we get that $\eta^H_{\rm max} \leq 2\log d_{\rm max}(\cG)$. Finally, using identical arguments as before, we can show that the while loop terminates in $O((1/\alpha)\log n)$ iterations. This completes the proof of Theorem \ref{thm:partial-ug-2}

\subsubsection*{Bounding the expected number of hyperedges cut}		

The following lemma  bounds the expected fraction of hyperedges cut by an $\ell^2_2$-$\ell_2$ hypergraph orthogonal separator, 

\begin{lemma}			\label{lem:cut-bound}
	Let $H = (V_H,E_H)$ be a hypergraph with arity at most $r$. Furthermore, let $\{x_v\}_{v \in V_H}$ be a set of vectors satisfying the $\ell^2_2$-triangle inequality. 
	Let $\Gamma$ be an $\ell^2_2$ - $\ell_2$ $(\alpha,m,\beta)$-hypergraph orthogonal separator with $\beta = 1$.
		Then, 
		\[
		\sum_{e \in E_H} \Pr_\Gamma\left[e \in \partial_H(\Gamma)\right] \leq \alpha m \sum_{e \in E_H} \max_{a,b \in e}\|x_{a} - x_b\|^2 + \alpha m \log m \log \log m\sqrt{\eta^{\max}_H   \sum_{e \in E_H} \max_{a,b} \|x_a - x_b\|^2} 
		\]
		In particular, when $H$ is the label extended hypergraph of a \stronguniquegames~instance $\cG$, we can take $\eta^H_{\rm max} \leq 2 \log d_{\rm max}(\cG)$, where $d_{{\rm max}}(\cG)$ is the maximum degree of any vertex in $\cG$.
\end{lemma}	 

\begin{proof}
	The proof of the first point is implicit in the proof of Theorem 6.3 from \cite{lm16}. The bound on $\eta^H_{\rm max}$ uses the following observation. 
	\begin{claim}[Claim 6.2~\cite{lm16}]
		Given a hypergraph $H = (V,E)$, suppose there is way to chose one vertex per hyperedge such that no vertex is chosen more than once. Then $\eta^H_{\rm max} \leq \max_{e \in H} \log |e|$
	\end{claim}	
	Now for the the hypergraph $H$ constructed from the \stronguniquegames~instances, we know that every hyperedge can be uniquely identified from a vertex $v \in V$. Therefore, 
	\begin{eqnarray*}
	\eta^H_{\rm max} \leq \max_{v \in V'} \log |\cE(v)| = \max_{v \in V_{\sym}}{\rm deg}_{G_\sym}(v) + 1 
	&=& \max_{v \in V}{\rm deg}_{G_\sym}(v) + 1 \\
	&=& \max_{v \in V}{\rm deg}_{G}(v) + 1 \leq 2d_{\rm max}(\cG).
	\end{eqnarray*}
\end{proof}

\part{Hardness of Approximation}

\newcommand{\meas}[1]{\mu\left(#1\right)}
\newcommand{\Inf}[2]{{\sf Inf}_{#1}\left[#2\right]}
\newcommand{\tx}{\tilde{x}}
\newcommand{\oct}{\textsc{OddCycleTransversal}}

In this section we shall prove our hardness results for \strug. We recall the theorem for convenience.

\hardness*

Theorem \ref{thm:strong-ug-informal} follows by combining Conjecture \ref{conj:ug} with the following theorem.

\begin{theorem}			\label{thm:strug-red}
	There exist constants $C,C',C'' \in \mathbbm{R}^+$ such that the following holds. 
	Let $\epsilon \in (0,1), k \in \mathbbm{N}$ and $d \geq C'' \epsilon^{-2} \log k$ be fixed constants such that $\sqrt{\epsilon \log d \log k} \leq 1/100 C$. Let $\epsilon_c = \epsilon/d$ and $\epsilon_s = (\epsilon/k)^{C'd/\epsilon^2}$ for any $\epsilon$ such that $\epsilon_c,\epsilon_s \leq \epsilon_0$ (where $\epsilon_0$ is from Conjecture \ref{conj:ug}). Then, given a $(1-\epsilon_c,\epsilon_s)$-\ug~instance $\cG_0$, there exists a polynomial time reduction to a \strug~instance $\cG$ on alphabet $[k]$ and max degree bounded by $d$ such that the following conditions are satisfied
	\begin{itemize}
		\item If $\cG_0$ is a YES instance, then $\cV(\cG) \geq 1 - 4\epsilon$.
		\item If $\cG_0$ is a NO instance, then $\cV(\cG) \leq 1 - C\sqrt{\epsilon \log k \log d}$.
	\end{itemize}
	Furthermore, the constraints in $\cG$ are linear.
\end{theorem}

In the following sections we proceed to prove the above theorem.

\section{Preliminaries}

\subsection{Problem Definitions} 
 
Now we formally define the several intermediate problems that will be used in the sequence of reduction used in the proof of Theorem \ref{thm:strug-red}. We begin by defining the \hug~ problem which is as follows.

\begin{problem}[$(\alpha,\beta)$\hug~]
	\label{def:abhug}
	An instance $\cH(V,E,[k],\{\pi_{e,v}\}_{v \in V,e \in E},\cD_H)$ of \hug~ consists of $d$-ary constraints on the vertex set $V$. Each hyperedge constraint $e (v_{i_1},v_{i_2},\ldots,v_{i_d}) \in E$ is identified with $d$-bijections $\pi_{e,v_1},\pi_{e,v_2},\ldots,\pi_{e,v_d}: [k] \to [k]$. A labeling $\sigma:V \to [k]$ is said to {\em satisfy} the hyperedge $e$ if for every $i,j \in [d]$ we have $\pi_{e,v_i}(\sigma(v_i)) = \pi_{e,v_j}(\sigma(v_j))$.
The goal of this problem is to compute an assignment that maximizes the fraction of the satisfied hyperedges;

Now given parameters $\alpha,\beta \in [0,1]$, with $\alpha \geq \beta$, the objective of $(\alpha,\beta)$-\hug~ is to distinguish between the following two cases:
	\begin{itemize}
		\item [YES]: There exists a labeling $\sigma:V \to [k]$ such that 
		\[
		\Pr_{e \sim \cD_H} \left[ \sigma \mbox{ satisifies } e\right] \geq \alpha
		\]		
		\item [NO]: For every labeling $\sigma:V \to [k]$ 
		\[
		\Pr_{e \sim \cD_H} \left[ \sigma \mbox{ satisifies } e\right] \leq \beta
		\] 
	\end{itemize} 	
\end{problem}

 Next we define bipartite versions of the \strug~problem.
 
 \begin{problem}[$(\alpha,\beta)$-\sbug~]
 	An instance $G_{SB}(V_L,V_R,E,[k],\{\pi_{vu}\}_{(u,v) \in E},\mu)$ of \sbug~ consists of bipartite multi-graph with left and right vertex sets $V_L,V_R$ and edge set $E \subseteq V_L\times V_R$. Here every edge $(u,v) \in E$ is associated with a bijection $\pi_{v \to u}: [k] \to [k]$. Given a labeling $\sigma: V_L \cup V_R \to [k]$, we say that the labeling satisfies an edge $(u,v) \in E$ if $\sigma(u) = \pi_{v \to u}(\sigma(v))$. Furthermore, $\mu:V_L \mapsto \mathbbm{R}_{+}$ is a measure on left vertices.
 	
 	Given parameters $\alpha,\beta \in [0,1]$, with $\alpha \geq \beta$, the objective of $(\alpha,\beta)$-\hug~ is to distinguish between the following two cases:
 		\begin{itemize}
 			\item [YES]: There exists a labeling $\sigma:U \cup V \to [k]$ such that 
 			\[
 			\Pr_{u \sim \mu} \left[ \forall v \in N_{G_{SB}}(u), \quad \pi_{v \to u}(\sigma(v)) = \sigma(u)\right] \geq \alpha
 			\]
 			\item [NO]: Any labeling $\sigma:V_L \cup V_R \to [k]$ satisfies
 			\[
 			\Pr_{u \sim \mu} \left[ \forall v \in N_{G_{SB}}(u), \quad \pi_{v \to u}(\sigma(v)) = \sigma(u)\right] < \beta
 			\]
 		\end{itemize} 	
 	 
 	Finally, for any integer $D \in \mathbbm{N}$, we shall say that $G_{SB}$ is a $D$-\sbug~ instance if vertices in $V_R$ have maximum degree at most $d$.
 \end{problem}

\subsection{Fourier Analysis over $V^n_M$}			\label{sec:fourier}
	
	Let $G = (V_M,E)$ be a graph on $m$ vertices with weighted adjacency matrix $A \in \mathbbm{R}^{m \times m}$.  Consider the space of functions $f:V_M \to \mathbbm{R}$. We can equip the space of such functions with the inner product 
	\[
	\langle f,g \rangle_M \defeq \Ex_{x \sim \mu_M}\left[f(x) g(x) \right] 
	\]
	where $\mu_M$ is the stationary measure for the random walk on $G$. Then one can verify that the space of functions $f:V_M^n \to \mathbbm{R}$ forms a vector space under the usual addition and multiplication of functions, and therefore admits an orthonormal basis with respect to the inner product $\langle \cdot, \cdot \rangle_M$. In particular, we are interested in the basis given by the eigenfunctions of the Markov chain $M$, as stated formally below.
	
	\begin{theorem}[Theorem 4.4~\cite{OD13}]
	Let $D \in \mathbbm{R}^{m \times m}$ be the diagonal matrix whose diagonal entries are given by the vertex weights i.e, $D_{ii} = \sum_{j \in [m]} A(i,j)$. Let $A_M = D^{-1}A$ be the row-normalized transition probability matrix of the Markov chain $M$ corresponding to the random walk on $G$. Then there exist right eigenvectors $\chi_1,\chi_2,\ldots,\chi_m$ with eigenvalues $1 = \lambda_1 \geq \lambda_2 \cdots \geq \lambda_m$ such that they form an orthonormal basis for functions $f:V_M \mapsto \mathbbm{R}$, where $\chi_1$ can be taken to be the all ones vector. In particular, every function $f:V_M \mapsto \mathbbm{R}$ admits a Fourier decomposition of the form
	\[
	f = \sum_{v \in V_M} \hat{f}(v)\chi_v
	\]
	where $\hat{f}(v)$ are the fourier coefficients determined by the inner product $\langle \cdot, \cdot \rangle_M$.
	\end{theorem}
	One can easily extend the above to functions over the $n$-fold product Markov chain $M^n$ by defining the inner product between functions $f,g:V^n_M \to \mathbbm{R}$ as $\langle f,g \rangle_M = \Ex_{x \sim \mu^n_M}[ f(x)g(x)]$. From this we can build a Fourier basis for functions $f:V_M^n \to [0,1]$ as $\left\{ \chi^{(S)} = \otimes_{i \in S} \chi_{S_i} : S \in V_M^n \right\}$.	Hence, any function $f: V_M^n \to \mathbbm{R}$ admits a multilinear polynomial decomposition of the form
	\[
	f = \sum_{S \in V_M^n} \hat{f}(S) \chi^{(S)}.
	\] 
	Finally, we can extend the above to the space of functions $f:V_M^n \to [k]$ as follows. We shall relax the range of the functions to be the $(k-1)$-simplex and interpret the functions as $f: V_M^n \to \Delta_k$. Given that the functions are now real vector valued, we can equip the space of such functions with the inner product 
	\[
	\langle f,g \rangle_M \defeq \Ex_{x \sim \mu^n_M}\left[\langle f(x),g(x) \rangle \right] = \sum_{i \in [k]}\Ex_{x \sim \mu^n_M} \left[f^i(x)g^i(x) \right]
	\] 
	where $f^i:V_M^n \to \mathbbm{R}_+$ is the $i^{th}$-coordinate function for the vector valued function $f$. Using this, the Fourier analytic concepts defined above can be extended to functions with the domain as $[k]$.

	The following proposition is well known.
	\begin{proposition}\cite{OD14}			\label{prop:fourier}
		Let $f,g :V_M^n \to \mathbbm{R}$ be a pair of real valued functions over $\mathbbm{R}^n$. Then
\begin{itemize}
	\item {\it Plancharel's Identity}: $\langle f, g \rangle_M = \sum_{S \in V_M^n} \hat{f}(S)\hat{g}(S)$
	\item {\it Influence}: The influence of the $i^{th}$ coordinate is expressed in terms of the Fourier coefficients as 
	\[ \Inf{i}{f}	= \sum_{S : \chi_{S_i} \neq {\bf 1}} \hat{f}(S)^2  \]
	Using the above expression, we can also extend the notion of influences for vector valued functions $F: V^n_M \mapsto \Delta_k$ by defining
		\[ \Inf{i}{F}  = \sum_{j \in [k]} \Inf{i}{F^{j}}. \]
\end{itemize}
\end{proposition}

The following standard fact will be helpful in our analysis.
	
\begin{fact}
\label{fact:covar}
	Let $M$ be a Markov chain on $k$ vertices with spectral gap $\epsilon$. Consider the distribution over random vertices $(x,y_1,y_2,\ldots,y_d)$ where $x \sim \mu_M^n$ and $y_1,\ldots,y_d \sim \mu^n_M(x)$ where $\mu^n_M(x)$ is the distribution over random neighbors of vertex $x$. Let $Y = (Y_1,Y_2,\ldots,Y_n)$ be the independent ensemble of random variables with $Y_i := \{\chi_r(y_{j}(i))\}_{r \in [k],j \in [d]}$ with every $i \in [n]$. Then the following properties hold.
	\begin{itemize}
		\item[1.] {\bf Matching Moments}. There exists a corresponding set of Gaussian random variables $Z = (Z_1,Z_2,\ldots,Z_n)$ with matching covariance structure i.e., for any $i,i' \in [n]$ we have $\Ex_Y [Y_iY^\top_{i'}] = \Ex_Z [Z_iZ^\top_{i'}]$. In particular the random variables in $Z_i$ are identified as $(Z(y_{rj}(i)))_{r \in [k],j \in [d]}$.
		\item[2.] {\bf Row-wise Structure}.	Furthermore, there exists a vector $\overline{\rho} \in [-1,1]^{kn}$ such that the random variable  
		$Z_{y_j} = (Z(y_{rj}(i)))_{r \in [k], i \in [n]}$ are $\overline{\rho}$-correlated copies of a $kn$-dimensional standard Gaussian vector $Z_x$.
		\item[3.] {\bf Bounded Correlations}. Moreover, $\|\overline{\rho}\|_\infty \leq 1 - \epsilon$, and if the transition probability matrix of $M$ is positive semidefinite, then $\overline{\rho}$ is nonnegative.	
	\end{itemize}
\end{fact}

\begin{proof}
Let $\lambda_1,\ldots,\lambda_{k}$ be the eigenvalues corresponding to the eigenfunctions $\chi_1,\ldots,\chi_k$ of $M$ (i.e, they are the eigenvalues of the row normalized transition probability matrix $A_M$). Since for any $x \in V_M^n$ and $S \in V_M^n$ we have $\chi^S(x) = \prod_{i \in S} \chi_{S_i}(x(i))$ and under the distribution $x(i)$'s are independent for different choices of $i \in [n]$ (and similarly for the $y'_j$'s as well), it suffices to prove the statement for $n = 1$. 

Let $\chi_1,\ldots,\chi_k$ denote the Fourier basis for functions $f: V_M \to \mathbbm{R}$ (given by the corresponding right eigenvectors of $A_M$). We begin by analyzing the correlations between different types of variables.
		\begin{enumerate}
			\item[(a)] For any $i \neq j$ and $x,y_a,y_b$ we have $\Ex_{x \sim \mu_M} [\chi_i(x)\chi_j(x)] = \Ex_{x \sim \mu_M}\Ex_{y_a \sim \mu_M(x)} [\chi_i(y_a) \chi_j(y_a)] = \langle \chi_i, \chi_j \rangle_M = 0$ because $\chi_i,\chi_j$ are orthogonal.  Furthermore,
			\[
			\Ex_{x \sim \mu_M}\Ex_{y_b \sim \mu_M(x)} [\chi_i(x) \chi_j(y_b)] = \langle \chi_i, A_M \chi_j \rangle = \lambda_i \langle \chi_i, \chi_j \rangle = 0
			\]
			using Proposition \ref{prop:eig-facts}. Similarly we also have $\Ex_{x}\Ex_{y_b,y_{b'}} [\chi_i(y_b) \chi_j(y_{b'})]= 0$.
			\item[(b)] For any $i \in [k]$ and $a \in [d]$ we have $\Ex [\chi_i(x)^2] = \Ex [\chi_i(y_a)^2] = \|\chi_i\|^2_M = 1$.
			\item[(c)] For any $i \in [k]$ and $a \in [d]$ we have $\Ex [\chi_i(x) \chi_i(y_a)] = \lambda_i$ and $\Ex [\chi_i(y_a) \chi_i(y_b)] = \lambda_i^2$. We prove these identities in Proposition \ref{prop:eig-facts}. 
		\end{enumerate}
		Now for every $r \in [k]$, we define the following ensemble of Gaussian random variables.
		\begin{enumerate}
			\item We define $Z_{x}(r) \sim N(0,1)$ to be standard normal random variable.
			\item For every $j \in [d]$, we define $Z_{y_j}(r) = \lambda_r Z_{x}(r) + \sqrt{1 - \lambda^2_r}\xi_{r,j}$ where $\xi_{r,j} \sim N(0,1)$ is an independent standard normal for every $r \in [k],j \in [d]$.		
		\end{enumerate}
	Now for $r \in [k],j \in [d]$, we identify the random variable $Z_{y_{j}}(r)$ with $\chi_{r}(y_j)$. We quickly verify that with this indexing the $Z$-variables have matching covariance structure with the $Y$-variables. For any $r,r' \in [k]$ with $r \neq r'$, it is easy to see that $Z_{y_{j}}(r),Z_{y_{j'}}(r')$ are independent independent Gaussians and therefore they satisfy $(a)$. Since the marginal distribution of every $Z_{y_{j}}(r)$ random variable is $N(0,1)$, point $(b)$ is satisfied. And finally, for a fixed $r \in [k]$, the above construction identically captures the non-trivial $\lambda_r,\lambda^2_r$-correlations, therefore satisfying $(c)$.
	
	Finally, since the correlations between the Gaussian in the ensembles are products of the eigenvalues of $M$, we can establish that the cross-variable correlations are bounded by $\max_{i \geq 2} \lambda_i \leq 1 - \epsilon$. Furthermore, they are nonnegative whenever the transition probability matrix of $M$ is PSD.
	\end{proof}
	
\subsection{Noise Operators}

\begin{definition}[Markov Chain Noise Operator] 				\label{defn:mark-noise}
The noise operator $\Gamma_{1 - \eta}$ is defined on functions $f:V_M^n \to [k]$ as follows.
Given $x \in V_M^n$, we sample $x'\underset{1 - \eta}{\sim} x$ 
as follows. For every $i \in [n]$, we independently set $x'(i) = x(i)$ with probability $1 - \eta$, 
and with probability $\eta$, we sample $x'(i) \sim \mu_M(x(i))$. 
Then $\Gamma_{1 - \eta}f(x) \defeq \Ex_{x' \underset{1 - \eta}{\sim} x} f(x')$.
\end{definition}

The following facts about the Fourier decay properties of the noise operators are well known.

\begin{lemma}[Lemma 5.4, Lemma C.1 \cite{LRV13}]			\label{lem:four-decay}
Let $f: V^n_M \to [0,1]$ be a function defined on the Markov chain $M^n$, and let $Q$ be multilinear polynomial representation of $\Gamma_{1 - \eta}f$ in the Fourier basis over $V^n_M$. Let $\epsilon$ be the spectral gap of the Markov chain $M^n$. Then the following properties hold for every $\eta \in (0,1)$.
\begin{itemize}
	\item For every $p \geq 1$, we have ${\rm Var}(Q^{>p}) \leq (1 - \epsilon\eta)^{2p}$.
	\item The sum of influences is bounded i.e, $\sum_{i \in [n]} \Inf{i}{Q} \leq \frac{1}{\epsilon\eta}$.
\end{itemize}
\end{lemma}

We shall also need the notion of the noise operator on the Gaussian space.

\begin{definition}[Gaussian Noise Operator] For any $\overline{\rho} = (\rho_1,\rho_2,\ldots,\rho_n) \in [-1,1]^n$, we use $U_{\overline{\rho}}$ to denote the Gaussian noise operator with correlation vector $\overline{\rho}$. Formally, given $x \in \mathbbm{R}^n$, a $\overline{\rho}$-correlated Gaussian $y \underset{\overline{\rho}}{\sim} x$ is defined as follows. For every $i \in [n]$, we set $y(i) = \rho_i x(i) + \sqrt{1 - \rho^2_i} z(i)$ where $z \sim N(0,1)^n$ is an $n$-dimensional random Gaussian vector. Then $U_{\overline{\rho}} f (x) \defeq \Ex_{y \underset{\overline{\rho}}{\sim}x  }\big[f(y)\big]$.
\end{definition}

\section{Hardness of \hug~}

The main result of this section is the following theorem which gives arity dependent hardness for \hug.

\begin{theorem}					\label{thm:hyperug}
	Let $d,\epsilon,k$ be as in Theorem \ref{thm:strug-red}. Fix $\epsilon > 0$ and let $\epsilon_c = \epsilon/d, \epsilon_s = (\epsilon/k)^{C'd/\epsilon^2}$ such that $\epsilon_c,\epsilon_s \leq \epsilon_0$. Then there exists a polynomial time reduction from $(1-\epsilon_c,\epsilon_s)$-\ug~to $(1 - 4\epsilon,1 - C\sqrt{\epsilon \log d \log k})$-\hug, where $C,C' > 0$ are absolute constants.	
\end{theorem}

\subsection{Markov chain gadget and Long codes}
\label{sec:gadget}
We describe the basic Markov chain gadget and its underlying properties. Consider the graph defined on vertices $V_M = \{s_i,t_i\}_{i \in [k]}$ with adjacency matrix $A$, where the entires of $A$ are as follows. Let $G'$ be an unweighted regular expander with {\em constant} spectral gap $\alpha > 0$ of {\em constant} degree $g$ on the vertex set $\{t_i\}_{i \in [k]}$.
For $x,y \in V_M$, we have 
\[A(x,y) = 
\begin{cases}
1 - \epsilon & if x = y  = s_i \mbox{ for some } i \in [k] \\
\epsilon & if x = t_i, y \in \{s_i\} \cup N_{G'}(t_i)\\
(g+1)\epsilon & if x = y = t_i \mbox{ for some } i \in [k]
\end{cases}
\]
Let $M$ be a Markov chain with $V_M = \{s_i,t_i\}_{i \in [k]}$ whose transition probability matrix is given by weighted random walk on $M$. We now state some lemmas which summarize the properties of the Markov chain that will be used in the reduction.
\begin{claim}		\label{cl:M-prop}
	The Markov chain $M$ satisfies the following properties:
	\begin{itemize}
		\item[(i)] The row normalized transition probability matrix $A_M$ is PSD.
		\item[(ii)] Let $\mu_M:V_M \to [0,1]$ be the stationary distribution of $M$. Then for every $i \in [k]$, $\mu_M(s_i) = \frac{1}{k(1 + (2g+1)\epsilon)}$.
		\item[(iii)] The spectral gap of $A_M$ is at least $c \epsilon$ where $c = \min\{\alpha/24,1/10g\}$ .
	\end{itemize}
\end{claim}
\begin{proof}
	For any $i \in [k]$ we have 
	\[
	A(s_i,s_i) - \sum_{x \neq s_i} A(s_i,x_i) = 1 - 2\epsilon
	\]
	and 
	\[
	A(t_i,t_i) - \sum_{x \neq t_i} A(t_i,x) = (g+1)\epsilon - \epsilon - g\epsilon = 0
	\]
	Therefore using {\em Gershghorin's Circle Theorem} we know that $A$ is PSD. But then the weighted normalized matrix $ D^{-1/2} A D^{-1/2}$ (where $D$ is the diagonal matrix with vertex weights) is also PSD. Since eigenvalues of $A_M = D^{-1}A$ are the same as those of $D^{-1/2}AD^{-1/2}$, we have that $A_M$ is PSD.
	Furthermore,  the total weight of edges in the graph $G$ is $k + k(2g+1)\epsilon$. Therefore, for any $i \in [k]$ we have 
	\[
	\mu_M(s_i) = \frac{A(s_i,s_i) + A(s_i,t_i)}{k + k(2g+1)\epsilon} = \frac{1}{1 + (2g+1)\epsilon} 
	\]
	The proof of the spectral gap can be found in Section \ref{sec:spec-gap}.
\end{proof}

	Our long codes will be functions $f:V_M^n \to [k]$. We introduce some additional notation that will be useful. Define the map $\psi:V_M \to [k]$ as $\psi(x) = i$ is $x \in \{s_i,t_i\}$. Among functions $f:V_M^n \to [k]$, the $i^{th}$ dictator function, denoted by $\Lambda_i$ is defined as follows. For every $x \in [k]^n$, $\Lambda_i(x) = \psi(x_i)$. The follow lemma bounds the noise stability of dictators.

\begin{lemma}
	\label{lem:dict-ns}
For any $i \in [n]$, we have 
	\[ \Ex_{x \sim \mu_M^n} \Pr_{y_1,\ldots,y_d \sim \mu_M^n(x)} \left[\exists j \in [d]: \Lambda_i(y_j) \neq \Lambda_i(x)\right] \leq 2 \epsilon. \]
	\end{lemma}
	\begin{proof}
		For brevity, let $S = \{s_i\}_{i \in [k]}$ and $T = \{t_i\}_{i \in [k]}$. Now let $x \in V_M^n$ be such that $x_i = s_j$ for some $j \in [k]$. Then for any $y \sim \mu_M^n(x)$, we will also have $y_i  \in \{s_j,t_j\}$. Therefore, conditioned on $x_i \in S$, $\Lambda_i(y) = \Lambda_i(x)$ with probability $1$. Hence
\[
	\Ex_{x \sim \mu_M^n} \Pr_{y_1,\ldots,y_d \sim \mu_M^n(x)} \left[\exists j \in [d]: \Lambda_i(y_j) \neq \Lambda_i(x)\right] 
	\leq \Pr_{x \sim \mu_M^n} \left[ x_i \in T\right] \leq 2\epsilon .\]
	\end{proof}
	
\paragraph{Composition and Folding.} For any $x \in V^n_M$ and a permutation $\pi:[n] \to [n]$ we use $\pi \circ x := (x_{\pi(1)},x_{\pi(2)},\ldots,x_{\pi(n)})$ to denote the composition of $\pi$ with $x$. Furthermore, we assume the algebra $(\mathbbm{F}_{k},+)$ on the indices. 
For any $x \in V_M$ and $i \in [k]$ we define
	\[
	x \oplus i \defeq 
	\begin{cases}
	s_{j + i} & \mbox{if }  x = s_j\\
	t_{j + i} & \mbox{if }  x = t_j
	\end{cases}
	\] 
	and extend the above definition to $x \in V_M^n$.
	For any $x \in V_M^n$ and $i \in [k]$ we define
	$x \oplus i \defeq (x_1 \oplus i, \ldots, x_n \oplus i)$.
	\begin{definition}[Folding]
	\label{def:folding}
	Given a function $f:V_M^n \to [k]$, the folded function $\tilde{f}$ is defined as
	\[
	\tilde{f}(x) \defeq f(x \oplus (-\psi(x_1) + 1)) + \psi(x_1) - 1. 
	\]
	\end{definition}

	It is easy to see that $\tilde{f}$ constructed above is $\mathbbm{F}_k$ linear. To see this, fix a $r \in \mathbbm{F}_k$. Then,
	\begin{eqnarray*}
	\tilde{f}(x \oplus r) &=& f\left( (x  \oplus r ) \oplus (-\psi((x \oplus r)_1) + 1)\right) + \psi((x \oplus r)_1) - 1 \\
	&=& f\left( (x  \oplus r ) \oplus (-\psi((x_1 \oplus r)_1) + 1)\right) + \psi(x_1 \oplus r) - 1 \\
	&=& f\left( (x  \oplus r ) \oplus (-\psi(x_1) - r + 1)\right) + \psi(x_1) + r - 1 \\
	&=& f\left( (x ) \oplus (-\psi(x_1) + 1)\right) + \psi(x_1) - 1 + r \\
	&=& \tilde{f}(x) + r
	\end{eqnarray*}
	
	The following claim summarizes the key properties of the folding operation defined above.
		
	\begin{claim}					\label{cl:folded}
		The folded codes defined above satisfy the following properties.
		\begin{enumerate}
			\item For any function $f:V_M^n \to [k]$, for any $u \in [k]$ we have $\Pr_{x \sim \mu_M^n} \left[\tilde{f}(x) = u\right] = \frac1k$.
			\item The dictator functions are folded i.e., if $f = \Lambda_i$, then $f = \tilde{f}$.
		\end{enumerate}
	\end{claim}
	\begin{proof}
		From the $\mathbbm{F}_k$ linearity of $\tilde{f}$ (established above), we know that for any $x \in V_M^n$ and $r \in [k]$, we have $\tilde{f}(x  \oplus r) = \tilde{f}(x) + r$. Furthermore, note that $\mu^n_M(x \oplus r) = \mu^n_M(x)$. Therefore, the inverse set $\tilde{f}^{-1}(u)$ for $u \in [k]$ are cosets of the set $f^{-1}(1)$ under the operation $\oplus$ with identical probability masses. Since the masses of the pre-image sets add up to $1$, this establishes the first point.
		
		For the second point, fix a $f = \Lambda_i$ for some $i \in [n]$. Then for any $x$, we have 
		\[ \tilde{f}(x) = \Lambda_i(x \oplus (-\psi(x) \oplus {\bf 1})) + \psi(x) - 1 
		= \psi(x_i \oplus ( -\psi(x_i) + 1)) + \psi(x) - 1 
		= \psi(x_i)= f(x)
		\]

	\end{proof}

We shall need the following observation.

\begin{claim}				\label{cl:vert-exp}
For any dictator function $\Lambda_i$ with $i \in [n]$, and $x \in V_M^n$, $\Pr_{x' \underset{1- \eta}{\sim} x} \left[\Lambda_i(x) \neq \Lambda_i(x')\right] \leq \eta$.
\end{claim}
\begin{proof}
	Note that the $i^{th}$ dictator depends only on the $i^{th}$ coordinate and therefore,
\[ \Pr_{x' \underset{1- \eta}{\sim} x} \left[\Lambda_i(x) \neq \Lambda_i(x')\right] 
 = \Pr_{x' \underset{1- \eta}{\sim} x} \left[\psi(x_i) \neq \psi(x_i')\right] 
\leq \Pr_{x' \underset{1- \eta}{\sim} x} \left[x(i) \neq x'(i)\right] \leq \eta.
	\]
\end{proof}

\subsection{The Reduction}					
\label{sec:test}
Our reduction is a generalization of the $2$-query long code test from \cite{KKMO07} to the setting of hypergraph constraints, combined with the Markov operator defined in Section \ref{sec:gadget}. Given a \ug~instance $\cG(V_\cG,E_\cG,[k],\{\pi_{u \to v}\}_{(u,v) \in E_\cG})$ we construct a \hug~instance $\cH(V_\cH,E_\cH,[k],\{\pi_{e,u}\}_{e \in E_\cH, u \in e})$ as follows. The variable set $V_\cH$ is the set of long code tables $\{f_v\}_{v \in V_\cG}$. The distribution over hyperedge constraints in $\cH$ is described by the following dictatorship test.

\begin{figure}[h!]
	\begin{mdframed}
		Given long codes $\{f_v\}_{v \in V_\cG}$ for every vertex $v \in V_\cG$. \\
		{\bf Folding}: From $f:V_M^n \to [k]$ construct $\tilde{f}: V_M^n \to [k]$ as 
		in Definition \ref{def:folding}. 
	
		{\bf Test}: 
		\begin{enumerate}
			\item Sample a random vertex $v \sim V_\cG$ and neighbors $w_1,\ldots,w_d \sim N_{\cG}(v)$. Let $\tilde{f}_{v}, \tilde{f}_{w_1},\ldots,\tilde{f}_{w_d}:V_M^n \to [k]$ be the corresponding folded long codes.
			\item Sample $x \sim \mu_M^n$ and $y_1,\ldots,y_d \sim \mu_M^n(x)$.
			\item For every $i \in [d]$, sample $\tilde{y}_i \underset{1 - \eta}{\sim} y_i$.
			\item Accept if and only if for every $j,j' \in [d]$ we have 
			\[
			\tilde{f}_{w_j}(\pi_{w_j \to v} \circ \tilde{y}_j) = \tilde{f}_{w_j'}(\pi_{w_{j'} \to v} \circ \tilde{y}_{j'}) 
			\]
		\end{enumerate}
	\end{mdframed}
	\caption{PCP Verifier for Hypergraph Unique Games}
	\label{fig:pcptest-1-intro}
\end{figure}

\paragraph{Parameters of the reduction.}
Let $\xi_{k,\epsilon,d} := C\sqrt{\epsilon \log d \log k}$ for some constant $C>0$. Throughout this section, we shall be working with the following parameters:
\begin{itemize}
	\item $\eta = \frac{\epsilon}{100d}$.
	\item $\delta = \frac1k$.
	\item $\alpha = \left(\frac{\epsilon}{k}\right)\left(\epsilon\right)^d$.
	\item $\zeta \in \left(0,\xi^2_{k,\epsilon,d}\right)$.
	\item $\tau = (\delta\zeta/dk)^{8d^2\log (2k/\epsilon)/\epsilon^2}$.
	\item $\rho(\tau,\epsilon,\eta,\alpha) = \tau^{\epsilon\eta/\log(1/\alpha)} \leq \epsilon^2\delta^2\zeta^2/dk$.
\end{itemize}
and $d,\epsilon,k$ are fixed such that they satisfy the conditions of Theorem \ref{thm:strong-ug-informal}.

\subsection{Completeness}		\label{sec:comp}

Suppose there exists an assignment $\sigma:V_\cG \to [n]$ that satisfies $1 - \epsilon/d$ fraction of the constraints. For every $v \in V_\cG$, define $f_v \defeq \Lambda_{\sigma(v)}$. Note that the dictator functions are folded (Claim \ref{cl:folded}), therefore $\tilde{f}_{w_j} = f_{w_j}$ for every $j \in [d]$. 

\begin{claim}				\label{cl:comp}
	With probability at least $1 - \epsilon$, the constraints $(v,w_i) \ \forall i \in [d]$ are satisfied
	by $\sigma$, i.e. $\pi_{w_j \to v}(\sigma(w_j)) = \sigma(v)$. 
\end{claim}

\begin{proof}
For a vertex $v$, let $p_v$ denote the fraction of constraints incident on $v$ that
are satisfied by $\sigma$. Then $\E_v p_v \geq 1 - \epsilon/d$.

Since $w_1, \ldots, w_d \sim N_{\cG}(v)$ are chosen randomly, the probability that they
are all satisfied is $p_v^d$. Since $v$ is also chosen randomly, the probability that 
 $(v,w_i) \ \forall i \in [d]$ are satisfied by $\sigma$ is at least
\begin{align*}
	\E_v p_v^d & \geq \left( \E_v p_v \right)^d \geq \left(1 - \frac{\epsilon}{d} \right)^d
	\geq 1 - \frac{\epsilon}{d} \cdot d = 1 - \epsilon.
\end{align*}

\end{proof}

Therefore, for every $j \in [d]$, we have $
\pi_{w_j \to v} \circ f_{w_j} = \pi_{w_j \to v} \circ \Lambda_{\sigma(w_j)} = \Lambda_{\sigma(v)}$. Furthermore, we also have $\tilde{y}_j(u) = y_j(u)$ for every $j \in [d]$ with probability at least $1 - \eta d$ and from Lemma \ref{lem:dict-ns} we know that for any dictator function 
\[
\Ex_{x \sim \mu_M^n} \Pr_{y_1,\ldots,y_d \sim \mu_M^n(x)} \left[\exists j \in [d] : \Lambda_{\sigma(u)}(y_i) \neq \Lambda_{\sigma(u)}(x) \right] \leq 2 \epsilon .
\]
Therefore, the test accepts with probability at least $1 - 3\epsilon - d\eta \geq 1 - 4\epsilon$.

\subsection{Soundness}

In this section we prove the soundness guarantee of the PCP verifier as stated in the following theorem.

\begin{theorem}
\label{thm:sound}
	Suppose there exist long code tables $\{\tilde{f}_v\}_{v \in V}$ for which the test accepts with probability at least $1 - \xi_{k,\epsilon,d} + \xi$. Then there exists a labeling $\sigma:V(\cG) \to [n]$ which satisfies at least $(\epsilon/k)^{C'(1/d \epsilon^2)}$-fraction of constraints in $\cG$, where $C'$ is an absolute constant independent of other parameters. 
\end{theorem}

The rest of this section is dedicated towards proving the above theorem. For ease of notation, we will assume that the codes are folded and denote $f_v = \tilde{f}_v$. As discussed in Section \ref{sec:fourier}, for any vertex $v$, the corresponding long code $f_v$ can be interpreted as $f_v = (f^1_v,f^2_v,\ldots,f^k_v): V_M^n \to \Delta_k$ where $f^1_v,f^2_v,\ldots,f^k_v:V_M^n \to [0,1]$ are coordinate functions. In particular $f_v(x) = u$ for some $u \in [k]$ is equivalent to $(f^1_v,f^2_v,\ldots,f^k_v) = e_u$ where $e_u$ is the $u^{th}$ standard basis vector. With this interpretation, the probability of the test accepting can be expressed as 
\[ \Pr\left[\mbox{Test Accepts}\right] 
	 = \Ex_v \Ex_{w_1,\ldots,w_d \sim N_{\cG}(v)}\Ex_{x,\left(y_j\right)^d_{j = 1}}\Ex_{\tilde{y}_j \underset{1 - \eta}{\sim} y_j} \left[\sum_{i \in [k]}  \prod_{j \in [d]} f^i_{w_j}\left(\pi_{w_j \to v} \circ \tilde{y}_j\right) \right] \]
where $x \sim \mu^n_M$ and $y_1,\ldots,y_d \sim \mu^n_M(x)$. Now, using the independence in the choices of the $w_i$s, the $y_i$s and the $\tilde{y}_i$s,  
\begin{multline*}
\Ex_{\tilde{y}_j \underset{1 - \eta}{\sim} y_j} \left[\sum_{i \in [k]}  \prod_{j \in [d]} f^i_{w_j}\left(\pi_{w_j \to v} \circ \tilde{y}_j\right) \right] 
= \sum_{i \in [k]} \prod_{j \in [d]} \Ex_{\tilde{y}_j \underset{1 - \eta}{\sim} y_j} \left[ f^i_{w_j}\left(\pi_{w_j \to v} \circ y_j\right) \right] \\
= \sum_{i \in [k]}  \prod_{j \in [d]} \Gamma_{1 - \eta}f^i_{w_j}\left(\pi_{w_j \to v} \circ y_j\right) .
\end{multline*}
where recall $\Gamma_{1 - \eta}$ is the noise operator for functions on $V^n_M$ (Definition \ref{defn:mark-noise}). Using this, 
\begin{align*}
\Pr\left[\mbox{Test Accepts}\right] 
	& = \Ex_v \Ex_{w_1,\ldots,w_d \sim N_{\cG}(v)}\Ex_{x,y_1,\ldots,y_d} \left[\sum_{i \in [k]}  \prod_{j \in [d]} \Gamma_{1 - \eta}f^i_{w_j}\left(\pi_{w_j \to v} \circ y_j\right) \right] \\
	& = \sum_{i \in [k]} \Ex_v \Ex_{w_1,\ldots,w_d \sim N_{\cG}(v)}\Ex_{x,y_1,\ldots,y_d} \left[\prod_{j \in [d]}  \Gamma_{1 - \eta}f^i_{w_j}\left(\pi_{w_j \to v} \circ y_j \right)\right] \\ 
	& = \sum_{i \in [k]} \Ex_v\Ex_{x, y_1,\ldots,y_d}  \left[\prod_{j \in [d]} \Ex_{w_j \sim N_{\cG}(v)} \Gamma_{1 - \eta}f^i_{w_j} \left(\pi_{w_j \to v} \circ y_j \right) \right] 
	\qquad \textrm{(using independence)}. 
\end{align*}
Let $g^i_v \defeq \Ex_{w_j \sim N_{\cG}(v)} \left[ f^i_{w_j} \left(\pi_{w_j \to v} \circ y_j \right) \right]$
Then,
\begin{equation} 
\Pr\left[\mbox{Test Accepts}\right] 
 = \sum_{i \in [k]} \Ex_v\Ex_{x,y_1,\ldots,y_d}  \left[\prod_{j \in [d]} \Gamma_{1 - \eta} g^i_v(y_j) \right].
\end{equation}
where the averaged functions $\Gamma_{1 - \eta}g^i_v$ are bounded in $[0,1]$. Now, suppose the test accepts with probability at least $1 - \xi_{k,\epsilon,d} + \zeta$,
where $\xi_{k,\epsilon,d}$ is as defined in Section \ref{sec:test}.
Then,
\[
\Pr_{v \sim V_\cG} \left[\Ex_{x,y_1,\ldots,y_d} \left[\sum_{i \in [k]} \prod_{j \in [d]}  \Gamma_{1 - \eta}g^i_{v}(y_j)\right] \geq 1 - \xi_{k,\epsilon,d} + \frac{\zeta}{2}\right] \geq \frac{\zeta}{2\xi_{k,\epsilon,d} - \zeta}. \]
This follows from the observation that for a random variable $Z \in [0,1]$ with 
$\Ex[Z] \geq 1 - \xi_{k,\epsilon,d} + \zeta$,
\[ \Pr\left[Z < 1 - \xi_{k,\epsilon,d} + \frac{\zeta}{2} \right]
 = \Pr\left[1 - Z > \xi_{k,\epsilon,d} - \frac{\zeta}{2} \right]
\leq \frac{\Ex[1 - Z]}{\xi_{k,\epsilon,d} - \frac{\zeta}{2}}
\leq \frac{\xi_{k,\epsilon,d} - \zeta}{\xi_{k,\epsilon,d} - \frac{\zeta}{2}}
= 1 - \frac{ \zeta}{2 \xi_{k,\epsilon,d} - \zeta}. \]

Define the set $v$'s for which the test passes with probability at least $1 - \xi_{k,\epsilon,d} + \zeta/2$ as $V_{\rm large}$. Fix such a $v \in V_{\rm large}$. We will need the following lemma. .
\begin{lemma}				
	\label{lem:invar}
	Let $h: V_M^n \to \mathbbm[0,1]$ such that for every $i \in [n]$, $\Inf{i}{ \Gamma_{1 - \eta} h} \leq \tau$ and let $\delta = 1/k = \Ex_{x \sim \mu_M^n}g(x)$. Then,  
	\[
		\Ex_{x,y_1,\ldots,y_d} \left[\prod_{j \in [d]} \Gamma_{1 - \eta} h (y_j) \right] \leq \delta -   \delta \xi_{k,\epsilon,d} + 2\sqrt{d}\rho(\tau,\epsilon,\eta,\alpha).
	\]
\end{lemma}
We defer the proof of the above to Section \ref{sec:invar}. Using Lemma \ref{lem:invar}, we obtain the following corollary.
\begin{corollary}
	\label{cor:invar}
	The following holds for every vertex $v \in V$. Let $g_v = (g^1_v,\ldots,g^k_v)$ be the concatenation of the coordinate functions and let $\delta = \frac1k$. If
	\[
	\Ex_{x}\Ex_{y_1,\ldots,y_d} \left[\sum_{i \in [k]} \prod_{j \in [d]} \Gamma_{1 - \eta} g^i_{v}(y_j)\right] \geq  1 - \xi_{k,\epsilon,d} + \xi/2,
	\]
	then there exists $i \in [n]$ such that $\Inf{i}{\Gamma_{1 - \eta} g_v} \geq \tau$.
\end{corollary}

\begin{proof}
	For contradiction, assume that $\max_{i \in [n]} \Inf{i}{ \Gamma_{1 - \eta} g_v} \leq \tau$. Since $g_v(x) = \Ex_{w \sim N_\cG(v)}f_w\left(\pi_{w \mapsto v} \circ x \right)$, and the range of $f_w$'s are $\Delta_k$, it follows that the range of $\Gamma_{1 - \eta} g_v$ is still $\Delta_k$. Therefore using the definition of influences from Proposition \ref{prop:fourier}, for any $i \in [n]$ we get that
\[	
	\tau  \geq\Inf{i}{\Gamma_{1 - \eta} g_v} = \sum_{j \in [k]}\Inf{i}{\Gamma_{1 - \eta} g^j_v}  \geq \max_{j \in [k]} \Inf{i}{\Gamma_{1 - \eta}g^j_v}. 
\]	
Therefore, for every $g^j_v:V_M^n \to \mathbbm{R}_+$ we have $\max_{i \in [n]}\Inf{i}{\Gamma_{1 - \eta'} g^j_v} \leq \tau$. Furthermore, since $f_w$'s are folded, we have $\Ex f^i_{w_j} = 1/k = \delta$ for every choice of $i \in [k]$, and consequently $\Ex_x g^i_v(x) = \Ex_{w \sim v} \Ex_{x} f^i_w( \pi_{w \to v} \circ x) = \delta$. Invoking Lemma \ref{lem:invar} for any $i \in [n]$ we get that  
	\[
	\Ex_{x,y_1,\ldots,y_d}\left[ \prod_{j \in [d]}  \Gamma_{1 - \eta} g^i_v(y_j)\right] \leq \delta - \delta \xi_{k,\epsilon,d} + 2\sqrt{d}\rho(\tau,\epsilon,\eta,\alpha).
	\]
	Combining the above bound for every $i \in [k]$ we get that 
	\begin{eqnarray*}
	\Ex_{x,y_1,\ldots,y_d} \left[\sum_{i \in [k]} \prod_{j \in [d]} \Gamma_{1 - \eta} g^i_{v}(y_j)\right]
	&\leq& \sum_{i \in [k]} \left( \delta - \delta \xi_{k,\epsilon,d} + 2\sqrt{d}\rho(\tau,\epsilon,\eta,\alpha) \right) \\
	&\leq& 1 - \xi_{k,\epsilon,d} +   2dk\rho(\tau,\epsilon,\eta,\alpha). 
	\end{eqnarray*}
	Since by our choice of parameters we have $2dk\rho(\tau,\epsilon,\eta,\alpha) \leq \xi/4$, this gives us the contradiction.
\end{proof}
Finally the following lemma is the standard influence decoding step which decodes a good labeling for the underlying \uniquegames~instance.
\begin{lemma}
\label{lem:inf-dec}
	Let $\cG$ be such that 
	\[
	\Pr_{v \sim V_{\cG}}\left[ \max_{i \in [n]} \Inf{i}{\Gamma_{1 - \eta}g_v} \geq \tau\right] \geq \gamma \]
	Then there exists a labeling $\sigma:V \to [n]$ which satisfies at least $\Omega(\gamma \tau^2\epsilon^2\eta^{2} k^{-5})$ fraction of edges.
\end{lemma}
The proof of the above lemma follows using techniques identical to \cite{KKMO07}. We include a proof of it in Section \ref{sec:inf-dec} to derive explicit lower bounds for the fraction of constraints satisfied.

\subsubsection*{Putting things together}
\begin{proof}[Proof of Theorem \ref{thm:sound}]
Using our choice of parameters (Section \ref{sec:test}), 
we know that $V_{\rm large} \subset V$ is of size at least 
$\gamma|V| = \Abs{V} \zeta/(2 \xi_{k,\epsilon,d} - \zeta) \geq \Abs{V}/20$. 

Furthermore, using Corollary \ref{cor:invar}, we get that for every $v \in V_{\rm large}$, $\max_{i \in [n]}\Inf{i}{\Gamma_{1 - \eta} g_v} \geq \tau$. Now, using Lemma \ref{lem:inf-dec}, and plugging in the values of the parameters $\eta,\tau,\gamma$ as functions of $\epsilon,d,k$ we get that there exists an assignment to the \uniquegames~instance which satisfies at least 
\[
C_0k^{-5}\gamma\tau^2\epsilon^2\eta^{2} \geq C_0 \epsilon^2(\delta\zeta/dk)^{8d^2\log (2k/\epsilon)/\epsilon^2}(\epsilon/10d)^2/k^5 \geq (\epsilon/k)^{C'd^2/\epsilon^2}
\]fraction of edges. 
\end{proof}
Now we prove Theorem \ref{thm:hyperug}.
\begin{proof} [Proof of Theorem \ref{thm:hyperug}]
Let $\cG$ be an $(1 - \epsilon_c,\epsilon_s)$-\ug~instance with $\epsilon_c,\epsilon_s$ chosen as in the statement of theorem. Let $\cH$ be the \hug~instance output by the reduction in Figure \ref{fig:pcptest-1-intro}. Now suppose $\cG$ is a YES instance. Then using the arguments from Section \ref{sec:comp} we know that there exists a labeling which satisfies at least $1 - 4\epsilon$ fraction of hyperedge constraints in $\cH$. On the other hand, if $\cG$ is a NO instance. Then using Theorem \ref{thm:bound} no labeling satisfies more than  $1 - C\sqrt{\epsilon\log d \log k}$-fraction of constraints in $\cH$ (here $C$ is the constant from Theorem \ref{thm:sound}). Combining these two cases completes the proof of Theorem \ref{thm:hyperug}.   
\end{proof}

\subsection{Proof of Lemma \ref{lem:invar}}
\label{sec:invar}

The proof of Lemma \ref{lem:invar} uses the following isoperimetry result on Gaussian graphs,
we prove it in Section \ref{sec:gaussiso}.
\begin{restatable}{rethm}{gaussiso}		
\label{thm:bound}
	Let $f: \mathbbm{R}^m \to [0,1]$ such that $\Ex_{g \sim N(0,1)^m} f(g) = \delta$. Let $g \sim N(0,1)^m$ and $g_1,\ldots,g_d \underset{\overline{\rho}}{\sim} g$ be an ensemble of $m$-dimensional Gaussian random variables which are $\overline{\rho}$ correlated with $g$ such that $\overline{\rho}$ is nonnegative and  $\|\overline{\rho}\|_\infty \leq 1 - \epsilon$. Then
	\[
	\Ex_{g_1,\ldots,g_d } \left[\prod_{i \in [d]} f(g_i)\right] \leq \delta -  C\delta \sqrt{\epsilon \log d \log \frac1\delta}
	\]
	where $C> 0$ is an universal constant.
\end{restatable}
We shall also use the following variant of invariance principle (Theorem 3.6~\cite{IM12}).
\begin{theorem}[Invariance Principle] 			
\label{thm:clt}
	Let $Y = (Y_1,\ldots,Y_n)$ be an independent sequence of ensembles on a finite probability space $(\prod_{i \in [n]} \Omega_i,\prod_{i \in [n]}\mu_i)$, such that probability of any atom is at least $\alpha$ i.e., $\Pr_{\mu_i}\left[Y_i = y \right] \geq \alpha$ for all $y \in \Omega_i$. Let $Q$ be a $d$-dimensional multilinear polynomial such that ${\rm Var}(Q_j(Y)) \leq 1$ and ${\rm Var}(Q^{>p}_j(Y)) \leq (1-\epsilon\eta)^{2p}$ for $p = O(\log(1/\tau)/ \log(1/\alpha))$. Furthermore, let $\max_i \Inf{i}{Q_j} \leq \tau$ Finally, let $\psi:\mathbbm{R}^k \to \mathbbm{R}$ be  $d$-Lipschitz continuous. Then,
		\[
		\left|\Ex\left[\psi(Q(Y))\right] - \Ex\left[ \psi(Q(Z))|\right]\right| \leq O\left(d \tau^{\epsilon\eta/\log\frac1\alpha}\right)
		\]
	where $Z$ is a Gaussian ensemble with matching covariance structure.
\end{theorem}

\begin{proof}[Proof of Lemma \ref{lem:invar}]
Recall that in the setting of the lemma we have $\max_{i \in [n]} \Inf{i}{\Gamma_{1 - \eta}g} \leq \tau$. Define the map $\psi: \mathbbm{R} \to \mathbbm{R}$ as follows
\[
\psi(x) \defeq 
\begin{cases}
0 & \mbox{ if } x < 0 \\
x & \mbox{ if } x \in [0,1] \\
1 & \mbox{ if } x > 1
\end{cases}
\]
and let $\psi'(y_1,y_2,\ldots,y_d) \defeq \prod_{j \in [d]} \psi(y_j)$. 
\begin{claim}
$\psi'$ is $\sqrt{d}$-lipschitz continuous. 
\end{claim}
\begin{proof}
Observe that for any $y,y' \in \mathbbm{R}^d$ we have 
\begin{align*}
\left|\psi'(y_1,y_2,\ldots,y_d) - \psi'(y'_1,y'_2,\ldots,y'_d)\right|
&= \left| \prod_{j \in [d]} \psi(y_j) - \prod_{j \in [d]} \psi(y'_j) \right| \\
&\leq 
\sum_{i = 1}^d\left| \prod_{j \leq  i} \psi(y_j) \prod_{j' > i} \psi(y'_{j'}) - \prod_{j \leq i + 1} \psi(y_j) \prod_{j' > i+1} \psi(y'_{j'})\right| \\
&\leq \sum_{i = 1}^d\prod_{j < i} \psi(y_j) \prod_{j' > i} \psi(y'_{j'})\left|\psi(y_i) - \psi(y'_i)\right| \\
&\leq \sum_{i = 1}^d\left|\psi(y_i) - \psi(y'_i)\right| \\
&\leq \sum_{i = 1}^d\left|y_i - y'_i\right| \leq \sqrt{d}\|y - y'\|_2.
\end{align*}
where the first inequality can be obtained by writing the expression inside the absolute value as a telescoping sum with the intermediate summands consisting of terms which are products of $\psi(y_j)$'s for some $j \leq i$ and $\psi(y'_j)$'s for $j > i$, followed by an application of triangle inequality. 
\end{proof}
Now, since $g$ is bounded between $[0,1]$ we have 
\begin{eqnarray*}
\Ex_x \Ex_{y_1,\ldots,y_d} \left[ \prod_{j \in [d]}  \Gamma_{1 - \eta} g(y_j) \right] 
= \Ex_{x}\Ex_{y_1,\ldots,y_d } \left[  \prod_{j \in [d]} \psi\left(\Gamma_{1 - \eta} g(y_i)\right)\right]. 
\end{eqnarray*}

We shall now  pass on to the Gaussian space using the Invariance principle (Theorem \ref{thm:clt}) as follows. For every $i \in [n]$, define $Y_i := \{\chi_r(y_{j}(i))\}_{r \in [2k],j \in [d]}$, and let $Y = (Y_1,Y_2,\ldots,Y_n)$ denote the ensemble of random variables. Let $Q:Y \to \mathbbm{R}^d$ be the vector valued multilinear polynomial $(Q_1(Y),Q_2(Y),\ldots,Q_d(Y))$, where for every $j \in [d]$, $Q_j(Y)$ is the multilinear polynomial representation of the function $\Gamma_{1 - \eta}g(y_j)$ in terms of the Fourier expansion over $V^n_M$. Then we observe the following. 

\begin{enumerate}
	\item For every $j \in [d]$, the marginal distribution of $y_j$ is $\mu^n_M$. Therefore by the assumption on $\Gamma_{1 - \eta}g$ for Lemma \ref{lem:invar}, we have 
	\[
	\max_{i \in [n]}\Inf{i}{Q_j(Y)} = \max_{i \in [n]}\Inf{i}{\Gamma_{1 - \eta}g(y_j)} \leq \tau.
	\]
	\item Again, since $Q_j(Y) = \Gamma_{1 - \eta} g$, with the range of $g$ bounded in $[0,1]$ we have ${\rm Var}(Q_j) \leq 1$. Furthermore, since the spectral gap of the Markov chain $M$ is at least $c\epsilon$ (Claim \ref{cl:M-prop}), using the first point of Lemma \ref{lem:four-decay} we have ${\rm Var}(Q^{>p}_j) \leq (1 - c\epsilon \eta)^{2p}$ for any $p \geq 1$. 
	\item From Fact \ref{fact:covar}, we know that there exists an ensemble of Gaussian random variables $Z = (Z_1,Z_2,\ldots,Z_n)$ with matching the covariance structure of $Y$ i.e, $\Ex_Y Y_i Y^\top_{i'}  = \Ex_Z Z_i Z^\top_{i'}$ for every $i,i' \in [n]$. In particular, there exists a correlation vector $\overline{\rho} \in [-1,1]^{2kn}$ such that $Z_{y_j}$'s are ${\overline{\rho}}$ correlated copies of a $2kn$-dimensional Gaussian vector $Z_{x} \sim N(0,1)^{2kn}$. From the spectral gap guarantee, we know that the cross variable correlations are bounded by $1 - c\epsilon$. Furthermore, since the transition matrix of the Markov chain $M$ is PSD (Claim \ref{cl:M-prop}), the correlations in the ensemble $Z$ are {\em nonnegative}.
\end{enumerate}	
	Note that in ensemble constructed above, the ensemble $Y_i$ is completely determined by the $d$-tuple $(y_1(i),y_2(i),\ldots,y_d(i))$. And any realization in the support of the random $d$-tuple of vertices $(y_1(i),y_2(i),\ldots,y_d(i))$ appears with probability at least $\alpha = (\epsilon/k)(\epsilon)^d$. Therefore, instantiating Theorem \ref{thm:clt} with $\psi',Q,Y,Z$ as defined above and writing $G \defeq \Gamma_{1 - \eta} g$, we can upper bound the above expression as follows
\begin{align*}
	&\Ex_{x \sim \mu^n_M}\Ex_{y_1,\ldots,y_d {\sim} \mu^n_M(x)  }\left[  \prod_{j \in [d]} \psi\left(\Gamma_{1 - \eta} g(y_j)\right)\right] \\
	&\leq \Ex_{Z_x \sim N(0,{\rm I}_{2kn})}\Ex_{Z_{y_1},\ldots,Z_{y_d} \underset{\overline{\rho}}{\sim} Z_x } \left[ \prod_{j \in [d]} \psi\left(G(Z_{y_j})\right)\right] + \sqrt{d}\rho(\tau,\epsilon,\eta,\alpha) \qquad \textrm{(Using Theorem \ref{thm:clt})} \\
	&\overset{1}{\leq }\Ex_{Z_x \sim N(0,{\rm I}_{2kn})} \left[ \psi\left(G(Z_x)\right)\right]  -  C\cdot \Ex_{Z_x \sim N(0,{\rm I}_{2kn})} \left[ \psi\left(G(Z_x)\right)\right] \sqrt{\epsilon \log d \log \frac{1}{\Ex_{Z_x} \left[ \psi\left(G(Z_x)\right)\right]}} + \sqrt{d}\rho(\tau,\epsilon,\eta,\alpha) \\
&\overset{2}{\leq }\Ex_{x \sim \mu^n_M} \left[ \psi\left(G(x)\right)\right] - C\left(\Ex_{x \sim \mu^n_M} \left[ \psi\left(G(x)\right)\right] - \rho(\tau,\epsilon,\eta,\alpha)\right) \\
& \qquad \qquad \qquad \qquad \qquad \qquad \quad
\times  \sqrt{\epsilon \log d \log \frac{1}{\Ex_{x\sim \mu^n_M} \left[ \psi\left(G(x)\right)\right] + \rho(\tau,\epsilon,\eta,\alpha)}} +2 \sqrt{d}\rho(\tau,\epsilon,\eta,\alpha)	\\
&\overset{3}{\leq} \delta - C' \delta \sqrt{\epsilon \log d \log k} .
\end{align*}
We justify steps $1$ to $3$. In step $1$, from Fact \ref{fact:covar} we know that $Z_{y_l} \underset{\overline{\rho}}{\sim} Z_{x}$ for every $l \in [d]$, and $Z_x \sim N(0,I)$. Furthermore, $\|\overline{\rho}\|_\infty \leq 1 - c\epsilon$ and $\overline{\rho}$ is nonnegative. Then step $1$ is established using Theorem \ref{thm:bound}. Step $2$, follows again by an application of Theorem \ref{thm:clt} instantiated with the $1$-lipschitz function $\psi(\cdot)$. Finally for step $3$ we observe that for $x \sim \mu^n_M$, and $x' \underset{1 - \eta}{\sim }x$, the random $n$-ary vertex $x'$ is distributed as $\mu^n_M$. Therefore, 
\[
\Ex_{x} \left[ \psi\left(\Gamma_{1 - \eta}g(x)\right)\right] = \Ex_{x} \left[ \Gamma_{1 - \eta}g(x)\right] =  \Ex_x \Ex_{x' \underset{1 - \eta}{\sim} x} \left[ g(x') \right] 
= \Ex_{x'} g(x') = \delta 
\]
which along with the fact that $\rho(\tau,\epsilon,\eta,\alpha) \ll \delta^2\epsilon^2/d$ by our choice of parameter $\tau$, gives us the required bound.

\end{proof}

\section{$d$-ary Gaussian Isoperimetry}
\label{sec:gaussiso}

In this section, we prove Theorem \ref{thm:bound}.
\gaussiso*
The key tool used is the following generalization of Borell's isoperimetric inequality to collections of Gaussians. 
\begin{theorem}[\cite{IM12} Theorem 4.5]					\label{thm:egt}
	Fix $\overline{\rho} \in [0,1]^n$. Suppose $h_1,h_2,\ldots,h_d \sim N(0,{\rm I}_n)$ are jointly normal with ${\rm Cov}(h_i,h_j) = {\rm Diag}(\overline{\rho})$ for all $i \neq j$, where ${\rm Diag}(\overline{\rho})$ is the diagonal matrix with the diagonal entries given by the vector $\overline{\rho}$. Then for any choice of sets $A_1,A_2,\ldots,A_d \subseteq \mathbbm{R}^n$ we have
	\[
	\Pr_{h_1,\ldots,h_d}\left[\bigwedge_{j \in [d]}\{h_j \in A_j\}\right] \leq \Pr_{h_1,\ldots,h_d}\left[\bigwedge_{j \in [d]}\{h_j \in H_j\}\right]
	\] 
	where for every $j \in [d]$, $H_j$ is the halfspace given by $H_j = \{x \in \mathbbm{R}^n | x_1 \leq \Phi^{-1}(\gamma^n(A_j))\}$. Here $\gamma^n$ being the $n$-dimensional Gaussian measure and $\Phi(\cdot)$ is the Gaussian CDF function.
\end{theorem}

\begin{proof}[Proof of Theorem \ref{thm:bound}]
	We begin by rewriting the LHS of the above equation as 
	\[
	\Ex_g \Ex_{h_1,\ldots,h_d \sim N_\rho(g)}\left[\prod_{j \in [d]} f(h_j)\right] = \Ex_g \left[\left(U_{\overline{\rho}} f(g)\right)^d\right] .
	\]	
	This follows using the definition of the Gaussian noise operator and the observation that conditioning on the realization of the Gaussian vector $g$, the $\rho$-correlated Gaussians $h_1,\ldots,h_d$ are drawn independently. As is standard, Claim \ref{cl:01sets} implies that it suffices to prove the above for any set indicator $f:\mathbbm{R}^n \to \{0,1\}$ i.e., $f$ is the indicator function of a set i.e.,
	\[
	\argmax_{\substack{f:\mathbbm{R}^n \to [0,1] \\ \Ex_g f(g) = \delta}} \Ex_g \left[\left(U_{\overline{\rho}} f(g)\right)^d\right]
	= \argmax_{\substack{f:\mathbbm{R}^n \to \{0,1\} \\ \Ex_g f(g) = \delta}} \Ex_g \left[\left(U_{\overline{\rho}} f(g)\right)^d\right]
	\]
	Therefore without loss of generality we can assume that $f(\mathbbm{R}^n) = \{0,1\}$. By construction, for every $j \in [d],i \in [m]$ we can write $h_j(i) = \rho_ig(i) + \sqrt{1 - \rho^2_i} z_j(i)$, where $z_1(i),z_2(i),\ldots,z_d(i) \sim N(0,1)$ are independent Gaussians. Therefore for any $i,j \in \{0,1\ldots,d\}$ we have that $h_i,h_j$ are marginally $N(0,{\rm Id}_m)$ and ${\rm Cov}(h_i,h_j) = \left({\rm Diag}(\overline{\rho})\right)^2$. Define $A = {\rm supp}(f)$ to be the set indicated by the function $f$, and let $H \subset \mathbbm{R}^n$ be the halfspace $H = \{x \in \mathbbm{R}^n | x_1 \leq t \}$ with $t = \Phi^{-1}(\delta)$. Instantiating Theorem \ref{thm:egt} with $A_j = A$ and $H_j = H$ for every $j \in [d]$ we get that 
	\begin{eqnarray*}
		\Ex_g\Ex_{h_1,\ldots,h_d \sim N_\rho(g)} \left[\prod_{j \in [d]}f(h_j)\right] 
		&{=}&\Ex_g\Pr_{h_1,\ldots,h_d} \left[ \wedge_{j \in \{1,\ldots,d\}} h_j \in A \right] \\
		&{\leq}&\Ex_g\Pr_{h_1,\ldots,h_d} \left[ \wedge_{j \in \{1,\ldots,d\}} h_j \in H \right] \\
		&=&\Ex_g\Pr_{h_1,\ldots,h_d} \left[ \wedge_{j \in [d]} h_{j,1} \leq t \right] \\
		&\leq& \delta - C\delta\sqrt{\epsilon \log d \log\frac1\delta} \qquad \textrm{(Using  Proposition \ref{prop:line-bound})}. 
	\end{eqnarray*}
\end{proof}

\begin{claim}					
	\label{cl:01sets}
	Fix $\overline{\rho} \in (0,1)^n$ and $\delta \in (0,1)$. Then among all functions $f:\mathbbm{R}^n \to [0,1]$ such that $\Ex_{g \sim N(0,1)^n}[f(g)] = \delta$, the quantity $\Ex_{g \sim N(0,1)^n}\left[(U_{\overline{\rho}}f)^d\right]$ is maximized for some $f: \mathbbm{R}^n \to \{0,1\}$. 
\end{claim}
\begin{proof}
	Define the set $S_{n,\delta} := \{f:\mathbbm{R}^n \to [0,1] : \Ex_g[f(g)] = \delta \}$. Clearly $S_{n,\delta}$ is convex. Now we claim that the functional $f \mapsto \Ex_g\left[(U_{\overline{\rho}} f(g))^d\right]$ is convex. For any pair of functions $f,f' \in S_{n,\delta}$, $g \in \mathbbm{R}^n$ and $\lambda \in [0,1]$ we have 
	\begin{eqnarray*}
		\left(U_{\overline{\rho}} (\lambda f + (1 - \lambda) f')(g)\right)^d &=& \left(\lambda \left(U_{\overline{\rho}} f(g)\right) + (1 - \lambda)\left( U_{\overline{\rho}} f' (g)\right)\right)^d  \\
		&\leq& \lambda \left(U_{\overline{\rho}} f(g)\right)^d + (1 - \lambda)\left( U_{\overline{\rho}} f' (g)\right)^d  .
	\end{eqnarray*}
	where the first equality follows from linearity of $U_{\overline{\rho}}$ and the last inequality follows from the observations that $f,f'$ are nonnegative functions, and the map $x \mapsto x^d$ is convex for $x \geq 0$ and $d \geq 1$. Therefore, integrating both sides with respect to the $n$-dimensional Gaussian measure we get that 
	\[
	\Ex_{g} \left[\left(U_{\overline{\rho}} (\lambda f + (1 - \lambda)f')(g)\right)^d\right] \leq \lambda \Ex_{g} \left[\left(U_{\overline{\rho}} f(g)\right)^d\right] + (1 - \lambda)\Ex_{g} \left[\left(U_{\overline{\rho}} f')(g)\right)^d\right]. 
	\]
	which establishes the convexity of the map $f \mapsto \Ex_g\left[(U_{\overline{\rho}} f(g))^d\right]$. Therefore, the maximizer of the map over the set $S_{n,\delta}$ must be an extreme point. Using the fact that the extreme points of the set $S_{n,\delta}$ are set indicator functions (Exercise 11.25~\cite{OD14}), completes the proof.
\end{proof}

\subsection{Stability on a Line}

We shall need the following Gaussian stability bounds for halfspaces. 

\begin{fact}				\label{lem:stab-half}
	There exists a constant $\delta_0 \in (0,1)$ such that the following holds for all $\delta \in (0,\delta_0]$ and every $\nu \in \left(0,\frac{1}{10\sqrt{\log\frac{1}{\delta}}}\right)$. Let $H \subset \mathbbm{R}^n$ be a halfspace satisfying $\gamma^n(H) = \delta$, and  let $H^{-\nu} \defeq \{x \in H | {\rm dist}(x,\partial H) \geq \nu \}$ be the $\nu$-shift of the halfspace $H$. Then we have $\gamma^n(H) - \gamma^n(H^{-\nu}) \geq (1/8)\delta\nu\sqrt{\log 1/ \delta}$.
\end{fact}	

\begin{proof}
	Since the Gaussian measure is rotation invariant, without loss of generality, we can assume that $H: = \{x \in \mathbbm{R} | x \leq t\}$. We shall use the following well known bounds for the Gaussian cdf function.
	\begin{fact}[Eq. 7.1.13 \cite{AS65}]			\label{fact:cdf-bound}
		For every $t \in (-\infty,0)$ we have 
		\[
		\frac{1}{\sqrt{2 + t^2} + |t|} \cdot \frac{1}{\sqrt{2\pi}}e^{-t^2/2} \leq \Phi(t) \leq \frac{1}{t} \cdot \frac{1}{\sqrt{2\pi}}e^{-t^2/2}. 
		\]
	\end{fact}
	The following is an immediate consequence of the above fact.
	\begin{claim}					\label{cl:t-bound}
		There exists constant $t_0 \in (-\infty,0)$, such that for every $t \leq t_0$ we have 
		\[
		\frac12 \sqrt{\log \frac{1}{\Phi(t)}} \leq |t| \leq 2 \sqrt{\log \frac{1}{\Phi(t)}}. 
		\]
	\end{claim}
	\begin{proof}
		Let $\delta := \Phi(t)$. Using Fact \ref{fact:cdf-bound}, we have $e^{t^2/2} \geq (3\sqrt{2\pi} \delta |t|)^{-1}$, which by taking log on both sides and rearranging gives us
		\[
		\frac{t^2}{2} + \log(3\sqrt{2\pi}|t|) \geq \log\frac{1}{\delta} \Rightarrow |t| \geq \frac12\sqrt{\log \frac{1}{\delta}}
		\]
		by choosing $|t|$ to be large enough. On the other hand, again using Fact \ref{fact:cdf-bound} we observe that $e^{t^2/2} \leq 1/(\delta t)$. Therefore again taking log on both sides and the fact that $|t|$ is large enough, we have $|t| \leq 2 \sqrt{\log (1/\delta)}$.
	\end{proof}
	Denote $t: = \Phi^{-1}(\delta)$ and observe that from our choice of parameters $t + \nu \leq - (1/2)\sqrt{\log 1/\delta} + (1/100)\sqrt{\log 1/\delta} \leq 0$ (using the lower bound on $|t|$ from Claim \ref{cl:t-bound}). Since for any $x \in (-\infty,0)$, the map $x \mapsto e^{-x^2/2}$ is an increasing function we have 
	\begin{align}				
		\gamma^n(H) - \gamma^n(H^{-\nu}) & =  \Phi(t) - \Phi(t-\nu) \label{eqn:gauss-bd} \\
		&= \int^{t}_{t-\nu} \frac{1}{\sqrt{2\pi}} e^{-x^2/2} dx \nonumber \\
		&\geq \frac{ e^{-(t-\nu)^2/2}}{\sqrt{2\pi} }\int^{t}_{t-\nu} 1.dx \nonumber \\
		&= \frac{ e^{-(t-\nu)^2/2}}{\sqrt{2\pi}}\nu \nonumber \\
		&\geq (|t - \nu|)\Phi(t-\nu)\nu \nonumber & \textrm{(using Fact \ref{fact:cdf-bound})}  \\
		&\geq \frac{\nu |t|}{2}\Phi(t-\nu)& \textrm{($2\nu \leq 2 \leq |t|$ for small enough $\delta_0$)} \nonumber \\
		& \geq \frac{\nu }{4}\Phi(t - \nu)\sqrt{\log\frac{1}{\delta}} & \textrm{(using Claim \ref{cl:t-bound})} \label{eqn:final}
	\end{align}
	On the other hand, 
	\begin{align*}
		\Phi(t) - \Phi(t-\nu) & = \int^{t}_{t - \nu} \frac{e^{-x^2/2}}{\sqrt{2\pi}} dx 
		\leq \nu \frac{e^{-t^2/2}}{2} & \textrm{(using $t - \nu < t < 0$)} \\
		& \leq 3\nu|t|\Phi(t) & \textrm{(using Fact \ref{fact:cdf-bound})} \\
		& \leq 6 \nu \delta \sqrt{\log 1/\delta} & \textrm{(using Claim \ref{cl:t-bound})}. 
	\end{align*}
	Since by our choice of parameters we have $\nu\sqrt{\log (1/\delta)} \leq 0.1$, it follows that 
	\[
	\Phi(t - \nu) \geq \Phi(t) - 6\nu\delta\sqrt{\log \frac1\delta } \geq \frac{\delta}{2}.
	\]
	Plugging in the above bound into Eq. \ref{eqn:final} gives us the required lower bound on Eq. \ref{eqn:gauss-bd}.
\end{proof}
Proposition \ref{prop:line-bound} is folklore; we include a proof of it here for the sake of completeness.
\begin{proposition}[Folklore]
	\label{prop:line-bound}
	Let $\delta_0$ be as in Lemma \ref{lem:stab-half}. Then there exists constants $C,C_0$ such that the following holds. Fix any $\delta \in (0,\delta_0)$, $\epsilon \in (0,1)$ and $d \geq 10\log(1/\epsilon\delta)$  such that $\sqrt{\epsilon \log d \log (1/\delta)} \leq 1/10C$. Define $\rho:= 1- \epsilon$. Let $g \sim N(0,1)$, and $h_0,\ldots,h_d \sim N_\rho(g)$.  we have 
	\begin{equation}
	\Ex_g\Pr_{h_0,\ldots,h_d}\left[ \wedge_{i \in [d]} \left\{h_j \leq t \right\}\right] \leq \Phi(t) - C\Phi(t)\sqrt{\epsilon \log d \log\frac{1}{\Phi(t)}}
	\end{equation}
	where $\Phi(\cdot)$ is the CDF of the standard normal distribution.
\end{proposition}

\begin{proof}
	For brevity, let $\delta := \Phi(t)$ and $f: \mathbbm{R} \to \{0,1\}$ be the halfspace indicator function $f(x) = \mathbbm{1}(x_1 \leq t)$. Let $I$ denote the interval $\left[t - C'\sqrt{\epsilon \log d}, t\right]$ for some constant $C'$ to be chosen later. Using conditional expectations we can upper bound the expectation as 
	\begin{multline}
		\label{eq:bound0}
		\Ex_g\Pr_{\left(h_j\right)^d_{j = 1}}\left[ \wedge_{j \in [d]} \left\{h_j \leq t \right\}\right] 
		\leq \Pr_g\Big[g \leq t - C'\sqrt{\epsilon \log d}\Big] + 
		\Pr_g \Big[g \in  I \Big]\Pr_{\left(h_j\right)^d_{j = 1} \big| g \in I}\left[ \wedge_{j \in [d]} \left\{h_j \leq t \right\}\right] \\  
		+ \Pr_g \Big[g \geq t \Big]\Pr_{\left(h_j\right)^d_{j = 1} \big| g \geq t}\left[ \wedge_{j \in [d]} \left\{h_j \leq t \right\}\right]. 
	\end{multline}
	Recall that for every $j \in [d]$, we can write $h_j = (1 - \epsilon)g + \sqrt{2 \epsilon - \epsilon^2} z_j$ where $z_1,z_2,\ldots,z_d \sim N(0,1)$ are independent Gaussians. First we consider the case $g \geq t$. 
	Since $t > 0$, we have $(1 - \epsilon)g > t$ (recall that $\epsilon \in (0,1)$). Then
	\begin{equation}			\label{eq:bound1}
	\Pr_{\left(h_j\right)^d_{j = 1} \big| g \geq t} \left[\wedge_{j \in [d]} \left\{ h_j \leq t \right\} \right]
	\leq \Pr_{z_1,\ldots,z_d \sim N(0,1)} \left[\wedge_{j \in [d]} \left\{ z_j \leq  0\right\} \right] 
	= 2^{-d} .
	\end{equation}
	Now we proceed to upper bound second term in equation \ref{eq:bound0}. Here we condition on $g \in I$.  Then we have $(1 - \epsilon)g \geq g \geq t - \frac{1}{10}\sqrt{\epsilon \log d}$ (recall that $t \leq 0$), 
	and in particular
	\[
	\bigwedge_{j \in [d]} \left\{ h_j \leq t \right\} \Rightarrow 
	\bigwedge_{j \in [d]} \left\{ \left(\sqrt{2 \epsilon - \epsilon^2} \right)z_j \leq \frac{1}{10}\sqrt{\epsilon \log d} \right\} \Rightarrow 
	\bigwedge_{j \in [d]} \left\{ z_j \leq \frac{1}{10}\sqrt{\log d} \right\}. 
	\]
	where in the last step we use that $2\epsilon - \epsilon^2 \geq \epsilon$ for $\epsilon \leq 1/2$. Hence conditioned on the event $g \in I$ and using the independence of the $z_j$s we get that 
	\begin{eqnarray}			\label{eq:bound2}
	\Pr_{h_1,\ldots,h_d \sim N_\rho(g)} \left[\wedge_{j \in [d]} h_j \leq t \Big| g \in I \right]
	\leq \Pr_{z_1,\ldots,z_d \sim N(0,1)} \left[\max_{j \in [d]} z_j \leq \frac{1}{10}\sqrt{\log d} \right] 
	\leq \exp\left(- \Omega(\log d)\right) \leq \frac12.
	\end{eqnarray}
	for large enough choice of $d$. The last equality combines with the fact that expected maximum of $d$ independent Gaussian random variables is at least $(1/2)\sqrt{\log d}$ with lipschitz concentration of gaussian random variables (Theorem 5.8~\cite{BLM13})  
	
	On the other hand, instantiating Fact \ref{lem:stab-half} with the halfspace $H = \mathbbm{1}( x \leq t)$ we get that
	\begin{eqnarray}			\label{eq:bound3}
	\Pr_{g \sim N(0,1)} \left[ g \in I\right] = \gamma^1(H) - \gamma^1\left(H^{-(1/10)\sqrt{\epsilon \log d}}\right)  \geq C \delta \sqrt{\epsilon \log d \log\frac1\delta} . 
	\end{eqnarray}
	Therefore combining the bounds from \ref{eq:bound1},\ref{eq:bound2} and \ref{eq:bound3}, we get that  	
	\begin{align*}
		& \Pr_g \Big[g \leq t - C\sqrt{\epsilon \log d }\Big] + \Pr_g \Big[g \in  I \Big]\Pr_{\left(h_j\right)^d_{j = 1} \big| g \in I}\left[ \wedge_{j \in [d]} \left\{h_j \leq t \right\}\right] +  \Pr_g \Big[g \geq t \Big]\Pr_{\left(h_j\right)^d_{j = 1} \big| g \geq t}\left[ \wedge_{j \in [d]} \left\{h_j \leq t \right\}\right] \\
		&\leq  \left(\Pr_g \Big[g \leq t \Big] -  \Pr_g\Big[g \in I \Big] \right) + \frac12\Pr_g\Big[g \in I\Big] + (1 - \delta)2^{-d} \\
		&\leq  \delta - C\delta\sqrt{\epsilon \log d \log\frac1\delta}		\tag{Since $d \geq 10 \log(1/\epsilon\delta)$}	
	\end{align*}
	Finally plugging in $\delta = \Phi(t)$ in to the above expression gives us the bound.
\end{proof}

\paragraph{Addressing the constant volume case}.

We point out that Proposition \ref{prop:line-bound} is stated for all $\delta \leq \delta_0$ for some constant $\delta_0$, and in particular does not address the setting $\delta = 1/2$ which is needed for our hardness result for \oddcycletransversal~(Theorem \ref{thm:oct-hardness}). For that we state the following variant of the Proposition \ref{prop:line-bound} which works for all $\delta$'s but gives weaker quantitative relationships between $\epsilon$ and $\delta$.
\begin{proposition}
	\label{prop:line-bound-2}
	Given $\delta \in (0,1/2]$, there exists $\epsilon(\delta) \in (0,1)$ and constants $C,C_0$ such that the following holds for every $\epsilon \leq \epsilon(\delta)$. Fix $d \geq 10\log(1/\epsilon\delta)$  such that $\sqrt{\epsilon \log d \log (1/\delta)} \leq 1/10C$. Define $\rho:= 1- \epsilon$. Let $g \sim N(0,1)$, and $h_0,\ldots,h_d \sim N_\rho(g)$. Then we have 
	\begin{equation*}
	\Ex_g\Pr_{h_0,\ldots,h_d}\left[ \wedge_{i \in [d]} \left\{h_j \leq t \right\}\right] \leq \Phi(t) - C\Phi(t)\sqrt{\epsilon \log d \log\frac{1}{\Phi(t)}}
	\end{equation*}
	where $\Phi(\cdot)$ is the CDF of the standard normal distribution.
\end{proposition}
The proof of the above is identical to the proof of Proposition \ref{prop:line-bound}, where instead of Fact \ref{lem:stab-half}, we use the stability bounds from the following lemma which follows directly from the {\em Gaussian Isoperimetric Inequality}.
\begin{fact}[Proposition 5.27, 11.49~\cite{OD14}]				\label{lem:stab-half-asymp}
	For all $\delta \in (0,1/2]$, there exists $\epsilon(\delta) \in (0,1)$ such that for all $\epsilon \leq \epsilon(\delta)$ we have 
	\[
	\gamma^n(\Phi(t)) - \gamma^n(\Phi(t - \epsilon)) \geq \frac{\epsilon}{10}\sqrt{\log (1/\delta)}
	\]
	where $t = \Phi^{-1}(\delta)$.
\end{fact}

\newcommand{\hyperedge}{e}

\section{\hug~to \strug}

Using the hardness of \hug, we now reduce to \strug~in the following steps.

\subsection{\hug~with uniform weights}

The following uniformization step is well known for the case $d = 2$ i.e, \ug. Our proof for general $d$ is an adaptation of Proposition 7.2 from \cite{LRV13} to our setting.
\begin{theorem}			\label{thm:unif-weights}
	Given a hyperedge weighted \hug~instance $\cH(V,E,[k],\{\pi_{e,u}\}_{e \in E, u \in e},\cD_{\cH})$, there exists an efficient procedure which constructs \hug~instance $\cH'=(V',E',[k],\{pi_{e,v}\}_{e \in E', v\in e},\cD_{\cH'})$ such that 
	\begin{itemize}
		\item $|V'| = \bigo{d |V|^3 |E|^3}$
		\item $\cH'$ satisfies
		\[
		\left| \argmax_{\sigma:[V] \to [k] } \Pr_{e \sim \cD_{\cH}}\left[\sigma \mbox { satisfies } e\right] - 
		\argmax_{\sigma':[V'] \to [k] } \Pr_{e \sim \cD_{\cH'}}\left[\sigma' \mbox { satisfies } e\right]  \right| \leq \frac{1}{|V|} 
		\]
		i.e., the optimal values of $\cH$ and $\cH'$ differ by at most $1/|V|$.
		\item For every $v \in V'$, we have 
		\[
		\Pr_{e \sim \cD_{\cH'}} \left[ v \in e \right] = \frac{d}{|V'|}.
		\]
	\end{itemize}
\end{theorem}

\begin{proof}
	We begin with an elementary pre-processing step. We construct an intermediate \hug~instance $\cH_0(V_0,E_0,[k],\{\pi_{e,u}\}_{e \in E_0, u \in e},\cD_{\cH})$ from $\cH$ as follows. Let $\mu_\cH$ be the measure on $E$ corresponding to distribution $\cD_{\cH}$. 
	\begin{enumerate}
		\item Let $\mu'(e)$ be $\mu_{\cH}(e)$ rounded down to a precision $1/(2|V|^3|E|^3)$. Delete all $e$ such that $\mu'(e) = 0$. 
		\item Delete from $V$ every vertex $u$ and all hyperedges constraints incident on it such that $\sum_{e \ni u} \mu'(e) < 1/|V|^2|E|^2$. 
	\end{enumerate}
	Denote the set of remaining vertices by $V_0$ and the set of remaining edges by $E_0$. Define 	
	\[
	\mu_{\cH_0}(e) \defeq \frac{\mu'(e)}{\sum_{e \in E_0} \mu'(e)}.
	\]    
	By construction, we have 
	\begin{equation}
	\label{eq:mupr}
	\mu_{\cH}(e) - \frac{1}{2 |V|^3 |E|^3} \leq \mu'(e) \leq \mu_{\cH}(e) \qquad \forall e \in E.
	\end{equation}
	Let $E_1$ and $E_2$ denote the subset of constraints that are deleted in steps $1$ and $2$ respectively. 
	\[ \mu_\cH(E_2) \leq \sum_{u \in V } \frac{1}{|V|^2 |E|^2}
	\leq \frac{1}{2 |E|^2}. \]
	Therefore, 
	\begin{align*} 
		\sum_{e \in E_0} \mu'(e) & \geq \sum_{e \in E \setminus E_2} \left(\mu_{\cH}(e) - \frac{1}{2 |V|^3 |E|^3} \right)
		\geq \mu_{\cH}(E) - \mu_{\cH}(E_2) -  |E|\frac{1}{2 |V|^3 |E|^3} \\ 
		& \geq 1 - \frac{1}{2|E|^2} - \frac{1}{2 |V|^3 |E|^2} \geq 1 - \frac{1}{|E|^2}. 
	\end{align*} 
	Since, $E_0 \subseteq E$ and  we have $\sum_{e \in E_0} \mu'(e) \leq 1$.
	Therefore, from eqn \ref{eq:mupr},  for every $e \in E_0$ we have 
	\[
	\mu_{\cH}(e) - \frac{1}{|E|^2} \leq \mu_{\cH_0}(e) \leq \frac{1}{1 - \frac{1}{|E|^2}}\mu_{\cH}(e), 
	\]		 
	and hence
	\[ |\mu_{\cH_0}(e) - \mu_{\cH}(e)| \leq \max \left\{\frac{1}{|E|^2}, \frac{1}{|E|^2 - 1} \mu_{\cH}(e) \right\}
	\leq \frac{1 }{|E|^2 - 1} . \]
	Therefore, $\|\mu_{\cH} - \mu_{\cH_0} \|_{\rm TV} \leq 1/n$. This implies that we can work with $\cH_0$ instead of $\cH$ by losing a factor of $1/n$ in the compleleteness and soundness. 
	
	Therefore going forward we shall assume $\cH$ has weights are rounded off to a precision of $1/(|E|^3|V|^3)$. Furthermore, define the map $\nu:V \to \mathbbm{R}_+$ as $\nu(u) = \Pr_{ e \sim \mu_{\cH_0}}\left[u \in e\right]$. Then we can construct nonnegative integers $\{n(v)\}_{v \in V}$ and $N \defeq d |V|^3 |E|^3$ such that (i) $\sum_{v \in V} n(v) = N$ and (ii) $n(v) = \nu(v) N/d$ for every vertex $v$. Now we describe the construction of the \hug~instance $\cH'$.
	
	{\bf Vertex Set}. For every vertex $v \in V$, we introduce $n(v)$ copies of $v$, and denote them by $S_v \defeq \{(v,1),\ldots,(v,n(v))\}$. The vertex set of $\cH'$ is then $V' = \cup_{v \in V} S_v$.
	
	{\bf Constraint Set}. For every hyperedge constraints $e$ supported on vertices $(v_1,\ldots,v_d)$ (with bijections $\{\pi_{e,v_i}\}_{i \in [d]}$) and for every choice of $I \in [n(v_1)] \times [n(v_2)] \times \cdots \times [n(v_d)]$, we introduce a hyperedge constraint $e_I$ on vertices $\{(v_1,I_1),\ldots,(v_d,I_d)\}$ with weight
	\begin{equation}
	\mu_{\cH'}(e_I) \defeq \Pr_{e' \sim \mu_{\cH_0}} \left[ e' = e \right] \Pr_{I' {\sim} \times^d_{i = 1} [n(v_i)] } \left[I' = I \right] = \frac{\mu_{\cH_0}(e)}{\prod_{i \in [d]} n(v_i)}. 
	\end{equation}
	The constraint on hyperedge $e_I$ is given the following: 
	\[
	\pi_{e_I,(v_a,I_a)}(\sigma(v_a,I_a)) = \pi_{e_I,(v_b,I_b)}(\sigma(v_b,I_b)), \qquad\forall a,b \in [d].
	\]
	where for every $a \in [d]$, $\pi_{e_I,(v_a,I_a)} = \pi_{e,v_a}$. For brevity, we shall use $\cC(e)$ to denote the cluster of hyperedges in $\cH'$ corresponding to the hyperedge constraint $e$. The above completes the construction of the \hug~instance $\cH'$. Now we analyze the guarantees of the construction. We first show that the vertices in $\cH'$ have uniform weights. Indeed, for a fixed vertex $(v,i) \in V'$ we have 
	\[
	\Pr_{e'_I \sim \mu_{\cH'}} \left[(v,i) \in e'_I\right] = \Pr_{e' \sim \mu_{\cH_0}}\left[ v \in e' \right] \Pr_{e',I}\left[I_v = i | v \in e'\right]
	= \nu(u) \left(\frac{1}{n(v)}\right) = \frac{d}{N}.
	\]
	where the last step follows from our choice of the integers $n(v)$ and $N$. Now we argue completeness and soundness.
	
	{\bf Completeness}. Suppose there exists $\sigma:V \to [k]$ which satisfies $(1 - \gamma)$-fraction of hyperedges. Construct the labeling $\tilde{\sigma}:V' \to [k]$ as follows. For every $v \in V$ and $i \in [n(v)]$, we assign the label $\tilde{\sigma}(v,i) \defeq \sigma(v)$. Then the fraction of constraints in $\cH'$ satisfied by $\tilde{\sigma}$ is given by 
	\begin{align*}
		&\Pr_{e' = (v_1,\ldots,v_d) \sim \mu_{\cH'}} \left[ \forall a,b \in [d], \quad \pi_{e',v_a}(\tilde{\sigma}(v_a)) = \pi_{e',v_b}(\tilde{\sigma}(v_b))  \right] \\
		&=\Pr_{e = (v_1,\ldots,v_d) \sim \mu_{\cH_0}} \Pr_{I \sim \times^d_{i = 1} [n(v_i)]} \left[ \forall a,b \in [d], \quad \pi_{e_I,(v_a,I_a)}(\tilde{\sigma}((v_a,I_a))) = \pi_{e_I,(v_b,I_b)}(\tilde{\sigma}((v_b,I_b)))  \right] \\
		&=\Pr_{e = (v_1,\ldots,v_d) \sim \mu_{\cH_0}} \Pr_{I \sim \times^d_{i = 1}[n(v_i)]} \left[ \forall a,b \in [d],\quad \pi_{e,v_a}({\sigma}(v_a)) = \pi_{e,v_b}({\sigma}(v_b)) \right] \\
		&=\Pr_{e = (v_1,\ldots,v_d) \sim \mu_{\cH_0}}  \left[ \forall a,b \in [d],\quad \pi_{e,v_a}({\sigma}(v_a)) = \pi_{e,v_b}({\sigma}(v_b)) \right] \\
		&\geq 1 - \gamma . 
	\end{align*}
	
	{\bf Soundness}. Suppose there exists a labeling $\sigma:V' \to [k]$ which satisfies at least $(1 - \gamma)$-fraction of the constraints. Then consider the following randomized labeling procedure for $V$. For every vertex $v \in V$, sample $i(v) \sim [n(v)]$ uniformly at random and label $\sigma'(v) \defeq \sigma((v,i(v)))$. We now lower bound the expected fraction of hyperedge constraints satisfied by $\sigma'$ as follows
	\begin{align*}
		&\Ex_{e = (v_1,\ldots,v_d) \sim \mu_{\cH_0}} \Ex_{\sigma} \left[\mathbbm{1}\left(\forall a,b \in [d], \pi_{e,v_a}({\sigma}'(v_a)) = \pi_{e,v_b}({\sigma}'(v_b)) \right)\right] \\
		&\Ex_{e = (v_1,\ldots,v_d) \sim \mu_{\cH_0}} \frac{1}{\prod_{i \in [d]} n(v_i)} \sum_{I\in \times^d_{i= 1} [n(v_i)]} \mathbbm{1}\left(\forall a,b \in [d], \pi_{e,v_a}({\sigma}(v_a,I_a)) = \pi_{e,v_b}({\sigma}(v_b,I_b)) \right) \\
		&= \Ex_{e = (v_1,\ldots,v_d) \sim \mu_{\cH_0}} \Ex_{I \sim \times^d_{i = 1} [n(v_i)]} \left[\mathbbm{1}\left(\forall a,b \in [d], \pi_{e_I,(v_a,I_a)}({\sigma}(v_a,I_a)) = \pi_{e_I,(v_{b},I_b)}({\sigma}(v_b,I_b)) \right)\right] \\
		&= \Pr_{e \sim \cD_{\cH},e_I \sim \cC(e)} \left[\forall a,b \in [d], \pi_{e_I,(v_a,I_a)}({\sigma}(v_a,I_a)) = \pi_{e_I,(v_{b},I_b)}({\sigma}(v_b,I_b))\right] \\
		&\geq 1 - \gamma
	\end{align*}
	which gives us the soundness direction.
	
\end{proof}

\subsection{\hug~to \sbug}
\label{sec:hypUG-BipUG}

\begin{theorem}
\label{thm:bipartite-ug}
	Given an instance uniformly weighted $\cH(V,E,[k],\{\pi_{e,v}\}_{e,v},\cD_H)$ of \hug~from Theorem \ref{thm:unif-weights}, we can construct an instance $G_{SB}(V_L,V_R,E_{SB}, \{\pi_{e,v}\},\mu_L)$ of \sbug~such that $|V_R| = n$.  Here $\mu_L: V_L \to \mathbbm{R}_{\geq 0}$ is a measure over left vertices which satisfy the following:
	\begin{itemize}
		\item[(i)] {\bf Completeness}: If there exists a labeling $\sigma:V \to [k]$ which  satisfies $(1 - \epsilon)$-fraction of hyperedges in $\cH$, then there exists a labeling a labeling $\sigma:V_L \cup V_R \to [k]$ such that 
		\[
		\Pr_{u \sim \mu_L}\left[\forall v \in N_{G_{SB}}(u), \pi_{vu}(\sigma(v)) = \sigma(u) \right] \geq 1- \epsilon.
		\]
		\item[(ii)] {\bf Soundness}: If there exists a labeling $\sigma':V_L \cup V_R \to [k]$ such that 
		\[
		\Pr_{u \sim \mu_L}\left[\forall v \in N_{G_{SB}}(u), \pi_{vu}(\sigma'(v)) = \sigma'(u) \right] \geq 1 - \gamma,
		\]
		then there exists a labeling $\sigma''$ which satisfies $(1 - \gamma)$-fraction of hyperedges in $\cH$.
		\item[(iii)] {\bf Right Uniformity:} For every vertex $v \in V_R$ we have $\Pr_{u \sim \mu_L} \left[ v \in N_{G_{SB}}(u) \right] = \frac{d}{|V_R|}$.
		\item[(iv)] {\bf Left Regularity:} The degree of every vertex $v \in V_L$ is $d$.
	\end{itemize}
\end{theorem}
\begin{proof}
Given a \hug~instance $\cH$ as in Theorem \ref{thm:unif-weights}, we construct an instance of \sbug~as follows. 
	Let $\mu_{\cH}: E \to [k]$ denote the probability distribution over the hypergraph constraints in $\cH$.
\begin{itemize}
	\item The right vertex set $V_r$ is the vertex set of $\cH$.
	\item For every \hug~constraint $\hyperedge$, we add a vertex $u_{\hyperedge}$ to $V_L$. 
	\item We define the measure $\mu_L$ as $\mu_L(u_e) = \mu_\cH(e)$ for every hyperedge constraint $e$. 
\end{itemize}

From the above construction, for any right vertex $v \in V_R$  we have 
\[ 
\Pr_{u \sim \mu_L} \left[ v \in N_{G_{SB}}(u) \right] = \Pr_{e \sim \mu_\cH} \left[ v \in e \right] = \frac{d}{n} 
\]
where the last step uses the uniformity guaranteed by the instances output in Theorem \ref{thm:unif-weights}. This implies the right uniformity property of $G_{SB}$. The left regularity of the graph $G_{SB}$ follows from the uniform arity of the hypergraph constraints. Now we analyze completeness and soundness of the reduction. 
	
	\paragraph{Completeness.} Suppose there exists a labeling $\sigma:V \to [k]$ which satisfies at least $1 - \epsilon$-fraction of hyperedge constraints. Then we construct the following labeling $\tilde{\sigma}: V_L \cup V_R \to [k]$. For any $v \in V_R = V$, we assign $\tilde{\sigma}(v) \defeq \sigma(v)$. For any hyperedge constraint $\hyperedge$ satisfied by $\sigma$, pick any $v \in N_{G_{SB}}\left(u_\hyperedge\right)$ and assign $\tilde{\sigma}(u_{\hyperedge}) \defeq \pi_{e,v}(\sigma(v))$. Note that since any such $\hyperedge$ is  satisfied, for any $v,v' \in \hyperedge$ we must have $\pi_{\hyperedge,v}(\sigma(v)) = \pi_{\hyperedge,v'}(\sigma(v'))$, and the choice of the vertex-bijection pair $(v,\pi_{e,v})$ used to label $u_\hyperedge$ does not matter. For the remaining $u_\hyperedge$ vertices, assign labels $\tilde{\sigma}(u_\hyperedge)$ arbitrarily. Therefore,
\begin{align*}
\Pr_{u_\hyperedge \sim \mu_L}\left[\forall v \in N_{G_{SB}}(u_\hyperedge), \pi_{v,u_{\hyperedge}}(\tilde{\sigma}(v)) = \tilde{\sigma}(u_\hyperedge) \right] 
= \Pr_{\hyperedge \sim \cD_\cH}\left[\forall v,v' \in \hyperedge, \pi_{\hyperedge,v}(\sigma(v)) = \pi_{\hyperedge,v}(\sigma(v')) \right] 
\geq 1 - \epsilon.
\end{align*}
\paragraph{Soundness.} Suppose there exists a labeling $\sigma': V_L \cup V_E \to [k]$ such that for a 
left vertex sampled according to $\mu_L$, the probability that all the constraints incident on it are satisfied
is at least $1 - \gamma$. Let $\sigma:V \to [k]$ be the restriction of labeling $\sigma'$ to the right vertex set $V_R$ i.e., for every $v \in V$, we assign $\sigma(v) = \sigma'(v)$. Now pick any hyperedge vertex $u_{\hyperedge} \in V_L$ such that all the constraints incident on it are satisfied by $\sigma'$.
	Let $e$ be supported on vertices $(v_1,v_2,\ldots,v_d)$. Then from the guarantee of $\sigma'$ we have $\pi_{e,v_j}(\sigma(v_j)) = \pi_{e,v_j}(\sigma'(v_j)) = \sigma'(e)$ for every $j \in [d]$ i.e, $\sigma$ satisfies hyperedge constraint $\hyperedge$. Since this holds for any satisfied hyperedge $u_\hyperedge$, at least $1 - \gamma$-measure of hyperedge constraints are satisfied by $\sigma$.
\end{proof}

\subsection{\sbug~to $D$-\sbug}				
\label{sec:BipUG-DBipUG}

Here we sparsify the \sbug~instances constructed in the previous section while approximately preserving the completeness and soundness.

\begin{theorem}				\label{thm:deg-red}
Let $\gamma,\epsilon \in (0,1)$ such that $\gamma \geq \epsilon^2$. 
There exists an efficient randomized procedure, that takes as input a \sbug~instance 
$G = (V_L,V_R,E,[k],\{\pi_{v \to u}\}_{(u,v) \in E},\mu)$ as constructed in Theorem \ref{thm:bipartite-ug}
and outputs a $(1-\epsilon,1 - \gamma)$-\sbug~ instance $G'' = (V_L'',V_R',E',[k],\{\pi_{v \to u}\}_{(u,v) \in E},\mu')$ such that 
the left degree of $G''$ is at most $d$ and the right degree is at most $D = C'' d \ell$, where $\ell = C''\epsilon^{-2} \log k$ (for some large enough constant $C'' > 0$), $\mu'$ is the uniform distribution on the left vertices and 
\begin{itemize}
		\item If $G$ is a YES instance, then there exists a subset $S \subseteq V_L''$ of size at least $\mu'(S) \geq 1 - 4\epsilon$ and a right labeling $\sigma:V \to [k]$ such that all the constraints incident on $S$ are satisfied.
		\item If $G$ is a NO instance, for any right labeling $\sigma:V \to [k]$ at most $(1 - \gamma/4)$-fraction of left vertices (w.r.t measure $\mu'$) have all the constraints incident on them satisfied.
	\end{itemize}
\end{theorem}
\begin{proof}
Let $n$ denote $|V_R|$.
	Our reduction to $D$-\sbug~has two steps. 
\begin{enumerate}
	\item First we construct an intermediate graph $G' = (V_L',V_R,E')$ as follows. For every right vertex $v \in V_R$, we sample a set of $\ell$-neighbors $N_{G'}(v) \subset N_{G}(v)$ where each neighbor is chosen according to $\mu_L(\cdot | v )$ i.e., the distribution over left vertices $u \in V_L$ conditioned on $v \in N_{G_{SB}}(u)$. The multi-set of left vertices is given by $V_L' \defeq \cup_{v \in V} N_{G'}(v)$. Note that for distinct right vertices $v,v'$, $N_{G'}(v) \cap N_{G'}(v')$ may not be empty; if there are $i$-copies of a single left vertex due to $i$-different right neighbors, we treat each of the $i$-copies as distinct vertices. Therefore, overall $|V_L'| = \ell n$. 
    The edge set is $E'$ is the set of edges in the subgraph induced by   $ V'_L \cup V_R$, and we also add the constraints corresponding to the edges in $E'$.
	\item Next, we construct $G'' = (V_L'',V_R',E'', \{\pi_{v \to u}\}_{(u,v) \in E''})$ by removing all right vertices with degree greater than $C'd \ell$, and all their left neighbors. Here $C' > 0$ is large constant to be fixed later.
	\end{enumerate}
	We first analyze the properties of the intermediate \sbug~instance $G'$.
	\begin{lemma}
	With constant probability (over $G'$), the following holds.
	\begin{itemize}
		\item If $G$ is a YES instance, then there exists a subset $U_0' \subseteq V_L'$ of size at least $(1 - 2 \epsilon)|V_L'|$ and a labeling $\sigma:V_L' \cup V_R \to [k]$ such that all the constraints incident on the vertices in $U_0'$ are satisfied.
		\item If $G$ is a NO instance, then for any labeling $\sigma:V_L' \cup V_R \to [k]$, at most $(1 - \gamma/2)$ fraction of vertices in $V_L'$ have the property that all the constraints incident on them are satisfied.
	\end{itemize}
	\end{lemma}

	\begin{proof}
	We argue completeness and soundness of the intermediate \sbug~instance $G'$.
	\paragraph{Completeness.} Let $\sigma:V_L \cup V_R \to [k]$ be an assignment such that there exists a subset $U_0 \subseteq V_L$ of measure at least $(1 - \epsilon)$ on which all the incident constraints are satisfied by the labeling $\sigma$. 
	Let $U_0' \defeq U_0 \cap V_L'$ (here we take multi-set intersection, i.e. if $v \in U_0$ has at least one copy in $V_L'$, then we keep all copies of $v$ in $V_L'$ in $U_0'$). The following  follows directly from the right uniformity of $G$.
	\begin{observation} 			\label{lem:obs}
		Consider the following procedure for sampling a left vertex: (i) sample $v \overset{\rm unif}{\sim} V_R$ and then (ii) sample $u' \sim \mu_L(\cdot | v)$. Then $u'$ is distributed according to $\mu_L$.
	\end{observation}
	\begin{proof}
	Fix a vertex $u \in V_L$. Let $N_{G}(u) = \{v_1,\ldots,v_d\}$. Then 
		\[\Pr_{v,u'}\left[u ' = u\right] = \frac1n\sum_{i \in [d]} \mu_L(u|v_i) = \frac1n\sum_{i \in [d]} \frac{\mu_L(u)}{\mu_L(\{u'':v_i \in N_G(u'')\})} 
		= \frac1n\sum_{i \in [d]} \frac{\mu_L(u)}{d/n} = \mu_L(u) . \]
	\end{proof}		
	
	Then using Observation \ref{lem:obs} and the guarantee of labeling $\sigma$ we have
	\begin{equation}
	\label{eq:u0}
	\Ex_{v \sim \mu_{V_R}} \Pr_{u \sim \mu(\cdot|v)} \left[u \in U_0\right]
	=  \Pr_{u \sim \mu_{V_L}} \left[u \in U_0\right] \geq 1 - \epsilon. 
	\end{equation}
	Now for every vertex $v \in V_R$, let $(v,1),(v,2),\ldots,(v,\ell)$ be the sequence of left neighbors sampled in $N_{G'}(v)$. For every $v \in V_R$ and $i \in [\ell]$, define the indicator random variable $X_{v,i} \defeq \mathbbm{1}((v,i) \in U_0)$. 
	Using equation \ref{eq:u0} we get that 
	\[
	\Ex_{G'}\left[ \frac{1}{n\ell} \sum_{v \in V_R} \sum_{i \in [\ell]}  X_{v,i}\right] = \Ex_{v \sim V_R}  \Ex_{i \in [\ell]}\Ex_{G'}\big[ X_{v,i}\big] = \Ex_{v \sim V_R} \Pr_{u \sim \mu(\cdot|v)} [u \in U_0] \geq 1 - \epsilon.
	\]  
	Since $X_{v,i}$'s are independent $0/1$-random variables, using Hoeffding's inequality we get that 
	\[
	\Pr_{G'}\left[\sum_{v,i} X_{v,i} - \Ex_{G'} \left[\sum_{v,i} X_{v,i} \right]  < -  \epsilon \ell n \right]
	\leq \exp\left( - 2\epsilon^2 \ell n \right) < 0.1.
	\]
	Therefore with probability at least $0.9$ we have 
	\begin{equation}
	\label{eq:u0card}
	\Ex_{v \sim V_R} \Pr_{u \sim N_{G'}(v)}\left[u \in U_0 \cap V_L'\right] \geq 1 - 2 \epsilon. 
	\end{equation}
	Extending $\sigma: V_L \cup V_R \to [k]$ to $\sigma: V_L' \cup V_R \to [k]$ in the natural way (if a vertex
	$v \in V_L$ has multiple copies in $V_L'$, then they all get the label $\sigma(v)$), we get that 
	$\sigma$ satisfies all the constraints incident on $U_0'$ 
	and inequality \ref{eq:u0card} implies that $|U_0'| \geq (1 - 2\epsilon)|V_L'|$.
	
	\paragraph{Soundness.} Fix a right labeling $\sigma:V_R \to [k]$; this defines a set $U^\sigma_0 \defeq \{u \in V_L | \pi_{v \to u}(\sigma(v)) = \pi_{v' \to u}(\sigma(v')) \forall v,v' \in N_G(u)\}$.
	As before, for every $v \in V_R$, let $(v,1),(v,2),\ldots,(v,\ell)$ be the sequence of left neighbors sampled in $N_{G'}(v)$. For every $i \in [\ell]$, define the indicator random variable $X^{(\sigma)}_{v,i} \defeq \mathbbm{1}\left((v,i) \in U^\sigma_0\right)$. Using the soundness guarantee and Observation \ref{lem:obs} we have 
	\[\Ex_G'\left[\frac{1}{n\ell}\sum_{v \in V_R} \sum_{i \in [\ell]}  X^{(\sigma)}_{v,i}\right] =  \Pr_{u \sim V_L} \left[ u \in U^\sigma_0\right] \leq  1 - \gamma.
	\]
	Using Hoeffding's inequality we get that 
	\[
	\Pr\left[\sum_{v,i} X(\sigma)_{v,i} - \Ex \left[ \sum_{v,i} X^{(\sigma)}_{v,i}\right] > \frac{\gamma n \ell}{100}\right]
	\leq \exp\left( - 2\gamma^2n\ell/10^4\right).
	\]
	The number of possible labelings of $V_R$ is $k^n$.
	Taking a union bound over all possible labelings of $V_R$ we get
	\[
	\Pr\left[\exists \sigma : V_R \to [k] \mbox{ s.t } 
	\sum_{v, i} X^{(\sigma)}_{v,i} - \Ex \left[\sum_{v,i} X^{(\sigma)}_{v,i} \right] 
	> \frac{\gamma n \ell}{100}\right]
	\leq \exp\left(n \log k - 2\gamma^2n\ell/10^4\right) < 0.1, 
	\]
	for $\ell = C''\gamma^{-2}\log k$ with $C''> 0$ chosen to be large enough. Therefore,
	with probability at least $0.9$, for every right labeling $\sigma: V_R \to [k]$
	the set $V_L' \cap U^\sigma_0$ (here again, we take multi-set intersection) satisfies 
	\[
		|V_L' \cap U^\sigma_0| =  \sum_{v \in V}\Ex_{v \sim V} \left[|N_{G'}(v) \cap U^\sigma_0|\right] = \ell n\Ex_{v \sim V} \left[\Pr_{u \sim N_{G'}(v)}[u \in U^\sigma_0]\right] \leq (1 - \gamma/2)|V_L'|. \]
	i.e., for any right labeling, at most $1 - \gamma/2$ fraction of left vertices have all the constraints incident on them satisfied by the labeling $\sigma$.

	\paragraph{Putting things together.}
	Using the union bound, we get that the completeness and soundness events hold simultaneously with probability 
		at least $0.8$.
	\end{proof}

Now we analyze $G''$.

\begin{lemma}
With constant probability (over $G'$), the following holds.
	\begin{itemize}
	\item If $G$ is a YES instance, then there exists a subset $S \subseteq V_L''$ of size at least $(1 - 4 \epsilon)|V_L''|$ and a right labeling $\sigma:V'_R \to [k]$ such that all the constraints incident on $S$ are satisfied.
	\item If $G$ is a NO instance, for any right labeling $\sigma:V \to [k]$ at most $1 -  \gamma/4$ fraction of left vertices have all the constraints incident on them satisfied.
	\end{itemize}
\end{lemma}

\begin{proof}
	Fix a right vertex $v' \in V_R$. For every $v \in V_R$, let $N_{G'}(v) = \{(v,1),(v,2),\ldots,(v,\ell)\}$ be the $\ell$-neighbors sampled from the left neighborhood of $v$. Furthermore, for every $v \in  V_R$, and $i \in [\ell]$ define the indicator random variable $Y_{v,i} \defeq \mathbbm{1}(v' \in (v,i))$ i.e., in the variable-constraint interpretation, it indicates whether variable $v'$ appears in the constraint represented by the vertex $(v,i)$. Then we get that 
	\begin{eqnarray*}
	\Ex_{G'} \left[{\rm deg}_{G'}(v') \right] 
	&=& \Ex_{G'}\left[ \sum_{v \in V_R}\sum_{i \in [\ell]} Y_{(v,i)} \right] 
	= n \ell \Ex_{G'}\Ex_{v\sim V_R}\Ex_{i \sim [\ell]} Y_{(v,i)} \\
	&{=}& n \ell \Ex_{v \sim V_R} \Pr_{u \sim N_G(v)} \left[v' \in u\right] \\
	&=& n \ell \mu_{L}\left(\{u: v' \in N_G(u)\}\right) 
		\qquad \textrm{(Using Observation \ref{lem:obs})} \\
	&=& n \ell\left(\frac{d}{n}\right) 
		\qquad \textrm{({\em right uniformity} of $G$ (Theorem \ref{thm:bipartite-ug}))} \\
	&=& \ell d.
	\end{eqnarray*}
	Since $Y_{v,i}$'s are independent $0/1$ random variables, using Chernoff bound, for any $\theta \geq 1$ we have 
	\begin{equation}					
	\label{eq:conc-bound}
	\Pr\left[{\rm deg}_{G'}(v) \geq (1 + \theta)d\ell\right] \leq \Pr\left[\sum_{v \in V_R}\sum_{i \in [\ell]} Y_{(v,i)} > (1 + \theta) d \ell\right] \leq \exp\left(-\frac{ \theta^2 d \ell}{4}\right).
	\end{equation}
	For any $i \in \mathbbm{N}$, let $m(i)$ be the random variable which denotes the number of right vertices with degree exactly $i$ in $G'$. Note that the above calculation implies that $\Ex_{G'}\left[n((1+\theta) d\ell)\right] \leq e^{-\theta^2d\ell/4} n$. Note that $m(i)$ is a strictly decreasing function of $i$. Therefore  the expected number of left vertices incident on large right degree vertices can bounded as follows
	\begin{align*}
	\Ex_{G'}\left[\sum^\infty_{i \geq C' d\ell} i m(i)  \right] 
	&\leq 	\sum^\infty _{ \theta = C' - 1} ((1 + \theta)d\ell + d\ell) \Ex_{G'}\left[m((1 + \theta)d \ell)\right]  \tag{Rounding the sum into intervals of length $d\ell$ }\\ 
	&\leq  \sum^\infty _{ \theta = C' - 1} 2(1 + \theta)d \ell e^{-\theta^2 d \ell/4} n  \tag{Since $\theta \geq C' - 1 \geq 1$}\\
	&\leq \int^\infty_{ \theta = C' - 1} 2(1 + \theta)d \ell n e^{-\theta^2 d \ell/4} d\theta \\
	&\leq e^{-d\ell}n
	\end{align*}
	where the last step follows for any large enough choice of $C'$. Therefore using Markov's inequality, with probability at least $0,9$, the number of left vertices deleted is at most $10e^{-d\ell} n \leq \epsilon n \leq  \gamma^2n$ (by choosing $\ell \geq \epsilon^{-2}\log(k) \geq\log(1/\gamma)$). For brevity, given a right labeling $\sigma:V \to [k]$, we say that left vertex $u \in V_L$ is {\em good} w.r.t. labeling $\sigma$, if $\sigma$ satisfies all the edges incident on it. We argue completeness and soundness for $G''$. 

\paragraph{Completeness.} Consider the right labeling $\sigma: V_R \to [k]$ for which the fraction of good left vertices in $G'$ is at least $(1 - 2\epsilon)$. From Eq. \ref{eq:conc-bound} we know that the truncation step removes at most $e^{-Cd \ell}$-fraction of left vertices. By our choice of parameters, $e^{-Cd \ell} \leq \epsilon$.
Therefore, the fraction of left vertices that are good in $G''$ is at least $(1 - 3\epsilon)/(1 - \epsilon) \geq 1 - 4\epsilon$.
	
\paragraph{Soundness.}
	Let $\sigma:V_L'' \cup V_{R}' \to [k]$ be a labeling for which there exists a good set $S \subset V_L''$ of size at least $(1 - \gamma/4)|V_L''|$. Construct $\sigma':V_L' \cup V_{R} \to [k]$ by assigning  $\sigma'(x) = \sigma(x)$ whenever $x \in V_L'' \cup V_R'$, otherwise we assign $\sigma'(x)$ arbitrarily. Observe that since in the construction of $G''$ from $G'$ we delete all the left vertices which are neighbors of the deleted right vertices, for every $u \in V_L''$ we have $N_{G''}(u) = N_{G'}(u)$. In particular, if $\sigma$ satisfies all the constraints incident on a left vertex $u \in V_L''$, $\sigma'$ satisfies all incident constraints on $u$ in $G'$. 
	Therefore $\sigma'$ satisfies all incident constraints for any $u \in S$. Furthermore, from our choice of $\ell$ we have 
	\[
	\frac{|S|}{|V_L'|} \geq \frac{(1 - \gamma/4)|V_L''|}{|V_L'|}  = (1 - \gamma/4)\frac{|V_L'| - |V_L' \setminus V_L''|}{|V_L'|}\geq (1 - \gamma/4) \frac{\ell n - e^{-d\ell} n}{\ell n}  \geq (1 - \gamma/3)
	\]
	which gives us the lower bound on the fraction of left vertices for which $\sigma'$ satisfies all constraints incident on it.
\end{proof}

\end{proof}

\subsection{\sbug~ to \strug}				\label{sec:DBigUG-StrongUG}

The last step uses the following lemma.
\begin{lemma}[\cite{KR08}, Lemma 3.8]				\label{lem:strong-ug}
	There exists an efficient procedure that given an unweighted \sbug~ instance $G = (V_L,V_R,E,[1],\{\pi_{v \to u}\}_{(u,v) \in E})$ with left and right degrees bounded by $d$ and $D$ respectively, outputs a UG instance $\psi(V,E',[k],\{\pi_{e,v}\}_{e,v})$ with the following properties.
	\begin{itemize}
		\item If there exists a labeling $\sigma: V_L \cup V_R \to [k]$ such that $(1 - \epsilon)$-fraction of left vertices in $G$ have all their incident constraints satisfied, then there exists a labeling to $\psi$ and a set $S \subseteq V$ of $(1 - \epsilon)$-fraction of vertices such that {\em all constraints induced in $S$} are satisfied.
		\item For any $\gamma > 0$, if there exists a labeling $\sigma'$ to $\psi$, and a set $S \subseteq V$ of $(1 - \gamma)$-fraction of variables such that {\em all constraints induced in $S$}  are satisfied, then there exists a labeling $\sigma$ to $G$ such that at least  $(1 - \gamma)$-fraction of left vertices in $G$ have all their constraints satisfied.
	\end{itemize}
	Moreover the degree of the vertices in $\psi$ is at most $dD$.
\end{lemma}

While the proof of the above lemma is identical to Lemma 3.8 in \cite{KR08}, we include the proof here because of two reasons. The statement of Lemma 3.8 does not give a bound on the degree of the \strug~instance, and therefore the reduction is needed to give a bound on the degree. Furthermore, our reduction to \oct~builds on this lemma, and its analysis is implicitly used in giving the guarantees for the reduction.

\begin{proof}
	Given $\cG(V_L,V_R,E,[k],\{\pi_{e,v}\})$ we construct the \strug~ instance $\psi$ as follows. The vertex set is $V = V_L$. Furthermore, for any $e_1,e_2 \in V_L$ and any $v \in N_G(e_1) \cap N_G(e_2)$ we introduce a constraint $\pi^{(v)}_{e_1 \to e_2} : = \pi_{v,e_2}\circ \pi^{-1}_{v,e_1}$. Note that for any pair of vertices $e_1,e_2$ we add a constraint for every vertex commonly adjacent to $e_1$ and $e_2$ i.e, the number of constraint edges between two left vertices $e_1,e_2$ is  $|N_G(e_1) \cap N_G(e_2)|$. In particular, the number of constraints incident on any vertex $v \in V_L$ is at most $dD$. Now we argue completeness and soundness of the above reduction.
	
	{\bf Completeness}: Suppose there exists a labeling $\sigma:V_L \cup V_R \to [k]$ and a set $S \subset V_L$ such that $\sigma$ satisfies all the constraints incident on $S$. We claim that the labeling $\sigma' := \sigma|_{V_L}$ satisfies all the constraints incident on vertices in $S$ in $\psi$. To see this, fix any $e_1,e_2 \in V_L$ such that they share a constraint $\pi = \pi^{(v)}_{e_1 \to e_2}$. Note that by guarantee of $\sigma$ we have $\sigma(e_1) = \pi_{v,e_1}(\sigma(v))$ and $\pi_{v,e_2}(\sigma(v)) = \sigma(e_2)$. Therefore,
	\[
	\pi^{(v)}_{e_1 \to e_2}(\sigma'(e_1)) = \pi_{v,e_2} \circ \pi^{-1}_{v,e_1} \circ \pi_{v,e_1}(\sigma(v)) = \pi_{v,e_2}(\sigma(v)) = \sigma'(e_2)
	\]
	i.e, $\sigma'$ satisfies the constraint $\pi^{(v)}_{e_1 \to e_2}$. Since the above arguments hold for any choice of $e_1,e_2 \in S$ with respect to the labeling $\sigma'$ the claim follows.
	
	{\bf Soundness}: Suppose there exists a labeling $\sigma$ and a subset $S$ of $1 - \gamma$ fraction of vertices in $V$ such that {\em all} constraints induced among $S$ are satisfied by $\sigma'$. Now consider the following decoding procedure. 
	\begin{enumerate}
		\item For every vertex $e \in S$, and $v \in N_\cG(e)$ assign $\sigma(v) = \pi^{-1}_{e,v}(\sigma'(e))$. As we shall show immediately, the choice of the left vertex used to label the right vertex will not matter. Furthermore, for every such $e$, assign $\sigma(e) = \sigma'(e)$.
		\item For any left over vertices in $V_L \cup V_R$, assign labels arbitrarily.
	\end{enumerate}
	
	We shall now argue that (i) Step $1$ above does not lead to inconsistent labelings to vertices in $V_R$ and (ii) the labeling $\sigma$ satisfies all the constraints incident on $S$ in $V_L$. Observe that assuming (i), (ii) follows immediately by definition of the decoding procedure.
	
	Hence it suffices to just show (i). Fix a right vertex $v \in V_R$ and $e_1,e_2 \in N_G(v)$ such that $e_1,e_2 \in S$. Then the labeling $\sigma'$ satisfies all the constraints within $e_1,e_2$ and in particular, $\pi^{(v)}_{e_1 \to e_2}$ i.e.,
	\[
	\pi^{(v)}_{e_1 \to e_2}(\sigma'(e_1)) = \sigma'(e_2) \Rightarrow \pi^{-1}_{e_2,v}(\sigma'(e_2)) = \pi^{-1}_{e_1,v}(\sigma'(e_1)) 
	\]	
	Therefore, decoding using $e_1$ and $e_2$ assigns the same labeling to the vertex $v$. Since this is true for any right vertex decoded using vertices in $S$, the claim follows.
\end{proof}

\subsection{Proof of Theorem \ref{thm:strug-red}}

Now we chain the various hardness results established in the previous sections to prove Theorem \ref{thm:strug-red} in three steps. We begin with an instance $\cG$ of $(1-\epsilon,\gamma)$-\ug~ and reduce it to an instance of $(1 - 2\epsilon,1 - O(\sqrt{\epsilon \log d \log k}))$-\hug~ using Theorem \ref{thm:hyperug}. Then combining Theorems \ref{thm:unif-weights}, \ref{thm:bipartite-ug} and \ref{thm:deg-red}, we reduce $(1 - 2\epsilon,1 - O(\sqrt{\epsilon \log d \log k}))$-$d$-\sbug~ such that the left degree of the output instances is bounded by $d$, and the right degree is bounded by at most $C''d\ell$. Finally, using Lemma \ref{lem:strong-ug}, we reduce the above instances to $(1 - C\epsilon, 1- C'\sqrt{\epsilon \log d \log k})$-\strug~ where the degree of the vertices are bounded by $d^2\ell \leq d^3$ (since by our choice of parameters we have  $d \geq \log k/\epsilon^2$ ). This completes the proof of Theorem \ref{thm:strug-red}.

\section{The \oddcycletransversal~problem}			
\label{sec:oct}

\begin{proof}[Proof of Corollary \ref{cor:oct}]
We first establish that~\oddcycletransversal~is an instance of~\stronguniquegames~with $k = 2$. Given an instance $G = (V,E)$ of \oddcycletransversal, construct the \stronguniquegames~instance $\cG (V,E,\mathbbm{F}_2,[2],\{\pi_e\}_{e \in E})$ with the constraint $\sigma(u) \neq \sigma(v)$ whenever $\{u,v\} \in E$. We claim that for any subset $S \subseteq V$, we have $\cV(\cG[S]) = 1$ if and only if $G[S] $ is bipartite. Indeed, suppose $G[S]$ is bipartite with bipartition $S = S_1 \uplus S_2$. Let $\sigma(u) = \mathbbm{1}(u \in S_1)$. Then, by construction, the only constraints in $\cG[S]$ are the ones crossing the cut $(S_1,S_2)$. Furthermore, for any $\{u,v\} \in E_G(S_1,S_2)$ we have $\sigma(u) \neq \sigma(v)$ (by construction of $\sigma$) i.e., the labeling $\sigma$ satisfies all the constraints in the induced game $\cG[S]$. Therefore, $\cV\left(\cG[S]\right) = 1$.

Conversely, suppose there exists subset $S \subset V $ such that $\cV\left(\cG[S]\right) = 1$. Let $\sigma : S \to \mathbbm{F}_2$ be the labeling that satisfies all the constraints induced in $\cG[S]$. Let $S_1 := \left\{u \in S | \sigma(u) = 1\right\}$ and $S_2 := S \setminus S_1$. Since all the constraints are of the form $\sigma(u) \neq \sigma(v)$, it follows that $S_1,S_2$ must be independent sets. Since $S = S_1 \uplus S_2$, therefore $G[S]$ must be bipartite.

Therefore, combining the two directions, it follows that \oddcycletransversal~is an instance of \stronguniquegames~with alphabet sets $\mathbbm{F}_2$. Combining this observation with the guarantees of  Theorem \ref{thm:partial-ug-2} (instantiated with $k = 2$) gives us the claim.
\end{proof}

\subsection{Hardness for the \oddcycletransversal~problem}

In this section, we prove Theorem \ref{thm:oct-hardness}.
This follows from the following theorem.
\begin{theorem} \label{thm:oct-red}
	Let $G = (V,E,\{\pm 1\},\{\pi_{v \to u}\}_{(u,v) \in E})$ be a \strug instance with degree bounded by $d$. Then there exists an efficient procedure that given $G$ outputs a \strug~instance $\psi(V',E',\{\pm 1\},\{\pi_{e,v}\}_{e,v})$ which satisfies the following guarantees for any choice of $\epsilon,\gamma \in (0,1)$.
	\begin{itemize}
		\item If $\sval{(G)} \geq 1 - \epsilon$, then $\sval{(\psi)} \geq 1 - \epsilon$.
		\item If $\sval{(\psi)} \geq 1 - \gamma$, then $\sval{(G)} \geq 1 - 2\gamma$.
	\end{itemize}
	Additionally, the degree of constraint graph here is bounded by $d + 1$, and all constraints in $\psi$ are {\sf NEQ} constraints. 
\end{theorem}
\begin{proof}
	
	Given \strug instance $G$, we construct an \oct~instance $\psi$ as follows.

	{\bf Vertex Set}. For every $e \in V$, introduce two vertices $e^+,e^-$. Let $V^+ = \{e^+ | e \in V_L\}$ and $V^- = \{e^- | e \in V_L\}$, and finally we define the vertex set to be $V' = V^+ \cup V^-$.
	
	{\bf Constraints}. For every vertex $e \in V$ we add a {\sf NEQ}  constraint between $e^+,e^-$. Furthermore for every edge constraint $(e_1,e_2) \in E$, we add constraints in the following way:
	\begin{itemize}
		\item If $\pi_{e_1 \to e_2} $ is an {\sf EQ} constraint, we add {\sf NEQ} constraints between $(e^+_1,e^-_2)$ and $(e^-_1,e^+_2)$.
		\item If $\pi_{e_1 \to e_2} $ is an {\sf NEQ} constraint, we add {\sf NEQ} constraints between $(e^+_1,e^+_2)$ and $(e^-_1,e^-_2)$.
	\end{itemize}
	Clearly $\psi$ consists of only {\sf NEQ} constraints. We now bound the degree of a vertex $e^\alpha_i \in V$. Firstly it has an edge with $e^{-\alpha}_i$. Furthermore, $e^\alpha_i$ has an edge incident on it for every $e_j$ which is incident on $e_i$ in the constraint graph $G$  (this also takes into account multiple edges that it can share with a single vertex). Therefore we can bound the degree by $1 + {\rm deg}_G(e_i) \leq 1 + d$. Next, we argue the completeness and soundness of the reduction. 
	
	{\bf Completeness}: Suppose there exists a labeling $\sigma:V\to \{\pm 1\}$ for which there exists a set $S \subset V$ for which all incident constraints in $G$ are satisfied. Let $S^+ = \{e^+ | e \in S\}$ and $S^- = \{e^- | e \in S\}$. Construct a labeling $\sigma' : V' \to \{\pm 1\}$ as follows. For every vertex $e \in S$, assign $\sigma'(e^+) = \sigma(e)$ and $\sigma'(e^-) = - \sigma(e)$. For every $e \in V \setminus (S^+ \cup S^-)$, we assign labels arbitrarily. We claim that the labeling $\sigma'$ must satisfy all constraints in $S' = S^+ \cup S^-$. Indeed, let $(e,e')$ be a constraint in $\psi[S']$. There are two cases:
	
	{\bf Case (i)} Suppose $e = e^+$ and $e' = e^-$ for some $e \in E$. Then clearly $\sigma'(e^+) = \sigma(e) = -\sigma'(e^-)$.
	
	{\bf Case (ii)} Suppose $e = e^\alpha_i$ and $e' = e^\beta_j$ for some $\alpha,\beta \in \{+,-\}$ and $i \neq j$. Consider the case $\alpha \neq \beta$, and without loss of generality let $\alpha = +, \beta = - $ (the other cases can be shown similarly). Then $(e_i,e_j)$ must be a {\sf EQ} constraint. Then
	\[
	\sigma'(e) =  \sigma'(e^+_i) = \sigma(e_i) \overset{1}{=}  \sigma(e_j) = -\sigma'(e^-_j) = -\sigma'(e'). 
	\]
	Here the first and last equalities follow using the definition of the labeling and  step $1$ uses the fact $(e_i,e_j)$ is an {\sf EQ} constraint in $G$. Furthermore by definition of $S'$, we must have $e_i,e_j \in S$, and therefore the labeling $\sigma$ must satisfy the constraints induced on $e_i,e_j$ including $(e_i,e_j)$.
	
	Similarly consider the case when $\alpha = \beta$. Suppose $\alpha = +$ (again the case $\alpha = -$ can be argued similarly). Then $(e_i,e_j)$ must be a {\sf NEQ} constraint and  therefore 
	\[
	\sigma'(e) = \sigma(e_i) \overset{1}{=}  -\sigma(e_j) = -\sigma'(e^+_j) = -\sigma'(e').
	\] 
	Again, in step $1$ uses the fact that $e_i,e_j$ in $S$, with $(e_i,e_j,)$ being a {\sf NEQ} constraint in $G$. This along with the guarantee of the labeling $\sigma$ gives us the two steps.
	
	The above arguments imply that all constraints induced on the set $S'$ are satisfied by the labeling $\sigma'$. Finally observe that $|S'| = 2|S| = 2(1 - \gamma)|V_R| = (1 - \gamma)|V|$, which establishes the completeness direction of the reduction.
	
	{\bf Soundness}. Suppose there exists a labeling $\sigma:V' \to \{\pm 1\}$ for which there exists a set $S \subset V'$ of size $(1 - \gamma)|V'|$ such that all constraints in $\psi[S]$ are satisfied by the labeling $\sigma$. Now let $\wh{V} = \{e \in V_L | e^+,e^- \in S\}$. Then 
	\[
	|\wh{V}| \geq |V| - |\{e: e^+ \in V' \setminus S\}| - |\{e: e^- \in V' \setminus S\}| \geq |V| - |V' \setminus S| \geq |V| - \gamma|V'| = (1 - 2\gamma)|V|.
	\]
	Construct a labeling $\sigma' : V \to \{\pm 1\}$ as follows. For every $e \in \wh{V}$, let $\sigma'(e) = \sigma(e^+)$. For the remaining vertices in $V \setminus \wh{V}$, we assign labels arbitrarily.  We claim that the labeling $\sigma'$ satisfies all constraints induced in $G[\wh{V}]$. To see this, fix a pair of vertices $e_i,e_j \in \wh{V}$ such that $(e_i,e_j)$ has a  constraint $\pi_{e_1 \to e_2}$in $G[\wh{V}]$. We now consider the following cases.
	
	{\bf Case (i)}: Suppose $\pi_{e_1 \to e_2}$ is an {\sf EQ} constraint. Now from the construction of the labeling $\sigma'$ we get that 
	\[
	\sigma'(e_i) = \sigma(e^+_i) \overset{1}{=} - \sigma(e^-_j) \overset{2}{=} \sigma(e^+_j) = \sigma'(e_j)
	\] 
	i.e., $\sigma'$ satisfies the equality constraint. Here steps $1$ and $2$ follow from the fact that since $(e_i,e_j) \in G[\wh{V}]$ is an {\sf EQ} constraint, there are {\sf NEQ} constraints on 
	$(e^+_i,e^-_j)$ and $(e^+_j,e^-_j)$ in $\psi$. Furthermore since $e^+_i,e^+_j,e^-_j \in S$, these constraints must be satisfied by the labeling $\sigma$. 
	
	{\bf Case (ii)}: Suppose $\pi_{e_1 \to e_2}$ is a {\sf NEQ} constraint. Again note that $e^+_i,e^+_j \in S$ and there is a {\sf NEQ} constraint on $(e^+_i,e^+_j)$ in $\psi$. Therefore we have 
	\[
	\sigma'(e_i) = \sigma(e^+_i) = - \sigma(e^+_j) = - \sigma'(e_j)
	\]
	i.e., the {\sf NEQ} constraint on $(e_i,e_j)$ in $G$ is satisfied by the labeling $\sigma'$.

\end{proof}

\paragraph{Acknowledgements.}
We thank Prahladh Harsha for pointing us to the \stronguniquegames~problem. AL was supported in part by SERB Award ECR/2017/003296 and a Pratiksha Trust Young Investigator Award.
\bibliography{main}
\bibliographystyle{amsalpha}
\appendix
\part{Appendix}
\section{Proof of Proposition \ref{prop:ssve-ug}}
Suppose $\sval(\cG) \geq 1 - \epsilon$. Then there exists a set $S \subset V$ such that $\sval(\cG[S]) = 1$ and $|S| \geq n(1 - \epsilon)$. Let $\sigma:S \to [k]$ be the partial labeling under which $\cG[S]$ is completely satisfiable. Consider the corresponding set $S'$ in the label extended graph $G$ defined as $S' = \left\{(a,\sigma(a)) | a \in S \right\}$. Clearly by construction, if $a,b \in E$ and $(a,b) \in \cG[S]$, then $\{(a,\sigma(a)),(b,\sigma(b))\} \in E\left(G[S']\right)$. Therefore, for any $(w,l) \in V'$,
\[
(w,l) \in \partial^{V}_{G}(S') \implies w \notin S
\]
which in turn implies that $\left|\partial^V_{G}(S')\right| \leq \left|\cup_{v \in S} \cC_v\right| \leq \epsilon n k$.

Now for the other direction, suppose there exists $S' \subset V$ which is a non repeating of size at least $(1 - \epsilon)n$ and $\partial^{V}_G(S') \leq \delta n$. Let
$S'' = S' \setminus \cup_{v \in {\rm Vert}\left(\partial^V_{G}(S')\right)} \cC_v$. Let $S := {\rm Vert}(S'')$ and $\sigma : S \to [k]$ be the corresponding labeling to the vertices in $S$ given by $S''$. Again, for any $(a,b) \in  E\left(\cG[S]\right)$ we have $\{(a,\sigma(a)),(b,\sigma(b))\} \in E\left(G[S'']\right)$, and therefore, $\pi_{ab}(\sigma(a)) = \sigma(b)$. Since this holds for any $a,b \in S $ such that $(a,b) \in E_{\cG}$, we have $\cV\left(\cG[S]\right) = 1$, and therefore, $\cV(\cG) \geq 1 - \epsilon \delta$.

\section{Improved Separators} \label{sec:separator}

The proof of Theorem \ref{thm:partial-ug-3} uses the following improved construction of $\ell^2_2$-separators.

\begin{theorem}  
	\label{thm:hos-3}
	Let $d \geq 2$ be an integer and let $H = (V,E)$ be a hypergraph with arity bounded by $d$. Then there is a randomized algorithm that given a set
	of vectors $\set{\bar{u}: u \in V}$ satisfying $\ell_2^2$ triangle inequalities,
	parameters $m \geq 2$ and $\beta \in (0,1)$, generates a hypergraph $m$-orthogonal separator with 
	probability scale $\alpha \geq 1/n$ and distortion $D = \bigo{\beta^{-1} m \log m \log \log m \sqrt{\log d M (1/\beta)}}$.
	Here $M = \sum_{e \in E} \max_{u,v \in e}\|\bar{u} - \bar{v}\|^2$. The algorithm runs in time ${\sf poly}(|E|,|E|/M)$.
\end{theorem}

Before we prove Theorem \ref{thm:hos-3}, we remark that plugging in its hypergraph orthogonal separator into the algorithm for Theorem \ref{thm:hos-1}, and repeating the analysis as is would give us the required bound in Theorem \ref{thm:partial-ug-3}. In particular, as in Theorem \ref{thm:hos-2}, we use the observation that if the constraint graph has degree bounded by $d$, then the arity of the label extended hypergraph is also bounded by $\bigo{d}$. Plugging in the upper bound on arity will give us the claim.

The construction in Theorem \ref{thm:hos-3} differs in the following key point. It uses the following improved variant of the basic {\em single scale embedding theorem} which is obtained by combining ideas from \cite{ALN08} and \cite{FLH08}. 

\begin{theorem}				
	\label{thm:separators}
	There exists constants $C > 0$ and $p \in (0,1)$ such that the following holds. Let $H = (V,E)$ be a hypergraph with arity bounded by $d$. Furthermore, suppose there exists a  map $\vphi:V \to \mathcal{S}^{d-1}$ (where $\mathcal{S}^{d-1}$ is the $\ell_2$-unit sphere in $d$-dimensions) such the map satisfies the $\ell^2_2$ triangle inequality i.e,. $\|\vphi(i)- \vphi(j)\|^2 \leq \|\vphi(i) - \vphi(k)\|^2_2 + \|\vphi(j) - \vphi(k)\|^2_2 $ for every triple $i,j,k \in V$. Then there exists  a distribution $\mu$ over sets $U \subseteq V$ such that the following holds. For every $(i,j) \in V \times V$ such that $\|\vphi(i) - \vphi(j)\|^2_2 \geq \Delta/16$ 
	\[
	\Pr_{U \sim \mu}\left[j \in U \mbox{ and } d_{\ell^2_2}(i,U) \geq \frac{C\Delta}{2\sqrt{\log (dm/\Delta)}} \right] \geq p
	\]
	where $m = \sum_{e \in E} \max_{i,j \in e}\|\vphi(i) - \vphi(j)\|^2_2 $.
\end{theorem}

The above theorem is obtained by combining the single-scale embedding theorem from 
\cite{ALN08} with the following key idea from \cite{FLH08}. Suppose one can construct 
an $\epsilon$-net of small cardinality for $(V,d_{\ell^2_2})$ in the $\ell^2_2$-metric. Then one can instantiate the single scale embedding theorem with $\epsilon$-net (instead of the full vertex set), and then observe that quantitatively similar guarantees can also be obtained for the full set as well (since the $\epsilon$-net is a coarse representation of the full set, and the single scale embedding theorem only needs to work with well separated points, and therefore the coarse representation suffices). Therefore, most of the work here is to show that one can obtain $\epsilon$-net of size $\tilde{O}(\sdp)$, which is done using arguments similar in spirit to \cite{FLH08}.

\subsection{Efficient $\epsilon$-nets in $\ell^2_2$}

Here we shall show that one can construct small sized $\epsilon$-net with respect to the  $\ell^2_2$-metric. We begin by defining $\epsilon$-nets for metric spaces.
\begin{definition}[$\epsilon$-net]
	Given a metric space $(X,d)$, a subset $X' \subseteq X$ is an $\epsilon$-net of $X$ if for every $x \in X$, there exists a point $x' \in X'$ such that $d(x,x') \leq \epsilon$.
\end{definition}
We shall also need the definition of the shortest path metric with respect to the hypergraph $H = (V,E)$. Given a map $\vphi:V \to \mathbbm{R}^d$ which satisfies the $\ell^2_2$-triangle inequality, we define the shortest path metric $d_\sdp$ on the $H$ as follows. For any pair of vertices $(i,j)$, we say that an ordered list of hyperedges $\{e_1,e_2,\ldots,e_R\}$ is an $(i,j)$ path if it satisfies the following conditions. 
\begin{itemize}
	\item $i \in e_1, j \in e_R$
	\item For every $a \in [R-1]$, we have $e_a \cap e_{a + 1} \neq \emptyset$.
\end{itemize}
For every $e \in E_H$, let $w(e) = \max_{a,b \in e}\|\vphi(a) - \vphi(b)\|^2$, and for any path $P$ we extend the definition as $w(P) = \sum_{e \in P} w(e)$. Now we are ready to define the shortest path metric with respect to the hypergraph $H$. 
Let $\cP_{ij}$ denote the set of all $(i,j)$ paths in $H$.
\begin{enumerate}
	\item For every $i \in V$, define $d_\sdp(i,i) \defeq 0$.
	\item For $i,j \in V$ such that $i \neq j$ we define 
	\[
	d_{\sdp}(i,j) \defeq \min_{P \in \cP_{ij}} w(P).
	\]
\end{enumerate}

We verify that $d_\sdp$ is a metric on $V$.
\begin{lemma}
	$d_{\sdp}:V \times V \to \mathbbm{R}_{\geq 0}$ is a metric.
\end{lemma}
\begin{proof}
	Towards showing that $d_\sdp$ satisfies the triangle inequality, fix vertices $i,j,k \in V$. We assume that $k \notin \{i,j\}$ (otherwise the triangle inequality trivially follows). Let $P_{ik} \in \cP_{ik}, P_{kj} \in \cP_{kj}$ be such that $d_\sdp(i,k) = w(P_{ik})$ and $d_\sdp(k,j) = w(P_{kj})$. Let $P = P_{ik} \hat{\odot} P_{kj}$ be the path obtained by concatenating the two paths; by definition it follows that $P \in \cP_{ij}$ and therefore $d_\sdp(i,j) \leq w(P) = w(P_{ik}) + w(P_{kj}) \leq d_{\sdp}(i,k) + d_{\sdp}(k,j)$. Since this holds for any arbitrary triple of vertices, the claim follows.
\end{proof}

\begin{remark}
	The shortest path metric $d_{\sdp}(\cdot,\cdot)$ is polynomial time computable on hypergraphs using dynamic programming.
\end{remark}

Now our first observation is that the $d_{\ell^2_2}$-metric is upper bounded by $d_{\sdp}$.

\begin{lemma}
	For any pair of vertices $(i,j) \in V \times V$ we have $d_{\ell^2_2}(i,j) \leq d_{\sdp}(i,j)$.
\end{lemma}
\begin{proof}
	When $i = j$, we trivially have $d_{\ell^2_2}(i,j) = d_\sdp(i,j) = 0$. Otherwise let $i \neq j$, in which case there exists $P = \{e_1,e_2,\ldots,e_R\}\in \cP_{ij}$ such that $w(P) =  d_\sdp(i,j)$. Let $v(0) = i$, $v(R) = j$, and for every $a \in [R-1]$, fix $v(a) \in e_a \cap e_{a+1}$. Since the vectors $\vphi(x): x \in V$ satisfy the $\ell^2_2$-triangle inequality, we have 
	\begin{eqnarray*}
		\|\vphi(i) - \vphi(j)\|^2_2 &\leq& \sum_{a = 0}^{R-1}\|\vphi(v(a)) - \vphi(v(a+1))\|^2 \\
		&\leq& \sum_{a \in [R]}\max_{p,q \in e_a}\|\vphi(p) - \vphi(q)\|^2 \\
		&\leq& \sum_{a \in [R]} w(e) = w(P) = d_\sdp(i,j).
	\end{eqnarray*}
	Since the above holds for every $(i,j) \in V \times V$, the claim follows.
\end{proof}

The second observation used is that one can efficiently construct a greedy $\epsilon$-net of small size in the $d_{\sdp}$ metric.
Recall that for $x \in V$, the ball of radius $\epsilon$ around $x \in V$ is defined as 
$\cB_\sdp(x,\epsilon) \defeq \{v \in V : d_{\sdp}(x,v) \leq \epsilon\}$.
For any $V' \subseteq V$, and $x \in V$, $d_{\sdp}(x,V') \defeq \min_{v \in V'}d_{\sdp}(x,v)$.

\begin{lemma}
	\label{lem:greedy-ball}
	Let the SDP value be denoted by $M$ and $d$ be the max-arity of the hypergraph. 
	There exists a polynomial time algorithm to construct an $\epsilon$-net $V'$ of $V$ such that $|V'|\leq 4ZdM/\epsilon$.
\end{lemma}
\begin{proof}
	Consider the following greedy procedure:
	\begin{enumerate}
		\item Initialize $V' \gets \emptyset$ and $S \gets V$.
		\item If there exists $x \in S$ such that $d_\sdp(x,V') > \epsilon$, update $V' \gets V' \cup\{x\}$, and update $ S \gets S \setminus \cB_\sdp(x,\epsilon)$. Otherwise terminate.
	\end{enumerate}
	Clearly the termination condition ensures that $V'$ is an $\epsilon$-net for $V$, so all that remains is to bound $|V'|$. Let $t = |V'|$ denote the total number of iterations. For every $i \in [t]$, let $V'_i,S_i$ be the $V',S$ sets at the beginning of the $i^{th}$ iteration and let $x_i$ be the vertex added during the $i^{th}$ iteration.  Furthermore, for every $i \in [t]$,let $B_i \defeq \cB_\sdp(x_i,\epsilon/4) \cap S_i$, and $\partial B_i \defeq \partial^E_H(B_i)$ be the hyperedge boundary of $B_i$. Finally we use $E(B_i)$ to denote the set hyperedges in the induced sub-hypergraph $H[B_i]$. 
	\begin{claim}
		The sets $B_1,\ldots,B_t$ are vertex disjoint.
	\end{claim}
	\begin{proof}
		Suppose not. For contradiction, let $a,b \in [t] : a < b$ and $B_a \cap B_b \neq \emptyset$. Then fix $z \in B_a \cap B_b$. By triangle inequality we have $d_{\sdp}(x_a,x_b) \leq d_{\sdp}(x_a,z) + d_{\sdp}(z,x_b) \leq \epsilon/2$ which implies that $x_b \in B_a = \cB_{\sdp}(x_a,\epsilon)$. But then $x_b \notin S_b$ which is a contradiction.
	\end{proof}
	We claim that for every $i \in [t]$ we have 
	\[
	w(E(B_i) \cup \partial B_i) \geq \epsilon/4 .
	\]
	To see this, fix an $i \in [t]$. Recall that $x_i$ is the vertex added in iteration $i$. Let $P_i = (e_1,\ldots,e_r)$ be a shortest $(V'_i,x_i)$-path such that 
	\[
	e_1 \cap V'_i \neq \emptyset,  \qquad\qquad x_i \in e_r, \qquad\qquad w(P_i) = d_\sdp(V'_i,x_i) > \epsilon.
	\]
	Since $d_\sdp(V'_i,x_i) > \epsilon$, we must have 
	\[
	P_i \cap \partial B_i = P_i \cap \partial^E_H(\cB_{\sdp}(x_i,\epsilon/4)) \neq \emptyset
	\]
	Let $e_j \in P_i \cap \partial B_i$ and $P'_i = (e_j,e_{j+1},\ldots,e_r)$ be the sub-path starting from $e_j$. 
	\begin{claim}
		There exists $x' \in e_j$ such that $d_{\sdp}(x',x_i) > \epsilon/4$
	\end{claim}
	\begin{proof}
		Suppose not, then for every $x' \in e_j$ we have $d_{\sdp}(x',x_i) \leq \epsilon/4$ and therefore $e_j \in E(B_i)$ which is a contradiction.
	\end{proof}
	Fix a $x'$ as guaranteed in the above claim. Since $P'_i$ is an $(x',x_i)$ path and we have $w(P'_i) \geq d_{\sdp}(x',x_i) > \epsilon/4$. Finally note that for every $e \in P_i'$, we have $e \in  E(B_i)$ or $e \in \partial B_i$ and therefore 
	\[
	w(B_i \cap \partial B_i) \geq w(P'_i) \geq \epsilon/4 .
	\] 
	Finally we bound the $t$ as follows. Consider the multi-set ${\bf S} = \cup_{i \in [t]} (E(B_i) \cup \partial B_i)$. Since the sets $B_1,\ldots,B_t$ are disjoint, by definition the sets $E(B_1),E(B_2),\ldots,E(B_t)$ are also disjoint, and therefore for any $i \in [t]$, every hyperedge $e \in E(B_i)$ is counted exactly once in ${\bf S}$. 
	Next, consider the case when $e \in \partial B_{j_1}, \partial B_{j_2}, \ldots,\partial B_{j_L}$. By definition, for every $h \in [L]$, we must have $e \cap B_{j_h} \neq \emptyset$. Furthermore, since $B_{j_1},\ldots,B_{j_L}$ are disjoint, the intersections $e \cap B_{j_1},e\cap B_{j_2},\ldots,e \cap B_{j_L}$ are also disjoint. Therefore 
	$L \leq |e| \leq d$ i.e., the hyperedge can be counted at most $d$ times in ${\bf S}$. Therefore,
	\[
	(\epsilon/4)t \leq \sum_{i \in [t]} w(E(B_i) \cup \partial B_i) = w_{\rm mult}({\bf S}) \leq d w(E_H) = d M
	\]
	which on rearranging gives us the bound. 
	
\end{proof}

\subsection{Proof of Theorem \ref{thm:separators}}
Towards proving Theorem \ref{thm:separators}, we shall use the following result on 
single scale embeddings for negative type metric spaces.
\begin{theorem}[Theorem 3.1~\cite{ALN08}]			\label{thm:embed}
	There exists constant $C'>0$ and $p \in (0,1)$ such that for every $N$-point $\ell^2_2$-metric space $(V',d_{\ell^2_2})$ and $\Delta > 0$, the following holds. There exists a distribution $\mu$ over subsets $U \subseteq V'$ such that for every $x,y \in V'$ with $d(x,y) \geq \Delta/16$ we have 
	\[
	\Pr_{U \sim \mu}\left[y \in U \mbox{ and } d_{\ell^2_2}(x,U) 
	\geq \frac{C'\Delta}{\sqrt{\log N}}\right] \geq p.
	\]
\end{theorem}

\begin{proof}[Proof of Theorem \ref{thm:separators}]
	Let $C > 0$ be the constant from Theorem \ref{thm:embed}. Let $V' \subset V$ be the  $\epsilon$-net of $V$ in the $d_{\ell^2_2}$-metric of size at most $16Md\sqrt{\log (d M/\Delta)}/(C\Delta)$ (with $\epsilon = C\Delta/(4'\sqrt{\log (d M/\Delta)})$ as guaranteed by Lemma \ref{lem:greedy-ball}). Note that from our choice of $\epsilon$ we have 
	\[
	\frac{C\Delta}{\sqrt{\log |V'|}}  \geq \frac{C\Delta}{\sqrt{\log (16 Md/C\Delta (\log Md/\Delta))}} \geq 0.9 \frac{C\Delta}{\sqrt{\log (dM/\Delta)}} 
	\]
	for $dM$ large enough. Therefore instantiating Theorem \ref{thm:embed} with $V'$ and $\Delta$, we get that there exists a distribution $\mu$ over subsets of $V'$ such that the following holds. For any  $x,y \in V'$ satisfy $d_{\ell^2_2}(x,y) \geq \Delta/16$ we have
	\begin{equation}				\label{eqn:sep-bound}
	\Pr_{U \sim \mu}\left[y \in U \mbox{ and } d_{\ell^2_2}(x,U) 
	\geq 0.9\frac{C\Delta}{\sqrt{\log (dm/\Delta)}}\right] 
	\geq \Pr_{U \sim \mu}\left[y \in U \mbox{ and } d_{\ell^2_2}(x,U) 
	\geq \frac{C\Delta}{\sqrt{\log |V'|}}\right] \geq p.
	\end{equation}
	Now for every $x \in V$, we identify $c(x) \in V'$ such that $d_{\ell^2_2}(x,c(x)) \leq \epsilon$. Let $\mu$ be the measure guaranteed be Theorem \ref{thm:embed}. Fix $x,y \in V$ such that $d_{\ell^2_2}(x,y) \geq 2\Delta/16$. Then it follows that $d_{\ell^2_2}(c(x),c(y))\geq 2\Delta/16 - 2\epsilon \geq \Delta/16$. Therefore from Eq. \ref{eqn:sep-bound} we have 
	\begin{align*}
		& \Pr_{U \subseteq_\mu V'}\left[c(y) \in U \mbox{ and } d_{\ell^2_2}(c(x),U) \geq 0.9\frac{C\Delta}{\sqrt{\log (dM/\Delta)}}\right] \geq p \\
		&\Rightarrow \Pr_{U \subseteq_\mu V'}\left[y \in U^{+\epsilon} \mbox{ and } d_{\ell^2_2}(x,U^{+\epsilon}) \geq 0.9\frac{C\Delta}{\sqrt{\log ( d M/\Delta)}} - 2\epsilon \right] \geq p \\
		&\Rightarrow \Pr_{U \subseteq_\mu V'}\left[y \in U^{+\epsilon} \mbox{ and } d_{\ell^2_2}(x,U^{\epsilon}) \geq \frac{C\Delta}{2\sqrt{\log (Md/\Delta)}} \right] \geq p \\
	\end{align*}
	where $U^{+\epsilon} = \left\{x \in V | d_{\ell^2_2}(x,U) \leq \epsilon \right\}$ is the $\epsilon$-blowup of the set $U$.
\end{proof}

\paragraph{Proof of Theorem \ref{thm:hos-3}}

The proof of the theorem is identical to the construction of $\ell^2_2$-orthogonal separators constructed in \cite{lm16} (Theorem 2.2). The key observation here is that the $\sqrt{\log n}$-term in the construction of $\ell^2_2$-orthogonal separator comes directly from $\Omega(1/\sqrt{\log n})$-separation guaranteed by Theorem \ref{thm:embed}. Therefore, instantiating the proof of Theorem 2.2 from \cite{lm16} with the improved arity dependent bounds from Theorem \ref{thm:separators} and then proceeding as is gives us the guarantees stated by Theorem \ref{thm:hos-3}.

\section{Miscellaneous Proofs}

\subsection{Eigenfunction Facts}

\begin{proposition}			\label{prop:eig-facts}
	Let $A \in \mathbbm{R}^{n \times n}$ be the adjacency matrix of a weighted graph $G = (V,E,w)$ such that $\sum_{ij}A_{ij} = 1$. Let $D$ be the diagonal matrix consisting of vertex weights and let $N = D^{-1}A$ be the corresponding row-normalized transition probability matrix. Let $\phi_1,\ldots,\phi_n \in \mathbbm{R}^n$ be the right eigenvectors of $N$ with eigenvalues $\lambda_1,\lambda_2,\ldots,\lambda_n$. Let $\mu:V \mapsto [0,1]$ be the stationary measure for the random walk on the graph $G$. Then the following identities hold: 
	\begin{itemize}
		\item For any $i \in [n]$, we have $\Ex_{x \sim \mu} \Ex_{y \sim \mu(x)} \phi_i(x) \phi_i(y) = \lambda_i$.
		\item For every $i,j \in [n]$ such that $i \neq j$ we have $\Ex_{x \sim \mu} \Ex_{y \sim \mu(x)} \phi_i(x) \phi_j(y) = 0$.
		\item For any $i \in [n]$, we have $\Ex_{x \sim \mu} \Ex_{y,y' \sim \mu(x)} \phi_i(y) \phi_i(y') = \lambda^2_i$.
		\item For any $i,j \in [n]$, such that $i \neq j$ we have $\Ex_{x \sim \mu} \Ex_{y,y' \sim \mu(x)} \phi_i(y) \phi_j(y') = 0$.
	\end{itemize} 
\end{proposition}

\begin{proof}
	For every vertex $x \in V$, let ${\rm deg}_A(x)$ denote the weighted degree of the vertex $x$ in $A$. Note that by definition, for every $i \in [n]$ we have $N \phi_i = \lambda_i \phi_i$. For the first point we can see that 
	\begin{eqnarray*}
	\Ex_{x \sim \mu} \Ex_{y \sim \mu(x)} \Big[\phi_i(x) \phi_i(y)\Big]
	&=& \sum_{x \in V} {\rm deg}_A(x) \sum_{y \in V} \frac{A(x,y)}{{\rm deg}_A(x)} \phi_i(x)\phi_i(y) \\
	&=& \sum_{x,y \in V} A(x,y)\phi_i(x)\phi_i(y) \\
	&=& \phi^\top_i A \phi_i \\
	&=& \phi^\top_i D (D^{-1} A) \phi_i \\
	&=& \phi^\top_i D N \phi_i \\
	&=& \lambda_i \phi^\top_i D \phi_i = \lambda_i \langle \phi_i , \phi_i \rangle_M = \lambda_i
	\end{eqnarray*}
	where the last equality follows from the fact that $\phi_i$'s form a orthonormal basis with respect to $\mu$. For the second point, using arguments identical to above we have 
	\[
	\Ex_{x \sim \mu} \Ex_{y \sim \mu(x)}\Big[ \phi_i(x) \phi_j(y) \Big]= \phi^\top_i A \phi_j = \lambda_j \langle \phi_i,\phi_j \rangle_M = 0
	\]
	using the orthonormality of the $\phi_i$'s. Finally for the third point we observe that 
	\begin{eqnarray*}
	\Ex_{x \sim \mu} \Ex_{y,y' \sim \mu(x)} \Big[ \phi_i(y) \phi_i(y') \Big]
	&=& \Ex_{x \sim \mu} \left(\Ex_{y,\sim \mu(x)} \Big[ \phi_i(y) \Big]\right)^2 \\
	&=& \Ex_{x \sim \mu} \left(\sum_{y} \mu(y|x) \phi_i(y)\right)^2 \\
	&=& \Ex_{x \sim \mu} \left((N \phi_i)(x)\right)^2 \\
	&=& (N \phi_i)^\top D N \phi_i \\
	&=& \lambda^2_i \phi^\top_i D \phi_i = \lambda^2_i \langle \phi_i, \phi_i \rangle_M =  \lambda^2_i
	\end{eqnarray*}
	For the last point, using identical arguments as above, we can show that
	\[
		\Ex_{x \sim \mu} \Ex_{y,y' \sim \mu(x)} \Big[ \phi_i(y) \phi_i(y') \Big] = \lambda_i\lambda_j \phi^\top_i D \phi_j = \lambda_i \lambda_j \langle \phi_i, \phi_j\rangle_M = 0
	\] 
\end{proof}

\subsection{Spectral Gap of $M$} \label{sec:spec-gap}

\begin{lemma}
	Let $M$ be the Markov chain on the the vertex set $V_M$ as defined in Section \ref{sec:gadget}. Then for any $r \geq 1$, the spectral gap of the Markov chain $M^r$ is at least $\min(\epsilon/10g, \alpha\epsilon/24)$.
\end{lemma}
\begin{proof}
	We first claim that it suffices to show the above for $t= 1$. Indeed, let $1 \geq \lambda_2 \ldots \geq \lambda_{|V_M|}$ be the eigenvalues of $M$. Then for any $t \geq 1$,  the eigenvalues of $M^t$ are given by the multi-set $\{\lambda_S\}_{S \in [|V_M|]^r}$ where $\lambda_S = \prod_{j \in S} \lambda_j$ then 
	\[
	1 - \max_{S \neq (1,1,\ldots,1)} |\lambda_S| \leq 1 - \max_{j \neq 1} |\lambda_j|
	\]
	Now, we explicitly list the transition probabilities for the Markov chain $M$ defined in Section \ref{sec:gadget}.
	\begin{itemize}
		\item For every $i \in [k]$, we have $p(s_i|s_i) = 1 -\epsilon$ and $p(t_i|s_i) = \epsilon$.
		\item For every $i \in [k]$, we have $p(t_i|t_i) = 1/2$ and $p(s_i|t_i) = 1/(2(g + 1))$. Moreover, for every $t_j \in N_{G'}(t_i)$, $p(t_i|t_j) = 1/(2(g+1))$
	\end{itemize}
	Let $\lambda_2$ denote the second largest eigenvalue of the Markov chain. Recall the variational form of the spectral gap
	\[
	1 - \lambda_2 = \inf_{f : f \neq {\rm const}} \frac{\cE(f,f)}{{\rm Var}(f)}
	\]
	where $\cE(f,f) = \sum_{u,v \in V} \mu(u)p(v|u) (f(u) - f(v))^2$ and ${\rm Var}(f) = \sum_{u,v \in V} \mu(u)\mu(v) (f(u) - f(v))^2$. 
	For convenience, define $\epsilon',D$ to be such that $1/(1 + (2g+1)\epsilon) = 1 - \epsilon'$, and $D = 2(g+1)$. Then from Claim \ref{cl:M-prop} we know that $\mu(s_i) = (1 - \epsilon')/k$ and $\mu(t_i) = \epsilon'/k$. Fix an arbitrary vector $f$. We analyze the numerator and denominator. 
	\begin{align}
		\cE(f,f) &= 	\sum_{i \in [k]} (f(s_i) - f(t_i))^2 \mu(s_i)p(t_i|s_i) + \sum_{i \in [k]} (f(s_i) - f(t_i))^2 \mu(t_i)p(s_i|t_i) \nonumber \\ 
		& \qquad + \sum_{i \in [k]}\mu(t_i)\sum_{t_j \in N_{G'}(t_i)} p(t_j|t_i)(f(t_i) - f(t_j))^2 \nonumber \\
		&= \sum_{i \in [k]} (f(s_i) - f(t_i))^2 \left(\frac{1-\epsilon'}{k}\cdot \epsilon + \frac{\epsilon'}{k}\cdot\frac{1}{D}\right) + \sum_{i \in [k]}\frac{ \epsilon'}{k}\sum_{t_j \in N_{G'}(t_i)} \frac{1}{D}(f(t_i) - f(t_j))^2 \nonumber \\
		&\geq \frac{\min\{\epsilon,\epsilon'\}}{Dk} \sum_{i \in [k]} (f(s_i) - f(t_i))^2 + 
		\frac{\epsilon'}{k} \sum_{i \in [k]}\sum_{t_j \in N_E(t_i)} \frac{1}{D} (f(t_i) - f(t_j))^2 \label{eq:Hnum}.
	\end{align}
	Note that in the above, the lazy walk terms do not appear, since they are supported on identical vertices. 	On the other hand, we also have 
	\begin{align*}
		{\rm Var}(f) &= \sum_{i,j \in [k]} (f(s_i) - f(s_j))^2 \mu(s_i)\mu(s_j) + 2 \sum_{i,j \in [k]} (f(s_i) - f(t_j))^2 \mu(s_i)\mu(t_j) \\
		& \qquad + \sum_{i,j \in [k]} (f(t_i) - f(t_j))^2 \mu(t_i)\mu(t_f) \\
		&= \frac{(1 - \epsilon')^2}{k^2}\sum_{i,j \in [k]} (f(s_i) - f(s_j))^2  + \frac{2 \epsilon'(1 - \epsilon')}{k^2}\sum_{i,j \in [k]} (f(s_i) - f(t_j))^2 \\ 
		& \qquad + \frac{(\epsilon')^2}{k^2}\sum_{i,j \in [k]} (f(t_i) - f(t_j))^2 . 
	\end{align*}
	Furthermore, we can bound 
	\begin{align*}
		\sum_{i,j} (f(s_i) - f(s_j))^2 &= \sum_{i,j \in [k]} (f(s_i) - f(t_i) + f(t_i) - f(t_j) + f(t_j) - f(s_j))^2 \\
		&\leq 3\sum_{i,j \in [k]} (f(s_i) - f(t_i))^2 + (f(t_i) - f(t_j))^2 + (f(t_j) - f(s_j))^2 \\
		&\leq 6k\sum_{i\in [k]} (f(s_i) - f(t_i))^2 + 3\sum_{i,j \in [k]}(f(t_i) - f(t_j))^2. 
	\end{align*}
	Again we can bound 
	\begin{eqnarray*}
		\sum_{i,j \in [k]} (f(s_i) - f(t_j))^2 &=& \sum_{i,j \in [k]} (f(s_i) - f(t_i) + f(t_i) - f(t_j))^2 \\
		&\leq& 2\sum_{i,j \in [k]} (f(s_i) - f(t_i))^2 + 2\sum_{i, j \in [k]}(f(t_i) - f(t_j))^2 \\
		&=& 2k\sum_{i} (f(s_i) - f(t_i))^2 + 2\sum_{i, j \in [k]}(f(t_i) - f(t_j))^2.
	\end{eqnarray*}
	Therefore, overall we can bound the denominator as 
	\begin{equation}
		\label{eq:Hdenom}
		{\rm Var}(f,f) \leq \frac{10}{k} \sum_{i \in [k]} (f(s_i) - f(t_i))^2 + \frac{8}{k^2}\sum_{i,j \in [k]} (f(t_i) - f(t_j))^2 . 
	\end{equation}
	Observe that 
	\[
	\frac{\frac{\min\{\epsilon,\epsilon'\}}{Dk}\sum_{i \in [k]}(f(t_i) - f(s_i))^2}{(10/k)\sum_{i \in [k]} (f(s_i) - f(t_i))^2} \geq \frac{\min\{\epsilon,\epsilon'\}}{10D}
	\]
	and 
	\[
	\frac{(\epsilon'/k) \sum_{i \in [k]} \sum_{t_j \in N_{G'}(t_i)} \frac{1}{D} (f(t_i) - f(t_j))^2} {(8/k^2)\sum_{i,j \in [k]} (f(t_i) - f(t_j))^2 }
	= \frac{\epsilon' g}{8D}\frac{\frac1k\sum_{i \in [k]} \sum_{t_j \in N_{G'}(t_i)} \frac{1}{g} (f(t_i) - f(t_j))^2} {\frac{1}{k^2}\sum_{i,j \in [k]} (f(t_i) - f(t_j))^2 }
	\geq \frac{\alpha\epsilon'}{24} \]
	since $D = 2(g+1) \leq 3g$, and the last inequality follows from the spectral gap of the Markov chain $G'$ (defined in the Section \ref{sec:gadget}) on the $t_i$-vertices. Now using equations \ref{eq:Hnum} and \ref{eq:Hdenom} and 
	the fact that $(A_1 + A_2)/(B_1 + B_2) \geq \min(A_1/B_1,A_2/B_2)$, we get
	\begin{align*}
		\frac{\cE(f,f)}{{\rm Var}(f)} & \geq 
		\frac{\frac{\min\{\epsilon,\epsilon'\}}{Dk} \sum_{i \in [k]}\sum_{t_j \in N_{G'}(t_i)} \frac{1}{D} (f(t_i) - f(t_j))^2 + \frac{\epsilon'}{k} \sum_{i \in [k]}\sum_{t_j \in N_E(t_i)} \frac{1}{D} (f(t_i) - f(t_j))^2 }
		{ \frac{10}{k} \sum_{i \in [k]} (f(s_i) - f(t_i))^2 + \frac{8}{k^2}\sum_{i,j \in [k]} (f(t_i) - f(t_j))^2 } \\
		& \geq \min \left\{ \frac{\frac{\min\{\epsilon,\epsilon'\}}{Dk}\sum_{i \in [k]}(f(t_i) - f(s_i))^2}{(10/k)\sum_{i \in [k]} (f(s_i) - f(t_i))^2}, \frac{(\epsilon'/k) \sum_{i \in [k]} \sum_{t_j \in N_{G'}(t_i)} \frac{1}{D} (f(t_i) - f(t_j))^2} {(8/k^2)\sum_{i,j \in [k]} (f(t_i) - f(t_j))^2 }\right\} \\
		& \geq \min\left\{\frac{\min\{\epsilon,\epsilon'\}}{10D}, \frac{\alpha\epsilon'}{24} \right\}
	\end{align*}
	Finally observe that 
	\[
	1 - \epsilon' = \frac{1}{1 + 2(g+1)\epsilon} \leq 1 - (g+1)\epsilon
	\]
	for small enough $\epsilon$. Therefore $\min\{\epsilon,\epsilon'\} \geq \epsilon$, which gives us the desired bound.
	
\end{proof}

\subsection{Proof of Lemma \ref{lem:inf-dec}}			\label{sec:inf-dec}
The proof of this follows along standard influence decoding steps~\cite{KKMO07,St10Thesis}. The overall quantitative bounds obtained are different simply because of the different choice of the underlying Markov operator. Recall that from the setting of the lemma we have
\[
\Pr_{v \sim V_{\cG}}\left[ \max_{i \in [n]} \Inf{i}{\Gamma_{1 - \eta}g_v} \geq \tau\right] \geq \gamma. 
\]
Let $V_{\rm large} \subset V$ be the subset of vertices with large averaged influences i.e.,
\[
V_{\rm large}: = \left\{v \in V_{\cG} |  \max_{i \in [n]} \Inf{i}{\Gamma_{1 - \eta}g_v} \geq \tau \right\}
\]
Furthermore, it is well known that for functions $f:V^n_M \mapsto \mathbbm{R}$,  the mapping $f \mapsto \Inf{i}{f}$ is convex. This follows immediately using the fact that the influence can be expressed as $\sum_{S: \chi_{S_i} \neq {\bf 1}} \hat{f}(S)^2$, and each mapping $f \mapsto \hat{f}(S)^2$ are convex in $f$ (which is thought of as defined by a vector of Fourier coefficients). Now fix a vertex  $v \in V_{\rm large}$. Then by the definition of $V_{\rm large}$, there exists $i \in [n]$ such that $\Inf{i}{\Gamma_{1 - \eta} g_v} \geq \tau$. Furthermore, since $\Inf{i}{g_v}  = \sum_{r \in [k]} \Inf{i}{g^r_v}$, there exists $ r \in [k]$ such that $\Inf{i}{g^r_v} \geq \tau/k$. Using Jensen's inequality and the linearity of the noise operator $\Gamma_{1 - \eta}$ we have 
\[
\frac{\tau}{k} \leq \Inf{i}{\Gamma_{1 - \eta} g^r_v}  
 = \Inf{i}{\Gamma_{1 - \eta}\Big(\Ex_{w \sim N_{\cG(v)}}  \pi_{w \to v} \circ f^r_w\Big)} 
 \leq  \Ex_{w \sim N_{\cG(v)}} \Inf{\pi^{-1}_{w \to v}(i)}{\Gamma_{1 - \eta} f^r_w} \].
and therefore by averaging we have 
\begin{equation}
\Pr_{w \sim N_{\cG}(v)} \left[\Inf{\pi^{-1}_{w \to v}(i)}{\Gamma_{1 - \eta} f^r_w} \geq \frac{\tau}{2k}\right] \geq \frac{\tau}{2k}.
\end{equation}
Call the set of neighbors satisfying the above event as $N^*_{\cG}(v)$. Note that the above conclusion holds for any $v \in V_{\rm large}$. Now, for every vertex $v \in V$, and $r \in [k]$ define the following sets
\[
L_{v,r} := \left\{i \in [n] \Big| \Inf{i}{\Gamma_{1 - \eta} f^r_v} \right\} \geq \frac{\tau}{2k} \ \ \ \ \textnormal{and} \ \ \ \ 
\tilde{L}_{v,r} := \left\{i \in [n] \Big| \Inf{i}{\Gamma_{1 - \eta} g^r_v} \geq \frac{\tau}{2k} \right\}.
\]
Using Markov's inequality and the bound on the sum of noisy influences (Lemma \ref{lem:four-decay}), we have $|L_{v,r}|,|\tilde{L}_{v,r}| \leq 2k/(\tau \epsilon \eta)$. Finally, we sample a labeling $\sigma: V_{\cG} \to [k]$ using the following randomized decoding

\begin{enumerate}
	\item For every vertex $v \in V_{\cG}$ do the following.
	\item Sample $r \sim [k]$ uniformly at random.
	\item With probability $1/2$, if $L_{v,r}$ is non-empty, assign $\sigma(v)$ uniformly from $L_{v,r}$. Otherwise assign $\sigma(v)$ arbitrarily.
	\item With probability $1/2$, if $\tilde{L}_{v,r}$ is non-empty, assign $\sigma(v)$ uniformly from $\tilde{L}_{v,r}$. Otherwise assign $\sigma(v)$ arbitrarily.
\end{enumerate}

Then we can bound the expected fraction of constraints in $\cG$, satisfied by $\sigma$ as 
\begin{align*}
&\Ex_\sigma\Ex_{v \sim V_{\cG}}\Ex_{w \sim N_{\cG}(v)} \left[\mathbbm{1}\left(\pi_{w \to u}(\sigma(w)) = \sigma(v)\right)\right]  \\
&\geq \gamma \Ex_\sigma\Ex_{v \sim V_{\rm large}}\Ex_{w \sim N_{\cG}(v)} \left[\mathbbm{1}\left(\pi_{w \to u}(\sigma(w)) = \sigma(v)\right)\right] \\
&\geq \frac{\gamma}{2k} \Ex_{v \sim V_{\rm large}}\Ex_{w \sim N^*_{\cG}(v)} \Ex_\sigma \left[\mathbbm{1}\left(\pi_{w \to u}(\sigma(w)) = \sigma(v)\right)\right] \\
&\geq \frac{\gamma}{2k} \left( \frac{1}{k^2}\cdot \frac{1}{4} \cdot\left(\tau^2\epsilon^2\eta^2/4k^2\right) \right) \\
\end{align*}
where the last inequality follows from the fact that conditioned on $v \in V_{\rm large}$ and $w \in N^*_{\cG}(v)$, there exists a choice of $r \in [k]$, $i \in \tilde{L}_{r,v}$ and $j \in L_{r,w}$ such that $\pi_{w \to v}(j) = i$. Since the number of ways $r,i,j$ can be chosen (whenever $\tilde{L}_{v,r},L_{w,r} \neq \emptyset$) is bounded by $k^2(2k/\epsilon\eta\tau)^2$, the claim follows.

\end{document}